\documentclass{article}
\usepackage[a4paper]{geometry}
\newcommand{\E}{E}	
\newcommand{\R}{\mathbb{R}} 

\usepackage{placeins}
\usepackage{amsmath, amsthm, amssymb}
\usepackage{graphicx}
\usepackage[font=small]{caption}

\DeclareMathOperator*{\tr}{\mathrm{tr}}
\DeclareMathOperator*{\diag}{\mathrm{diag}}
\DeclareMathOperator*{\argmin}{arg\,min}
\DeclareMathOperator*{\argmax}{arg\,max}

\RequirePackage{etoolbox}
\usepackage{natbib}
\usepackage[colorlinks,citecolor=blue,urlcolor=blue,filecolor=blue,backref=page]{hyperref}
\usepackage{graphicx}
\usepackage[utf8]{inputenc}
\usepackage{enumerate}
\usepackage{verbatim}
\usepackage{booktabs}
\usepackage{algpseudocode}
\usepackage[usenames,dvipsnames]{xcolor}
\usepackage{color}
\usepackage{dsfont}
\usepackage{float}
\usepackage{authblk}
\usepackage[utf8]{inputenc}
\usepackage{enumerate}
\usepackage{verbatim}
\usepackage{graphicx}
\usepackage{algorithm}
\usepackage{algpseudocode}
\usepackage{listings}
\usepackage{color}
\usepackage{subcaption}
\usepackage{float}
\usepackage{appendix}

\usepackage{hyperref}

\newtheorem{thm}{Theorem}
\newtheorem{prop}{Proposition}
\newtheorem{lemma}{Lemma}
\newtheorem{cor}{Corollary}
\newtheorem{remark}{Remark}
\theoremstyle{plain}
\newtheorem{assumption}{Assumption}

\title{Introducing prior information in Weighted Likelihood Bootstrap with applications to model misspecification}
\date{}

\author{Emilia Pompe\thanks{email: emilia.pompe@stats.ox.ac.uk}}
\affil{\small{Department of Statistics, University of Oxford, United Kingdom}}

\begin{document}
	
	\maketitle
	
	\begin{abstract}
		We propose Posterior Bootstrap, a set of algorithms extending Weighted Likelihood Bootstrap, to properly incorporate prior information and address the problem of model misspecification in Bayesian inference. We consider two approaches to incorporating prior knowledge: the first is based on penalization of the Weighted Likelihood Bootstrap objective function, and the second uses pseudo-samples from the prior predictive distribution. We also propose methodology for hierarchical models, which was not previously known for methods based on Weighted Likelihood Bootstrap. Edgeworth expansions guide the development of our methodology and allow us to provide greater insight on properties of Weighted Likelihood Bootstrap than were previously known. Our experiments confirm the theoretical results and show a reduction in the impact of model misspecification against Bayesian inference in the misspecified setting. 
		
\end{abstract}

	\section{Introduction}\label{section:introduction}
	\subsection{Model misspecification and robust inference in Bayesian statistics}\label{section:model_misspecification}
	Bayesian inference relies on an assumption that the true data-generating mechanism is known. However, in practice the models only provide an approximate reflection of  the reality, so model misspecification is a frequent concern \citep{de2013bayesian,walker13,lyddon2019general}. Throughout the paper we write $x_1, \ldots, x_n \sim f(x|\theta)$ to denote that $x_1, \ldots, x_n$ are independent and identically distributed according to a distribution with density $f(x|\theta)$. We consider a generic parametric Bayesian model for parameter of interest $\theta\in \R^d$
	\begin{align}\label{eq:Bayesian_model}
	\begin{split}
	x_1, \ldots, x_n  \sim f(x|\theta), \quad \theta \sim  \pi(\theta), \quad \theta \in \Theta, \end{split}
	\end{align}
	with prior $\pi$ and likelihood function $f$.
	
	Misspecification of model \eqref{eq:Bayesian_model}
	occurs when data is generated from a distribution $P^*$ with density $p^*$, which is not equal to $f(x|\theta)$ for any $\theta \in \Theta$.  In this case under regularity conditions the Bayesian posterior  asymptotically concentrates around the pseudo-true parameter $\theta^*$ minimizing $D_{KL}(p^*||f)$, the Kullback--Leibler divergence between $f(\cdot|\theta),$ $\theta \in \Theta,$ and $p^*$ \citep{berk1966limiting, berk1970consistency, bunke1998asymptotic}. 
	Applying standard Bayesian inference to misspecified models leads to misleading uncertainty estimates \citep{kleijn2012bernstein}, so there has been increasing interest in robust inference schemes \citep{grunwald2017inconsistency, jewson2018principles, hong20}.
	
	One method that has gained a lot of popularity replaces the likelihood $f(x|\theta)$ by a power likelihood $f(x|\theta)^{\eta}$ for a parameter ${\eta}>0$ \citep{jiang2008gibbs, grunwald2012safe, bissiri2016general, grunwald2017inconsistency}. The interpretation of $\eta$ is that it represents the learning rate, or the analyst's trust in the model. For $\eta<1$ the data have less influence than in the standard Bayesian inference, while using $\eta>1$  places more trust in the likelihood.
	
Resampling schemes are another tool for addressing Bayesian model misspecification \citep{waddell2002very, douady2003comparison}. The BayesBag algorithm \citep{buhlmann2014discussion, huggins2019using, huggins2020robust} combines samples from posteriors on bootstrapped data sets. Asymptotic theory for BayesBag was proved by \cite{huggins2019using}. 

Another group of robust inference methods is based on Weighted Likelihood Bootstrap \citep{newton1991weighted, newton1994approximate,lyddon2019general}, and this is also a starting point for our approach. 

	In this paper we develop Posterior Bootstrap, a set of scalable algorithms for performing principled inference in Bayesian models under possible model misspecification.  We extend Weighted Likelihood Bootstrap to incorporate prior information. Several  approaches to this task have been recently proposed in the literature \citep{ lyddon2018nonparametric, fong2019scalable, newton2018weighted}, to our knowledge however, this is the first attempt to analyse theoretically the impact of various ways of incorporating the prior.  
	Edgeworth expansions show that prior beliefs about the parameter of interest are properly represented in the inference in a similar way to standard Bayesian inference, and guide our choice of hyperparameters of the algorithms. Moreover, we propose a robust inference scheme for hierarchical models, which has not been done so far for methods based on Weighted Likelihood Bootstrap or power posteriors.

	\subsection{Properties of Weighted Likelihood Bootstrap}\label{section:reminders_wlb}
	Weighted Likelihood Bootstrap was originally proposed as a method for sampling approximately from a posterior distribution. Samples are drawn by independently optimizing the weighted log likelihood
	\begin{equation}\label{eq:standard_wlb_formula}
	\argmax_{\theta} \left\{\sum_{i=1}^{n} w_i \log f(x_{i}| \theta)\right\}, \quad w_1, \ldots, w_n \sim\text{Exp}(1).
	\end{equation}
	Combined with an adjustment by the Sequential Importance Resampling algorithm \citep{gordon1993novel, doucet2000sequential}, which could correct for the fact that Weighted Likelihood Bootstrap is approximate, it was an alternative to then emerging Markov chain Monte Carlo methods. A key advantage of this algorithm is that it allows to perform computations fully in parallel, therefore a recent resurgence of interest in this method  \citep{newton2018weighted, lyddon2019general, ng2020random, nie2020bayesian} can partly be explained by developments in parallel computing. Even though Weighted Likelihood Bootstrap appeared in the Bayesian literature, it does not incorporate a prior. A similar idea, called perturbing the minimand, was developed independently in the frequentist literature \citep{jin2001simple}. The analogy between these two methods was noticed by \cite{parzen2007perturbing}. Some of the frequentist approaches consider adding a lasso-type regularization term \citep{minnier2011perturbation, das2019perturbation}, which from the Bayesian point of view could be interpreted as adding a prior.  
	
	Asymptotic results for the Weighted Likelihood Bootstrap were proved by \cite{newton1991weighted} under a correctly specified model and extended by \cite{lyddon2019general} to the misspecified case. Let the maximum likelihood estimator $\hat{\theta}$ be $\argmax_{\theta \in \Theta}\sum_{i=1}^n f(x_i|\theta)$, and let $\tilde{\theta}$ be a draw from the Weighted Likelihood Bootstrap algorithm based on observations $x_1, \ldots, x_n$. Theorem~1 of \cite{lyddon2019general} states that for any measurable set $A$ under standard regularity conditions 
	\begin{equation}\label{eq:bvm_wlb}
	\text{pr} \left\{ n^{1/2} \left( \tilde{\theta} - \hat{\theta} \right) \in A  \big| x_{1:n} \right\} \to \text{pr} \left(Z \in A\right), \quad   Z \sim N\left\{0, J(\theta^*)^{-1} I(\theta^*) J(\theta^*)^{-1} \right\},
	\end{equation}
	where
	\begin{equation}\label{eq:I_J_matrices}
	I(\theta) = E_{p^*} \left[\left\{\nabla_{\theta} \log f(X| \theta)\right\} \left\{\nabla_{\theta} \log f(X| \theta)\right\}^T\right], \quad J(\theta) = -  E_{p^*}\left[ D_{\theta}^2\left\{  \log f(X| \theta)\right\}\right],
	\end{equation}
	and $\theta^*$ is the unique pseudo-true parameter minimizing $D_{KL}(p^*||f)$. The asymptotic covariance matrix in \eqref{eq:bvm_wlb} is the so-called sandwich covariance matrix known from the frequentist literature on model misspecification \citep{huber1967behavior} as an asymptotic covariance matrix of the maximum likelihood estimator under a potentially misspecified model. 
	
	We let $\Pi(\cdot| x_{1:n})$ be the Bayesian posterior distribution given observations $x_1, \ldots, x_n$ under model~\eqref{eq:Bayesian_model}, and by $\pi(\cdot|x_{1:n})$ we denote its density. Recall that the Bernstein--von Mises theorem for misspecified Bayesian models \citep{kleijn2012bernstein} states that
	\begin{equation}\label{eq:bvm_misspec}
	\Pi \left\{n^{1/2} \left( \theta - \hat{\theta} \right) \in A \big| x_{1:n}\right\} \to \text{pr} \left(Z' \in A\right),\quad Z' \sim N\left\{0, J(\theta^*)^{-1} \right\}.
	\end{equation}
	Convergence in \eqref{eq:bvm_misspec} holds in total variation, while \cite{lyddon2019general} considered convergence  in distribution. We show however in Corollary \ref{prop:convergence_of_densities} that convergence in  \eqref{eq:bvm_wlb} also holds in total variation. Note that under both standard Bayesian inference and Weighted Likelihood Bootstrap the asymptotic distribution is concentrated around the same mean, but with potentially different covariance matrices.   Comparing \eqref{eq:bvm_wlb} and \eqref{eq:bvm_misspec}, we conclude that for misspecified models under Weighted Likelihood Bootstrap credible sets are valid confidence sets in the frequentist sense, which is not the case for standard Bayesian posteriors \citep{kleijn2012bernstein}.  For correctly specified models, the covariance matrices in \eqref{eq:bvm_wlb} and \eqref{eq:bvm_misspec} are equal.

	Results obtained by  \cite{fushiki2005bootstrap, fushiki2010bayesian} explain why  Weighted Likelihood Bootstrap corrects for model misspecification in comparison with standard Bayesian inference. In particular, these two articles analyse the risk function
	\begin{equation}\label{eq:risk_function}
	\mathbb{E}_{x_{1:n}} \left[D_{KL}\left\{p^*||\hat{p}(\cdot|x_{1:n})\right\} \right] = \int p^*(x_{1:n}) \left\{\int p^*(x_{n+1}) \log \frac{p^*(x_{n+1})}{\hat{p}(x_{n+1}|x_{1:n})} \mathrm{d}x_{n+1} \right\} \mathrm{d}x_{1:n},
	\end{equation}
	where $\hat{p}(\cdot|x_{1:n})$ is the predictive distribution given data $x_{1:n}$ under a given method of inference. Proposition \ref{proposition:fushiki} below is a direct corollary of Theorem~1 of \cite{fushiki2010bayesian} and Theorems 1 and 3 of \cite{fushiki2005bootstrap}, see also \cite{lyddon2018nonparametric}. 
	\begin{prop}\label{proposition:fushiki}
		Consider model \eqref{eq:Bayesian_model} and assume that  $\theta^* \in \Theta$ minimizing $D_{KL}(p^*||f)$ is unique. Suppose that $I(\theta^*)$  and $J(\theta^*)$ are positive definite, and that $I(\theta^*) \neq J(\theta^*)$. Then the predictive distribution based on Weighted Likelihood Bootstrap asymptotically provides better prediction than the Bayesian predictive distribution, that is, the value of the empirical risk \eqref{eq:risk_function} is asymptotically smaller.
	\end{prop}
	In view of the robust inference methods discussed in Section~\ref{section:model_misspecification}, a natural question to ask is whether the result stated in Proposition \ref{proposition:fushiki} can be extended from standard Bayesian inference to power posteriors.  Theorem~\ref{thm:risk_wlb_power_posteriors} below shows that the answer is negative since under certain conditions power posteriors yield an asymptotically  smaller risk than Weighted Likelihood Bootstrap.
	\begin{thm}\label{thm:risk_wlb_power_posteriors}
		Consider model \eqref{eq:Bayesian_model} and let $\lambda_1, \ldots, \lambda_d$  denote the eigenvalues of $J(\theta^*)^{-1/2}I(\theta^*)J(\theta^*)^{-1/2}$, where $\theta^* \in \Theta$ is the unique pseudo-true parameter. Assume that we perform inference on $\theta$ using a power posterior $\pi_{\eta}$ with some $\eta>0$. Suppose that the eigenvalues satisfy either 
		$1 < \lambda_{1},\ldots,\lambda_d < \eta^{-1},
		$
		or $\eta^{-1} <\lambda_{1},\ldots,\lambda_d <1.
		$
		Then the risk given by \eqref{eq:risk_function} associated with the predictive distribution under Weighted Likelihood Bootstrap is larger than the risk associated with the predictive distribution under power posterior $\pi_{\eta}$.
	\end{thm}
	For prediction problems, Theorem~\ref{thm:risk_wlb_power_posteriors} shows that in certain scenarios power posteriors with an appropriate choice of power $\eta$ outperform Weighted Likelihood Bootstrap. We present the proof in  Supplementary Material \ref{supplementary:risk_proof}.
	
	\subsection{Edgeworth expansions}\label{section:edgeworth_expansions}
	In order to develop appropriate methodology for incorporating prior information into Weighted Likelihood Bootstrap, we first need to specify what tools can be used for measuring the impact of the prior on the draws $\tilde{\theta}$ from the resulting algorithm.  To this end, we examine the way the prior influences the posterior in standard Bayesian inference. Recall that by the Bernstein--von Mises theorem  the prior asymptotically does not influence the first order approximation, which necessitates considering higher order approximations to the posterior.
	
	In classical asymptotic theory of parametric inference the Edgeworth expansion of a density function is a well-established tool of analysing higher order approximations. For details we refer the reader to \cite{bhattacharya1978validity, ghosh1994higher, hall2013bootstrap}. Note that the Gaussian approximation appearing in the central limit theorem  is associated with an absolute error approximation of the order $O\left(n^{-1/2}\right)$, and the Edgeworth expansion should be seen as a natural extension of this approximation, up to the order $O\left(n^{-k/2}\right)$ for an arbitrary natural number $k$. The coefficients of this expansion are expressed in terms of  cumulants of the density of interest.
	
	Letting $\ell(x, \theta) = \log f(x|\theta)$,  we use
	\begin{equation*}
	I_n=  \frac{1}{n}\sum_{j=1}^n  \nabla_{\theta}\ell(x_j,\hat{\theta})\nabla_{\theta}\ell(x_j,\hat{\theta})^T, \quad  J_{n}  =  -\frac{1}{n}\sum_{j=1}^n  D_{\theta}^2 \ell(x_j,\hat{\theta})
	\end{equation*}
	as empirical estimates of $I(\theta^*)$ and $J(\theta^*)$ respectively. Let  $ \phi $ denote the density of a $ d $-variate standard normal distribution. Recall that in case of  well-specified models   the  Laplace approximation of the posterior density of  $I_n^{1/2}n^{1/2} \left(\theta - \hat{\theta} \right)$  is
	\begin{equation}\label{eq:edgeworth_expansion_Bayesian_posterior}
	\phi(y) \left[1 + I(\theta^{*})^{-1/2}\frac{ \left\{ \nabla_{\theta} \log \pi(\theta^*) \right\}^T y}{n^{1/2}} +  I (\theta^*)^{-3/2} \frac{\sum_{p,q,r=1}^d  c_{pqr}y_p y_q y_r}{6n^{1/2}}  \right]  + o\left(n^{-1/2}\right)
	\end{equation} 
	for $y = (y_1, \ldots, y_d) \in \R^d$, and $c_{pqr}$ is $E \left\{\partial^3 \ell(\theta,x)/\partial \theta_p \partial \theta_q \partial \theta_r\right\}$ evaluated at $\theta^*$. 
	Since in Bayesian posteriors the impact of the prior appears in the second order Edgeworth expansion, we will keep this as our primary requirement. That is, we demand that the prior information appears in the second order Edgeworth expansion obtained for Posterior Bootstrap. Ideally, for well-specified models  our method would incorporate the prior knowledge in the same way it is done in Bayesian inference, that is, through the same prior-dependent expression in the second order expansion. The reason for this is that in case of well-specified models the Bayesian posterior offers optimal information about the parameter given the data and the prior distribution \citep{bernardo2009bayesian}. Alternatively, as explained in \cite{bissiri2016general}, the posterior $\Pi(\theta|x_{1:n})$ is the only probability measure  minimizing the expected  loss under the negative log likelihood loss function, which is an appropriate loss function to be used in the well-specified case, and hence provides a valid and coherent update of beliefs about~$\theta$. 
	
	On the other hand, in case of misspecified models the first order expansions of Posterior Bootstrap and of the Bayesian posterior are different, as seen in  \eqref{eq:bvm_wlb} and \eqref{eq:bvm_misspec}, which is what allows Posterior Bootstrap to correct for model misspecification.

	\section{Posterior Bootstrap via penalization on the weighted log-likelihood}\label{section:pb_penalisation}
	\subsection{Methodology for non-hierarchical models}
	We now present Algorithm~\ref{alg:pb_general}, our Posterior Bootstrap method via prior penalization for model \eqref{eq:Bayesian_model}. The algorithm has only one parameter $w_0$ representing the strength of the 
	penalization imposed on the weighted log likelihood with the prior distribution. For $w_0= 0$, Algorithm~\ref{alg:pb_general} recovers  Weighted Likelihood Bootstrap.
	\begin{algorithm}[!ht]
		\caption{Posterior Bootstrap via prior penalization, returns samples from $P^{B,p}(\cdot|x_{1:n})$}\label{alg:pb_general}
		\begin{algorithmic}[1]	
			\For {j = 1, \ldots, N} \Comment{\textbf(in parallel)}
			\State{Draw $w_1^{(j)}, \ldots, w_{n}^{(j)} \sim\text{Exp}(1)$.}
			\State{Set $\tilde{\theta}^{(j)} = \argmax_{\theta} \left\{\sum_{i=1}^{n} w_i \log f(x_{i}| \theta) + w_0^T \log \pi(\theta)\right\}$.}
			\EndFor
			\State \Return {Samples $\big\{\tilde{\theta}^{(j)}\big\}_{j=1}^N$. }
		\end{algorithmic}
	\end{algorithm}
	
	In the basic version of the algorithm we assume that $w_0$ is a non-negative real number. In this case Algorithm~\ref{alg:pb_general} resembles the Weighted Bayesian Bootstrap \citep{newton2018weighted}. Instead of a fixed parameter $w_0$, \cite{newton2018weighted}  use a random variable $\lambda \text{Exp}(1)$ for a regularization constant $\lambda$, typically selected via cross-validation. If we additionally assume that the prior $\pi(\theta)$ factorizes with respect to all coordinates of $\theta = (\theta_1, \ldots, \theta_d)$ as 
	\begin{equation}\label{eq:prior_assumption}
	\pi(\theta) = \prod_{k=1}^d \pi_k(\theta_k),
	\end{equation} 
	we then treat $w_0$ as a $d$-dimensional vector $(w_{0,1}, \ldots, w_{0,d})$ with non-negative entries, rather than a real value. For simplicity of notation the penalization term $w_0^T \log \pi(\theta)$ in of Algorithm~\ref{alg:pb_general} denotes 
	$
	w_0^T \log \pi(\theta) = \sum_{k=1}^d w_{0,k} \log \pi_k(\theta_k).$ 
	We defer a discussion about appropriate choice of $w_0$ to Section~\ref{section:choice_w_0}. We let $P^{B,p}$ be the distribution of $\tilde{\theta}$ drawn from Algorithm~\ref{alg:pb_general} and by $p^{B,p}$ we denote its density. 
	\begin{remark}
		Algorithm~\ref{alg:pb_general} can be extended to the setting of \citep{bissiri2016general, lyddon2019general}, where the negative log likelihood $-\log f(x|\theta)$ is replaced with an arbitrary loss function. The asymptotic results will then be analogous to those stated in Section~\ref{section:theoretical_pb_penalization} under the same regularity conditions. 
	\end{remark}
	
	\FloatBarrier
	\subsection{Asymptotic results for Algorithm~\ref{alg:pb_general}}\label{section:theoretical_pb_penalization}
	In this section we provide the Edgeworth expansion of the output of Algorithm~\ref{alg:pb_general}, which generalizes the results obtained by \cite{lyddon2019general} to the penalized version, and strengthens them to the second order expansion.
	
	In what follows, if there is ambiguity, we write $E_z$ to denote expectation with respect to $z$. We consider the following set of assumptions.
	\begin{assumption}\label{assump:log_lik} \normalfont {(Log likelihood)}. The log likelihood function $\ell(x, \theta) = \log f(x|\theta) $ is  measurable and bounded from above, with $E_{p^*} \big|\ell(X, \theta) \big| < \infty$ for all $\theta \in \Theta$.
	\end{assumption}
	
	\begin{assumption}\label{assump:identifiability}\normalfont (Identifiability). There  exists a unique maximizing parameter value $$
		\theta^* = \argmax_{\theta \in \Theta} \E_{p^*} \left\{\ell( X, \theta)\right\},
		$$ and further for all $\delta > 0$ there exists an $\epsilon >0$ such that 
		\begin{equation*}
		\lim_{n \to \infty} P^* \left[ \sup_{\|\theta - \theta^*\| \geq \delta} \frac{1}{n} \sum_{i=1}^n \left\{ \ell ( X_i,\theta) - \ell ( X_i,\theta^*)\right\} \leq - \epsilon \right] =1,
		\end{equation*} 
		where $P^*$ is the data-generating distribution.
	\end{assumption}
	\begin{assumption}
		\normalfont\label{assump:log_lik_smoothness} (Smoothness of the log likelihood function). There is an open ball $B$ containing $\theta^*$ such that  $\ell(\theta, x)$ is three times continuously differentiable with respect to $\theta \in B$ almost surely on $x$ with respect to $P^*$. We additionally impose certain moment conditions on the partial derivatives. For details see Supplementary Material \ref{supplementary:summary_assumption}. 
	\end{assumption}
	\begin{assumption}\label{assump:positive_definite}
		\normalfont (Positive definiteness of information matrices and linear independence of partial derivatives). For $B$ like in Assumption \ref{assump:log_lik_smoothness}, the matrices $I(\theta)$ and $J(\theta)$ defined in \eqref{eq:I_J_matrices} are positive definite for $\theta \in B$ with all elements finite. Besides, the first and second partial derivatives
		$$
		\left\{\frac{\partial \ell(x, \theta)}{\partial \theta_p}, \frac{\partial^2 \ell(x,  \theta)}{\partial \theta_r \partial \theta_s} \right\}, \quad 1 \leq p \leq d, \, 1 \leq r \leq s \leq d
		$$
		are linearly independent as functions of $x$ at $\theta^*$.
	\end{assumption}
	
	\begin{assumption}\label{assump:weights}
		\normalfont (Distribution of weights).  The weights $w_{1:n} = (w_1, \ldots, w_n) $ are independent and identically distributed according to a continuous distribution $W$ having all moments and such that $\E(W) =1$ and $\text{var}(W) =1$.
	\end{assumption}

	Assumptions \ref{assump:log_lik} -- \ref{assump:log_lik_smoothness}, as well as positive definiteness of information matrices in Assumption \ref{assump:positive_definite} are direct analogues of regularity conditions considered in \cite{lyddon2019general}. Linear independence of partial derivatives expressed in Assumption \ref{assump:positive_definite} is a common assumption for obtaining Edgeworth expansions of the maximum likelihood estimator. Assumptions \ref{assump:log_lik} and \ref{assump:identifiability} ensure that as $n \to \infty$, with probability going to~1 distribution $P^{B,p}(\cdot| x_{1:n})$ concentrates around $\theta^*$. Assumptions \ref{assump:log_lik_smoothness} and \ref{assump:positive_definite} are used to show the validity of the Edgeworth expansion. Since the algorithms we present in the paper rely on the Weighted Likelihood Bootstrap formulation \eqref{eq:standard_wlb_formula}, by default the  weights follow the exponential distribution $\text{Exp}(1)$, thus satisfying Assumption \ref{assump:weights}, we allow however for other distributions of weights. We will use the above conditions throughout the paper to prove our theoretical results for various versions of Posterior Bootstrap.
	
	In the case of Algorithm~\ref{alg:pb_general}, we require the following smoothness assumption on $\pi(\theta)$ in the neighbourhood of $\theta^*$. 
	
	\begin{assumption}\label{assump:prior_smoothness}\normalfont {(Smoothness of log prior density).} The function $\log \pi(\theta)$ is  measurable, upper-bounded on $\Theta$, and three times continuously differentiable at $\theta^*$.
	\end{assumption}
	
	We denote by $w_{0} \cdot \nabla_{\theta} \log \pi(\theta^*)$ a vector with $ p $th entry $w_{0,p}{\pi'(\theta^*_p)}/{\pi(\theta^*_p)}$ if $w_0\in \R^d$ and $w_{0} \pi'(\theta^*_p)/\pi(\theta^*_p)$ when $w_0\in R$.
	When $\theta \in \R$, we adopt the following notation:
	\begin{align}\label{eq:notation_univariate}
	\mu_{3} &=  \E_{p^*} \left\{  D_{\theta}^3\ell(X,{\theta}^*)\right\},  
	&
	A_3  &=  \E_{p^*} \left\{\nabla_{\theta}\ell(X,{\theta}^{*})\right\}^3,
	&
	L_{1,2} &=  \E_{p^*}  \left\{D_{\theta}^2 \ell(X,{\theta}^{*}) \nabla_{\theta}\ell(X,{\theta}^{*})\right\}.
	\end{align}
	We denote by $\gamma= \E \left\{(W -1)^3\right\}$ the third central moment of $ W $, so $\gamma =2$ when $ W \sim \text{Exp}(1)$.
	
	\begin{thm}\label{thm:edgeworth_expansion_penalisation}
		Consider model \eqref{eq:Bayesian_model} with a parameter of interest $\theta \in \R^d$ and suppose Assumptions \ref{assump:log_lik} -- \ref{assump:prior_smoothness} hold. As $n \to \infty$, with probability going to 1 the Edgeworth expansion of the density of $I_n^{-1/2} J_{n} n^{1/2}(\tilde{\theta} - \hat{\theta})$  for $\tilde{\theta}$ drawn from Algorithm~\ref{alg:pb_general} is 
		\begin{equation}\label{eq:edgeworth_expansion_general_pb_misspec_multivariate}
		\phi(y) \left[1 + I(\theta^{*})^{-1/2}\frac{ \left\{ w_{0} \cdot \nabla_{\theta}\log \pi(\theta^*)\right\}^T y}{n^{1/2}} + \frac{P_{3,n}(y)}{n^{1/2}} \right] + o \left(n^{-1/2}\right),
		\end{equation}
		where $P_{3,n}(y)$ is an order three polynomial independent from the prior.
		When $d=1$, formula \eqref{eq:edgeworth_expansion_general_pb_misspec_multivariate} has the following more explicit form 
		\begin{align}\label{eq:edgeworth_expansion_general_pb_misspec}
		\phi(y) \left\{1 + \kappa_{1,n}y + \kappa_{3,n} \frac{y^3 - 3y}{6} \right\}+ o \left(n^{-1/2}\right),
		\end{align}
		where
		\begin{align*}
		\kappa_{1,n} & = \frac{ w_0\nabla_{\theta}\log \pi(\theta^*)}{n^{1/2} I(\theta^{*})^{1/2}} + \frac{\mu_{3}}{2 n^{1/2}}\frac{I(\theta^{*})^{1/2}}{J(\theta^{*})^2} + \frac{ L_{1,2}}{n^{1/2}I(\theta^{*})^{1/2}J(\theta^{*})}, \\
		\kappa_{3,n} & = \frac{\gamma A_{3}}{n^{1/2}I(\theta^{*})^{3/2}} + \frac{3\mu_{3}}{ n^{1/2}}\frac{I(\theta^{*})^{1/2}}{J(\theta^{*})^2} + \frac{ 6L_{1,2}}{n^{1/2}I(\theta^{*})^{1/2}J(\theta^{*})}. 
		\end{align*}
		Additionally, if the model is well-specified and $\gamma =2$, formula  \eqref{eq:edgeworth_expansion_general_pb_misspec} becomes
		\begin{align}\label{eq:edgeworth_expansion_general_pb_well_spec}
		\phi(y) \left[1 +  \frac{ w_0  \nabla_{\theta}\log \pi(\theta^*) y}{n^{1/2} I(\theta^{*})^{1/2}} -   \frac{\left(A_3+ \mu_3\right) y}{3n^{1/2}I(\theta^{*})^{3/2}}   
		+  \frac{\mu_{3}y^3  }{6n^{1/2} I(\theta^{*})^{3/2}}  \right]+ o \left(n^{-1/2}\right).
		\end{align}
	\end{thm} 
	The sketch of the proof of  Theorem~\ref{thm:edgeworth_expansion_penalisation} is as follows. We first expand $\sqrt{n}\left(\tilde{\theta} - \hat{\theta}\right)$ as a function of $\sum_{i=1}^d b(x_i, \hat{\theta})$, where $b(x_i, \hat{\theta})$ is a vector containing derivatives of $\ell(x_i, \theta)$ with respect to $\theta$. We then obtain the Edgeworth expansion of $\sum_{i=1}^d b(x_i, \hat{\theta})$ as a random variable in $w_{1:n}$. 	Results from Chapter 2 of \cite{ghosh1994higher} are used to justify that \eqref{eq:edgeworth_expansion_general_pb_misspec} is indeed a valid Edgeworth expansion, that is, that the incurred error is of order $o\left(n^{-1/2} \right)$. For the full proof with all calculations we refer the reader to the Supplementary Material \ref{supplementary:edgeworth_proof}. We also provide an Edgeworth expansion for the related method by \cite{newton2018weighted} in Supplementary Material \ref{supplementary:other_methods}.
	
	To derive formula \eqref{eq:edgeworth_expansion_general_pb_well_spec} from \eqref{eq:edgeworth_expansion_general_pb_misspec}, recall that for well-specified models we have 
	$I(\theta^*) = J(\theta^*)$ and by one of the Bartlett identities \citep{bartlett1953approximate, bartlett1953approximate2}
	\begin{equation}\label{eq:bartlett_identity}
	-L_{1,2} = \frac{\mu_3 + A_3}{3}.
	\end{equation}
	If additionally $\gamma =2$, we can simplify the formulas for $\kappa_{3,n}$ and $\kappa_{1,n}$ to
	\begin{align*}
	\kappa_{3,n} = & \frac{2 A_{3}}{n^{1/2}I(\theta^{*})^{3/2}} + \frac{6\mu_{3}}{2 n^{1/2} I(\theta^{*})^{3/2}}- \frac{2 \mu_3 + 2 A_3} {n^{1/2}I(\theta^{*})^{3/2}} = \frac{\mu_{3}}{n^{1/2} I(\theta^{*})^{3/2}},\\
	\kappa_{1,n} =& \frac{ w_0\nabla_{\theta}\log \pi(\theta^*)}{n^{1/2} I(\theta^{*})^{1/2}} + \frac{\mu_{3}}{6 n^{1/2} I(\theta^{*})^{3/2}} - \frac{A_3}{3n^{1/2}I(\theta^{*})^{3/2}}.
	\end{align*}

	 For well-specified models, it is insightful to compare the Edgeworth expansions obtained for Posterior Bootstrap \eqref{eq:edgeworth_expansion_general_pb_misspec_multivariate} with the Laplace approximations for Bayesian posteriors \eqref{eq:edgeworth_expansion_Bayesian_posterior}. Note that for $w_0 = 1$ the term corresponding to the prior, that is $n^{-1/2} I(\theta^{*})^{-1/2}\nabla_{\theta}\log \pi(\theta^*)$ appears in the same way in both formulas. Firstly, this suggests using $w_0=1$ if the analyst believes the model is well-specified, which should in fact be intuitive given that  for $w_0 =1$ 
	\begin{align*}
	\E_{w_{1:n}} \left\{\sum_{i=1}^{n} w_i \log f(x_{i}| \theta) + w_0 \log \pi(\theta)\right\} = \log \pi(\theta| x_{1:n}) + \text{c},
	\end{align*} 
	where the term $c$ does not depend on $\theta$. This also has implications when the analyst wants to choose between two prior distributions $\pi_1(\theta)$ and $\pi_2(\theta)$. For $w_0 = 0$ we then have $$
	\pi_1(\theta|x_{1:n}) - \pi_2(\theta|x_{1:n}) = p^{B,p}_{\pi_1}(\theta|x_{1:n}) - p^{B,p}_{\pi_2}(\theta|x_{1:n}) + o \left(n^{-1/2}\right),
	$$ 
	where $p^{B,p}_{\pi}$  denotes a density $p^{B,p}$  under prior $\pi$. This facilitates understanding the impact of the prior and suggests making decisions about the choice of the prior following the guidelines developed for Bayesian inference. 
	
	At the same time it should be noted that in general the Posterior Bootstrap method does not yield second order correctness to Bayesian posteriors for well specified models. This is due to the additional term $-n^{-1/2}I(\theta^{*})^{-3/2}\left(A_3+ \mu_3\right)/3$ in formula  \eqref{eq:edgeworth_expansion_general_pb_well_spec}, which is only equal to 0 in special cases. Further details about the consequences of this result are presented in Proposition \ref{prop:jeffreys_prior}.
	
	\begin{prop}\label{prop:jeffreys_prior}
		Consider model \eqref{eq:Bayesian_model} with a one-dimensional parameter $\theta$. Suppose $\tilde{\theta}$ is drawn from Weighted Likelihood Bootstrap, that is, from Algorithm~\ref{alg:pb_general} with $w_0 =0$. If the model is well-specified, then the Edgeworth expansion of  $I(\hat{\theta})^{1/2} n^{1/2}(\tilde{\theta} - \hat{\theta})$ is second order correct to the corresponding expansion for Bayesian posterior utilizing Jeffreys prior on $\theta$.
	\end{prop}
	\begin{proof}
		Suppose that $\pi(\theta)$ is a Jeffreys prior, that is $\pi(\theta) \propto \left\{I(\theta)\right\}^{1/2}$.
		Then 
		\begin{align*}
		\nabla_{\theta} \log \pi(\theta^*)  = \frac{I(\theta^*)^{-1/2}I(\theta^*)'}{2 I(\theta^*)^{1/2}} = \frac{I(\theta^*)'}{2 I(\theta^*)} = \frac{2 L_{1,2}(\theta^*)}{2 I(\theta^*)} = -\frac{A_3 + \mu_3}{3I(\theta^*)^{3/2}},
		\end{align*}
		where the last equality follows from the Bartlett identity \eqref{eq:bartlett_identity}. Comparing formula \eqref{eq:edgeworth_expansion_general_pb_well_spec} for $w_0 =0$ and \eqref{eq:edgeworth_expansion_Bayesian_posterior} completes the proof.
	\end{proof}
	The above result has an interesting and at the same time intuitive interpretation from the point of view of properties of Jeffreys priors.  Recall that Jeffreys priors are designed to be non-informative and represent ignorance about the parameter of interest. Weighted Likelihood Bootstrap can be characterized by an analogous property, as it does not incorporate any prior information about the parameter. In fact the question of second order equivalence between Weighted Likelihood Bootstrap and Bayesian posteriors with Jeffreys priors was initially raised in Section~4 of \cite{newton1994approximate}. Due to an omission in the Taylor expansion, the authors obtained an incomplete Edgeworth expansion, which lead to incorrect conclusions about second order correctness.
		
	We present below Corollary \ref{prop:convergence_of_densities} which can be directly inferred from the existence of Edgeworth expansions.
	\begin{cor}\label{prop:convergence_of_densities}
		Consider model \eqref{eq:Bayesian_model} and suppose $\tilde{\theta}$ is drawn from Algorithm~\ref{alg:pb_general}. Then, under conditions of Theorem~\ref{thm:edgeworth_expansion_penalisation}, $n^{1/2}(\tilde{\theta} - \hat{\theta})$ converges to the normal distribution $N\left\{0, J(\theta^*)^{-1} I(\theta^*) J(\theta^*)^{-1} \right\}$ in total variation for any choice of $w_0$.
	\end{cor}
	Corollary \ref{prop:convergence_of_densities} applied to $w_0 = 0$ generalizes Theorem~1 of \cite{lyddon2019general}, and the convergence expressed in \eqref{eq:bvm_misspec} can be stated in total variation, as it is typically done for the Bernstein--von Mises type results.
	\subsection{Choice of $w_0$ in Algorithm~\ref{alg:pb_general}}\label{section:choice_w_0}
	We start by providing some intuition why, if the user simply sets $w_0 =1$ under misspecification, their prior beliefs on the parameter may not be appropriately reflected in the inference. To this end, we consider a one-dimensional toy model 
	\begin{align}
	\begin{split}\label{model:toy_normal_normal_univariate}
	x_1, \ldots, x_n  \sim N(\theta, 1), \quad \theta \sim \text N(\mu_0, \sigma_0^2).
	\end{split}
	\end{align}
	Suppose model \eqref{model:toy_normal_normal_univariate} is misspecified and in fact $x_1, \ldots, x_n \sim N(0,\sigma^2)$ for $\sigma \neq 1$. The correct model would then be
	\begin{align}
	\begin{split}\label{model:toy_normal_normal_univariate_true}
	x_1, \ldots, x_n  \sim N(\theta, \sigma^2), \quad \theta \sim \text N(\mu_0, \sigma_0^2).
	\end{split}
	\end{align}
	Drawing a sample from Algorithm~\ref{alg:pb_general}  is equivalent to computing
	\begin{align}\label{eq:normal_model_equivalence}
	\begin{split}
	\argmax_{\theta} \left\{\sum_{i=1}^{n}  \frac{-w_i(x_i - \theta)^2}{2} -  \frac{w_0(\theta- \mu_0)^2}{2\sigma_0^2} \right\}
	=  \argmax_{\theta} \left\{\sum_{i=1}^{n}  \frac{-w_i(x_i - \theta)^2}{2 \sigma^2} -  \frac{w_0(\theta- \mu_0)^2}{2\sigma_0^2 \sigma^2} \right\}.
	\end{split}
	\end{align} 
	If the user knew the true data generating mechanism they would instead compute
	\begin{align}\label{eq:normal_model_equivalence_2}
	\argmax_{\theta}  \left\{\sum_{i=1}^{n} -\frac{w_i (x_i - \theta)^2}{2 \sigma^2} -  \frac{(\theta- \mu_0)^2}{2\sigma_0^2} \right\}.
	\end{align}
	Equation \eqref{eq:normal_model_equivalence} is equivalent to \eqref{eq:normal_model_equivalence_2} if $w_0 = \sigma^2$. In fact drawing samples according to \eqref{eq:normal_model_equivalence} with $w_0 = 1$ would be  equivalent to using the correct likelihood but having the prior belief encoded by $\theta \sim N\left(\mu_0, \sigma^2 \sigma_0^2\right)$ rather than $\theta \sim N\left(\mu_0,  \sigma_0^2\right)$.
	
	We now make a connection between the above observation and the Edgeworth expansion obtained in Theorem~\ref{thm:edgeworth_expansion_penalisation}, which will further guide our choice of $w_0$.. Let $J(\theta^*)$ and $I(\theta^*)$ be the information matrices for model \eqref{model:toy_normal_normal_univariate} at $\theta^*$, that is $J(\theta^*) =1$ and $I(\theta^*) = \sigma^2$. Let $I_{\text{true}}(\theta^*)$ be the information matrix at $\theta^*$ under the correct model \eqref{model:toy_normal_normal_univariate_true}. We thus have $I_{\text{true}}(\theta^*) = \sigma^{-2}$. In the second order expansion of the Bayesian posterior of $\sigma^{-1}n^{1/2} \left(\theta - \hat{\theta} \right)$ under the correct model, as shown in \eqref{eq:edgeworth_expansion_Bayesian_posterior}, the prior-dependent term is 
	$$
	n^{-1/2}I_{\text{true}}(\theta^*)^{-1/2} \pi'(\theta^*)/\pi(\theta^*) = n^{-1/2}\sigma \pi'(\theta^*)/\pi(\theta^*). 
	$$
	On the other hand, since $I(\theta^*)^{-1/2} J(\theta^*) = \sigma^{-1}$, by Theorem \ref{thm:edgeworth_expansion_penalisation} we get that under model \eqref{model:toy_normal_normal_univariate} and Algorithm~\ref{alg:pb_general} with $w_0 = \sigma^2$, the prior-dependent term in the second order Edgeworth expansion of $\sigma^{-1}n^{1/2} \left(\tilde\theta - \hat{\theta} \right)$  is
	$$
n^{-1/2}	I(\theta^*)^{-1/2}w_0 \pi'(\theta^*)/\pi(\theta^*) = n^{-1/2}\sigma \pi'(\theta^*)/\pi(\theta^*) .
	$$
	To summarize, in this special case the prior-dependent term in the Edgeworth expansion of Posterior Bootstrap under the misspecified model is the same as the prior-dependent term in Bayesian inference under the correct model.
	
	The example described above as well as the fact that Posterior Bootstrap corrects the covariance matrix of the asymptotic distribution  motivates the following methodology of setting $w_0$, for potentially multivariate models. Suppose the prior $\pi$ satisfies \eqref{eq:prior_assumption} so that $w_0$ can be treated as a $d$-dimensional vector. The goal is to match the prior-dependent terms in Edgeworth expansion for Algorithm~\ref{alg:pb_general} and in the Laplace approximation for the Bayesian posterior with a corrected scaling $I(\theta^* )^{-1/2} J(\theta^*)$, rather than $I(\theta^*)^{1/2}$. Thus we want to set $w_0$ so that for all $y \in \R^d$ it satisfies
	\begin{align*}%
	I(\theta^{*})^{-1/2}
	\begin{pmatrix} 
	w_{0,1} \frac{\pi'(\theta^*_1)}{\pi(\theta^*_1)}  y_1 \\
	\vdots \\
	w_{0,d} \frac{\pi'(\theta^*_d)}{\pi(\theta^*_d)} y_d
	\end{pmatrix}  = \left\{ I(\theta^* )^{-1/2} J(\theta^*) \right\}^{-1}
	\begin{pmatrix} 
	\frac{\pi'(\theta^*_1)}{\pi(\theta^*_1)}  y_1 \\
	\vdots \\
	\frac{\pi'(\theta^*_d)}{\pi(\theta^*_d)} y_d
	\end{pmatrix},
	\end{align*}
	which is equivalent to 
	\begin{align}\label{eq:try_setting_w_0}
	D(w_0) \begin{pmatrix}
	\frac{\pi'(\theta^*_1)}{\pi(\theta^*_1)}  y_1 \\
	\vdots \\
	\frac{\pi'(\theta^*_d)}{\pi(\theta^*_d)} y_d
	\end{pmatrix}  = I(\theta^{*})^{1/2}  J(\theta^*)^{-1} I(\theta^* )^{1/2}
	\begin{pmatrix} 
	\frac{\pi'(\theta^*_1)}{\pi(\theta^*_1)}  y_1 \\
	\vdots \\
	\frac{\pi'(\theta^*_d)}{\pi(\theta^*_d)} y_d
	\end{pmatrix},
	\end{align}
	where $D(w_0)$ is a diagonal matrix with entries $w_{0,1}, \ldots, w_{0,d}$ on the diagonal.
	Note that for well-specified models formula  \eqref{eq:try_setting_w_0} above would imply that $w_0 =(1, \ldots, 1)$, as expected. When $I(\theta^{*})^{1/2}  J(\theta^*)^{-1} I(\theta^* )^{1/2}$ is not diagonal, it is impossible to set $w_0$ so that \eqref{eq:try_setting_w_0} is satisfied for all $y \in \R^d$. We therefore set $w_0$ as
	\begin{equation*}
	w_0^* = \argmin_{w_0} \| D(w_0) - I(\theta^{*})^{1/2}  J(\theta^*)^{-1} I(\theta^* )^{1/2} \|_F = \diag \left\{I(\theta^{*})^{1/2}  J(\theta^*)^{-1} I(\theta^* )^{1/2}  \right\},
	\end{equation*}
	where $\| \cdot \|_F$ denotes the Frobenius norm. In practice, we approximate this formula with
	\begin{equation}\label{eq:setting_w_0}
	w_0^* = \diag \left(I_n^{1/2}  J_n^{-1} I_n^{1/2} \right),
	\end{equation}
	which ensures that all elements of $w_0$ are non-negative as required.
	If the prior does not factorize, then we must constrain the diagonal elements of $D(w_0)$ in \eqref{eq:try_setting_w_0} to be equal, so instead we propose to set $w_0$ as
	\begin{equation}\label{eq:setting_bar_w_0}
	\bar{w}_0^* = \frac{1}{d}\tr \left(I_n^{1/2}  J_n^{-1} I_n^{1/2} \right).
	\end{equation}
	In the case of factorizing priors, we advocate using formula \eqref{eq:setting_w_0} rather than \eqref{eq:setting_bar_w_0} as it allows for capturing different kinds of misspecification, such as overdispersion or underdispersion, in different coordinates. If parameter $\theta$ is one-dimensional, the formulae for $w_0^*$ and $\bar{w}_0^*$ coincide.
	
	An alternative strategy would be performing cross-validation, which could be particularly useful for prediction problems. The above approach, however, offers an automatic and principled way of setting $w_0$ so that the prior information is appropriately incorporated in the inference. This method has also a natural extension to hierarchical models as we demonstrate in Section~\ref{section:hierarchical_models}.
	\section{Extending Posterior Bootstrap to hierarchical models}\label{section:hierarchical_models}
	\subsection{Fixed number of groups $K$}
	We consider the following potentially misspecified hierarchical model. 
	\begin{gather}
	\begin{aligned}\label{model:hierarchical}
x_{k,1},\ldots x_{k,n_k} & \sim f_k(x_{k,i}| \theta_k),&   k = 1, \ldots, K,\\
	\theta_k  & \sim g_k(\theta_k|\lambda),  & k = 1, \ldots, K,\\
	\lambda &\sim p(\lambda).
	\end{aligned} 
	\end{gather}
	Each group $k$ has a sample size $n_k$ for $k=1, \ldots, K$. We define $x_{1:K, 1:n_k}$ as a vector of all observations $x_{k,j}$ for $j = 1, \ldots, n_k$ and $k = 1, \ldots, K$. The marginal prior on  $\theta  = (\theta_1, \ldots, \theta_K)$ is 
	\begin{equation}\label{eq:hierarchical_prior}
	\tilde{\pi}(\theta) = \int \left\{\prod_{k=1}^K g_k(\theta_k|\lambda) p(\lambda) \right\}\mathrm{d} \lambda.
	\end{equation}
	If the integral in \eqref{eq:hierarchical_prior} is tractable, then Algorithm~\ref{alg:pb_general} can be applied to model \eqref{model:hierarchical}, with the joint   prior $\tilde{\pi}(\theta)$ on $\theta$. 
	
	When this integral cannot be computed, we propose instead to use Algorithm~\ref{alg:pb_hierarchical_gibbs}, our Posterior Bootstrap method for hierarchical models. We let $P^{HB,p}(\cdot|x_{1:K, 1:n_k})$ denote the distribution of $(\tilde{\theta}_1, \ldots, \tilde{\theta}_K, \tilde{\lambda})$ under Algorithm~\ref{alg:pb_hierarchical_gibbs} and by $p^{HB,p}(\cdot|x_{1:K, 1:n_k})$ we denote its density. The prior information in Algorithm~\ref{alg:pb_hierarchical_gibbs} is introduced via penalization, in the spirit of Algorithm~\ref{alg:pb_general}. For each $k = 1, \ldots, K$ parameter $w_{0, [k]}$ plays for data $x_{k,1}, \ldots, x_{k, n_k}$ the role of $w_0$ used in Algorithm~\ref{alg:pb_general}. The fact that the penalization term depends on $\bar{\lambda}$, which is kept the same across groups, induces correlation between $\theta_1, \ldots, \theta_K$, in a similar way to what happens under standard Bayesian inference. 
	
	It is worth emphasizing here an additional advantage of using Algorithm~\ref{alg:pb_hierarchical_gibbs} for model \eqref{model:hierarchical}, even if the integral in \eqref{eq:hierarchical_prior} is tractable. Since the prior $\tilde{\pi}$ does not factorize, Algorithm~\ref{alg:pb_general} would need to be used with $w_0$ treated as a real number, whereas Algorithm~\ref{alg:pb_hierarchical_gibbs}  makes use of the conditional independence structure between $\theta_1, \ldots, \theta_K$ and allows the user to set $w_{0, [k]}$ separately for each group. Arguments outlined in Section~\ref{section:choice_w_0} show that this approach accounts better for potential misspecification of $f_1, \ldots, f_K$. 
	
	\begin{algorithm}[!ht]
		\caption{Posterior Bootstrap for hierarchical models via prior penalization, returns samples from $P^{HB,p}(\cdot|x_{1:K, 1:n_k})$}\label{alg:pb_hierarchical_gibbs}
		\begin{algorithmic}[1]	
			\For {k = 1, \ldots, K} 
			\State{ Set $\hat{\theta}_k = \argmax_{\theta_k}\sum_{i=1}^{n_k} \log f_k(x_{k,i}| \theta_k)$.}
			\EndFor
			\For {j = 1, \ldots, N} \Comment{\textbf(in parallel)  }
			\State{Sample $\bar{\lambda}^{(j)} \sim p\left(\lambda|\hat{\theta}_1, \ldots, \hat{\theta}_K\right)$.}
			\For {k = 1, \ldots, K} 
			\State{Draw $w_1^{(j)}, \ldots, w_{n_k}^{(j)} \sim\text{Exp}(1)$.}
			\State{ Set $\tilde{\theta}_k^{(j)} = \argmax_{\theta_k} \left\{ \sum_{i=1}^{n_k} w_i^{(j)} \log f_k(x_{k,i}| \theta_k) + w_{0,[k]}^T \log g_k\left(\theta_k|\bar{\lambda}^{(j)}\right) \right\}$.}
			\EndFor
			\State{Sample $\tilde{\lambda}^{(j)} \sim p\left(\lambda|\tilde{\theta}_1^{(j)}, \ldots, \tilde{\theta}_K^{(j)}\right)$.}
			\EndFor
			\State \Return {Samples $\Big\{ \left(\tilde{\lambda}^{(j)}, \tilde{\theta}_1^{(j)}, \ldots, \tilde{\theta}_K^{(j)}\right) \Big\}_{j=1}^N$. }
		\end{algorithmic}
	\end{algorithm}

We let $\theta_k^* = \argmax_{\theta}\E_{p_k^*}\left\{\log f_k(X|\theta) \right\}$, where $x_{k,1}, \ldots, x_{k, n_k}$ are generated from $P_k^*$. We treat here $\theta_1^*, \ldots, \theta_K^*$ as deterministic values, see Section~\ref{section:inference_lambda} for a scenario where $\theta_1^*, \ldots, \theta_K^*$ are themselves random variables. Firstly, as a direct consequence of Theorem~\ref{thm:edgeworth_expansion_penalisation}, we get the first order approximation for the asymptotic distribution of $\tilde{\theta}_k$. For any measurable set $A \subseteq \Theta$ we have
\begin{equation*}
P^{HB,p} \left\{ n^{1/2} \left( \tilde{\theta}_k - \hat{\theta}_k \right) \in A  \big| x_{1:K, 1:n_k} \right\} \to \text{pr} \left(Z \in A\right), \quad   Z \sim N\left\{0, J_k(\theta_k^*)^{-1} I_k(\theta_k^*) J_k(\theta_k^*)^{-1} \right\},
\end{equation*}
where $I_k$ and $J_k$ are given by \eqref{eq:I_J_matrices}, computed for $f_k$.
	
	We explain below why $\bar{\lambda}$ in Algorithm \ref{alg:pb_hierarchical_gibbs} is drawn from the posterior $p\left(\lambda| \hat{\theta}_1, \ldots, \hat{\theta}_K\right)$ rather then from the prior $p(\lambda)$, even though the latter may seem more intuitive. To this end, let us analyse the impact of the prior on the second order Edgeworth expansion on $\theta_k$. By construction of Algorithm~\ref{alg:pb_hierarchical_gibbs} and by Theorem~\ref{thm:edgeworth_expansion_penalisation}, if $n_1, \ldots, n_k$ are large, the prior-dependent term is
	\begin{equation}\label{eq:hierarchical_model_prior_impact}
I_k(\theta_k^{*})^{-1/2} n_k^{-1/2}	w_{0,[k]}^T \int \nabla_{\theta_k} \left\{\log g_k(\theta^*_k|\lambda) \right\} p(\lambda| \theta^*_1, \ldots,\theta^*_K) \mathrm{d} \lambda.
	\end{equation} On the other hand, since in fact the prior on $\theta$ has a form \eqref{eq:hierarchical_prior}, we would want the prior impact term to be given by $I_k(\theta_k^{*})^{-1/2} n_k^{-1/2}w_{0,[k]}^T \nabla_{\theta_k}\left\{ \log\tilde{\pi}(\theta^*)\right\}$. We observe that 
	\begin{align*}
	\nabla_{\theta_k} \log\tilde{\pi}(\theta) = & \frac{\nabla_{\theta_k}\left\{\int \prod_{j=1}^K g_k(\theta_j|\lambda) p(\lambda) \mathrm{d} \lambda\right\}}{\int \prod_{j=1}^K g_k(\theta_j|\lambda) p(\lambda) \mathrm{d} \lambda} \\
	=  & \frac{\int \left\{\nabla_{\theta_k}g_k(\theta_k|\lambda) \prod_{j=1, j \neq k}^K g_k(\theta_j|\lambda) p(\lambda) \right\} \mathrm{d} \lambda}{\int \prod_{j=1}^K g_k(\theta_j|\lambda) p(\lambda) \mathrm{d} \lambda}\\
	= & \int \left\{\frac{\nabla_{\theta_k} g_k(\theta_k|\lambda)}{g_k(\theta_k|\lambda)} \frac{\prod_{j=1}^Kg_k(\theta_j|\lambda) p(\lambda)}{\int \prod_{j=1}^K g_k(\theta_j|\lambda) p(\lambda) \mathrm{d} \lambda} \right\} \mathrm{d}\lambda \\
	= & \int \nabla_{\theta_k} \left\{\log g_k(\theta_k|\lambda) \right\} p(\lambda| \theta_1, \ldots, \theta_K) \mathrm{d} \lambda,
	\end{align*} 
	so $I_k(\theta_k^{*})^{-1/2} n_k^{-1/2} w_{0,[k]}^T \nabla_{\theta_k} \log\tilde{\pi}(\theta^*)$ is  indeed equal to \eqref{eq:hierarchical_model_prior_impact}. We therefore conclude  that if model \eqref{model:hierarchical} is well-specified, the impact of the prior plays asymptotically the same role as in case of Bayesian inference. Furthermore, this reasoning justifies setting  $w_{0,[1]}, \ldots, w_{0,[K]}$ using rules outlined in Section~\ref{section:choice_w_0} for $w_0$ in Algorithm~\ref{alg:pb_general}.
	
When presenting  Algorithm~\ref{alg:pb_hierarchical_gibbs}, we assumed implicitly that the user can sample from the posterior $p(\lambda|\theta_1, \ldots, \theta_K)$. In conjugate examples one can draw $\lambda$ directly, 
otherwise it would be necessary to employ a Markov chain Monte Carlo sampler for this step. We let $ n=\min_{k=1, \ldots, K} n_k$. Proposition~\ref{prop:wlb_lambda_posterior} demonstrates the asymptotic properties of $\tilde{\lambda}$ drawn from Algorithm~\ref{alg:pb_hierarchical_gibbs}, covering both the well-specified and misspecified cases. It shows that as $n \to \infty$, the distribution on $\tilde{\lambda}$ will tend to the same distribution as the Bayesian posterior on~$\lambda$.   A similar result was proved by \cite{petrone2014bayes} for marginal maximum likelihood estimators.

\begin{prop}\label{prop:wlb_lambda_posterior} Consider model \eqref{model:hierarchical} with a fixed number of groups $K$ and suppose $\tilde{\lambda}$ is drawn from Algorithm~\ref{alg:pb_hierarchical_gibbs}.  Suppose that Assumptions \ref{assump:log_lik} -- \ref{assump:positive_definite}  hold for each  likelihood function $f_k$  and that Assumption \ref{assump:prior_smoothness} is satisfied for $g_k$, for $k=1, \ldots, K$. Assume also that for each $\lambda$ the density $p(\lambda|\theta)$ is a continuous function of $\theta$ and that $p(\lambda|\theta)$ is jointly bounded as a function of $\theta$ and $\lambda$. Then as $n \to \infty$, we have the following convergence of densities
	\begin{equation}\label{eq:lambda_weak_convergence}
	p^{HB,p}\left(\tilde{\lambda}|x_{1:K, 1:n_k}\right) \to 	p\left(\tilde{\lambda}| \theta^*\right),
	\end{equation}
	with probability going to 1 with respect to $x_{1:K, 1:n_k}$.
\end{prop}
We present the proof in Supplementary Material \ref{supplement:proof_lambda_finite_K}.

	\subsection{Number of groups $K$ going to infinity}\label{section:inference_lambda}
We consider model \eqref{model:hierarchical} with $f_k(\cdot |\theta) =f(\cdot |\theta) $ and $g_k(\cdot|\lambda) = g(\cdot|\lambda)$ for all $\theta, \lambda$ and $k = 1, \ldots, K$.
An alternative to sampling $\tilde{\lambda}$ using Markov chain Monte Carlo would be performing another round of Posterior Bootstrap, where $\tilde{\theta}_1, \ldots, \tilde{\theta}_K$ is used as data. The caveat is, however, that since only asymptotic results are known for Posterior Bootstrap, for a fixed $K$ this strategy would lack a theoretical underpinning. We assume instead that the number of classes, now denoted by $K_n$, as well as the number of observations $n_k$ in each class, go to infinity.
The penalization strategy then yields 
	\begin{align}\label{eq:lambda_posterior_bootstrap}
	\begin{split}
	\tilde{\lambda} = & \argmax_{\lambda} \left\{ \sum_{k=1}^{K_n} v_k \log g(\tilde{\theta}_k| \lambda) + w_{g}^T \log p\left(\lambda\right) \right\}, \quad v_1, \ldots, v_{K_n} \sim \text{Exp}(1).
	\end{split}
	\end{align}
	The parameter $w_g$ is set according to methods described in Section~\ref{section:choice_w_0} for $w_0$, that is, formula \eqref{eq:setting_w_0} or \eqref{eq:setting_bar_w_0} for matrices $I_g$ and $J_g$, which we define explicitly in \eqref{eq:hierarchical_info_matrices}. We present the full resulting algorithm in Supplementary Material~\ref{supplement:proof_lambda_K_to_infty}, and we let $P^{HB,pp}\left(\cdot|x_{1:K, 1:n_k} \right)$ denote the distribution of $(\tilde{\theta}_1, \ldots, \tilde{\theta}_{K_n}, \tilde{\lambda})$ under this algorithm.

	We now describe the assumed data-generating mechanism of $x_{1:K_n, 1:n_k}$ under this scenario. Let $Q^*$ be a distribution on a space of distributions on $\R^s$ where $s$ is the dimension of a single observation $x_{k,i}$ for $k= 1, \ldots, K_n, i = 1, \ldots, n_k$. We have
	\begin{align}\label{eq:hierarchical_data_generating}
	\begin{split}
	x_{k,1},\ldots x_{k,n_k} &\sim P_k^*, \quad k= 1, \ldots, K_n,\\
	P_1^*, \ldots, P_{K_n} &\sim Q^*.
	\end{split}
	\end{align} 
	We use our standard notation for pseudo-true parameters $\theta_k^*$, that is, $\theta_k^* = \arg\max_{\theta_k} E_{p_k^*}\left\{\log f(x|\theta_k)\right\}$. Therefore, $\theta_1^*, \ldots, \theta_K^*$ are independent and identically distributed with respect to $Q^*$.  We additionally define $\lambda^* =  \arg\max_{\lambda} E_{Q^*}\left\{  \log g(\Theta^*|\lambda) \right\}$ and 
	$\hat{\lambda}_{\theta^*} =  \arg\max_{\lambda} \left\{ \sum_{k=1}^{K_n} \log g(\theta^*_k|\lambda) \right\}$.
	The information matrices for $\lambda$ are
	\begin{equation}\label{eq:hierarchical_info_matrices}
	I_g(\lambda) = E_{Q^*} \left[\left\{\nabla_{\lambda} \log g(\Theta^*| \lambda)\right\} \left\{\nabla_{\lambda} \log g(\Theta^*| \lambda)\right\}^T\right], \quad J_g(\lambda) = -  E_{Q^*}\left[ D_{\lambda}^2\left\{  \log g(\Theta^*| \lambda)\right\}\right].
	\end{equation}
	Theorem \ref{thm:hierarchical_large_K_result} shows that if $\tilde{\lambda}$ is drawn from \eqref{eq:lambda_posterior_bootstrap} and the rate at which  $K_n$ goes to infinity is slow enough compared with the rate at which $n\to \infty$, the distribution of $\tilde{\lambda}$ is asymptotically normal with the sandwich covariance matrix. 
	
		\begin{thm}\label{thm:hierarchical_large_K_result}
		Suppose that the data generating mechanism follows \eqref{eq:hierarchical_data_generating} and assume that  $K_n \to \infty$ as $n \to \infty$ at a rate such  $K_n = o\left(n^{1/2}/\log n\right)$, where $ n=\min_{k=1, \ldots, K_n} n_k$. Assume also that $\tilde{\lambda}$ is drawn from \eqref{eq:lambda_posterior_bootstrap}. Then under regularity conditions stated explicitly in Supplementary Material~\ref{supplement:proof_lambda_K_to_infty} we have
		\begin{equation*}
		P^{HB,pp}\left\{K_n^{1/2}\left( \tilde{\lambda} - \hat{\lambda}_{\theta^*} \right) \in A  \big|{x}_{1:K_n, 1:n_k} \right\} \to 	\textnormal{pr}  \left(Z \in A\right), \quad  Z \sim N\left\{0, J_g(\lambda^*)^{-1} I_g(\lambda^*) J_g(\lambda^*)^{-1} \right\},
		\end{equation*}
		with probability going to 1 with respect to ${x}_{1:K_n, 1:n_k}$ as $n \to \infty$, for matrices $I_g$ and $J_g$ defined in \eqref{eq:hierarchical_info_matrices}.
	\end{thm}
	The regularity conditions used in Theorem~\ref{thm:hierarchical_large_K_result} are analogues of Assumptions \ref{assump:log_lik} -- \ref{assump:prior_smoothness} for functions $ \log f(x|\theta)$ and $ \log g(\theta| \lambda)$. We additionally require $ \log g(\theta| \lambda)$ to be sufficiently smooth both with respect to $\lambda$ and $\theta$.
    The proof of Theorem~\ref{thm:hierarchical_large_K_result}  is reported in Supplementary Material~\ref{supplement:proof_lambda_K_to_infty}. 
	
Careful examination of the proofs of Proposition \ref{prop:wlb_lambda_posterior} and Theorem \ref{thm:hierarchical_large_K_result} shows that the same results would hold if we replaced $\tilde{\lambda}$ with $\bar{\lambda}$. If the parameter of interest is $\lambda$, we however advocate using $\tilde{\lambda}$ for inference, as it gives more accurate results. We illustrate the difference between $\bar{\lambda}$ and $\tilde{\lambda}$ on a synthetic dataset in Supplementary Material~\ref{supplement:dirichlet_allocation}. On the other hand, if only parameters $\theta_1, \ldots, \theta_K$ are of interest, the user can skip the step of drawing $\tilde{\lambda}$.

	To summarize, we have developed an algorithm for model \eqref{model:hierarchical} that demonstrates good asymptotic properties, both for inference on $\theta_1, \ldots, \theta_K$ and on  $\lambda$, while being robust to misspecification. Furthermore, even when the model is well-specified, the fact that Algorithm~\ref{alg:pb_hierarchical_gibbs} can be fully parallelized makes it an attractive alternative to  Metropolis-within-Gibbs samplers typically used for \eqref{model:hierarchical}.
	
	\section{Posterior Bootstrap via pseudo-samples from the prior predictive}\label{section:pb_pseudosamples}
	\subsection{Methodology}
	\label{subsection:pb_pseudosamples_methodology}
	Using Algorithm~\ref{alg:pb_general} is possible for most standard priors, however, for certain  distributions  combining them with Algorithm~\ref{alg:pb_general} would have adverse consequences for the inference, or could even be impossible from a computational point of view. Firstly, formulation of Algorithm~\ref{alg:pb_general} necessitates that evaluating $\log \pi(\theta)$ is possible for every $\theta$, which is not the case, for example for the original spike-and-slab prior \citep{mitchell1988bayesian}. Moreover, there is a widely used group of priors that are infinite on the boundary on the set, for example the Gamma distribution $\Gamma(\alpha,\beta)$ with $\alpha<1$, Dirichlet prior $\mathcal D(\alpha_1, \cdots, \alpha_k)$ with some of the $\alpha_j<1$ as advocated for instance by \cite{rousseau2011asymptotic} for mixture models. For these priors Algorithm~\ref{alg:pb_general} would always return the boundary of the set, regardless of the data. 
	
	To overcome these drawbacks, we consider another version of Posterior Bootstrap, proposed originally by \cite{fong2019scalable}. We summarize it in Algorithm~\ref{alg:pb_pseudosamples_prior}. The prior information is incorporated into the algorithms via drawing so-called pseudo-samples from the prior predictive. 
	\begin{algorithm}[!ht]
		\caption{Posterior Bootstrap via pseudo-samples from the prior predictive}\label{alg:pb_pseudosamples_prior}
		\begin{algorithmic}[1]	
			\For {j = 1, \ldots, N} \Comment{\textbf(in parallel)  }
			\State{Draw $\bar{\theta}_1^{(j)}, \ldots, \bar{\theta}_T^{(j)}\sim \pi(\theta)$.}
			\State{Generate prior pseudo-samples $x_{n+i}^{(j)} \sim f\left(\cdot |\bar{\theta}_i^{(j)}\right)$ for $i =1, \ldots, T$.}
			\State{Draw $w_1^{(j)}, \ldots, w_{n}^{(j)} \sim\text{Exp}(1)$ and $w_{n+1}^{(j)}, \ldots, w_{n+T}^{(j)} \sim\text{Exp}(T/c)$.}
			\State{Set $\tilde{\theta}^{(j)} = \argmax_{\theta} \left\{\sum_{i=1}^{n} w_i^{(j)} \log f(x_{i}| \theta) + \sum_{i=1}^{T} w_{n+i}^{(j)} \log f\left(x_{n+i}^{(j)}| \theta\right) \right\}$.}
			\EndFor
			\State \Return {Samples $\big\{\tilde{\theta}^{(j)}\big\}_{j=1}^N$. }
		\end{algorithmic}
	\end{algorithm}
	
	To obtain the Edgeworth expansion of Algorithm~\ref{alg:pb_pseudosamples_prior}, we no longer need Assumption \ref{assump:prior_smoothness} to hold so replace it with a milder condition. To this end we define the prior predictive density as $f_{\pi}(x)=\int f(x|\theta) \pi(\theta) \mathrm{d} \theta$. 
	
	\begin{assumption}\label{assump:prior_integral} \normalfont The quantity
		$\alpha_{\theta, \pi}  =\int   \nabla_{\theta} \left\{\log f(x|\theta)\right\}  f_{\pi}(x) \mathrm{d} x$ is continuous  as a function of $\theta\in B$.
	\end{assumption}
	
	\begin{thm}\label{thm:edgeworth_expansion_pseudo}
		Consider model \eqref{eq:Bayesian_model} with a parameter of interest $\theta \in \R^d$ and suppose Assumptions \ref{assump:log_lik} -- \ref{assump:weights}, \ref{assump:prior_integral} hold. As $n \to \infty$ with probability going to 1 the Edgeworth expansion of the density of $I_n^{-1/2} J_{n} n^{1/2}(\tilde{\theta} - \hat{\theta})$  for $\tilde{\theta}$ drawn from Algorithm~\ref{alg:pb_pseudosamples_prior} is given by:
		\begin{equation}\label{eq:edgeworth_expansion_pseudo_misspec_multivariate}
		\phi(y) \left\{1 + I(\theta^{*})^{-1/2}\frac{ \left( c\alpha_{\theta^*, \pi}\right)^T y}{n^{1/2}} + \frac{P_{3,n}(y)}{n^{1/2}} \right\} + o \left(n^{-1/2}\right)
		\end{equation}
		for a constant $\alpha_{\theta^*, \pi}$ defined like in Assumption \ref{assump:prior_integral} for $\theta=\theta^*$,	and for $P_{3,n}$ defined in Theorem~\ref{thm:edgeworth_expansion_penalisation}.
	\end{thm}
	
	The proof of Theorem~\ref{thm:edgeworth_expansion_pseudo} follows analogous steps to the proof of Theorem~\ref{thm:edgeworth_expansion_penalisation}, for details see Supplementary Material \ref{supplementary:psudosamples_proof}.

	Since the prior information enters into the second order Edgeworth expansion via an intractable constant $\alpha_{\theta^*, \pi}$, in this case the impact of the prior has a slightly more cumbersome interpretation than in case of Algorithm~\ref{alg:pb_general}.  However, to shed some light on possible interpretations of this constant, note that $\alpha_{\theta^*, \pi}$ is an integral of $\nabla_{\theta} \left\{\log f(x| \theta^*)\right\}$ with respect to the prior predictive $\int f(x|\theta) \pi(\theta)  \mathrm{d} \theta$ and can also be written as
	\begin{align*}
	\alpha_{\theta^*, \pi}  =\int  \left[ \nabla_{\theta} \left\{\log f(x|\theta^*)\right\}  \left\{ f_{\pi}(x)   - f(x|\theta^*) \right\}\right] \mathrm{d} x.
	\end{align*} 
	since $\int   \nabla_{\theta} \left\{\log f(x| \theta^*)\right\}f(x|\theta^*) dx =0$. Therefore $\alpha_{\theta^{*}, \pi}$ measures the difference between the prior predictive and the predictive distribution under $\theta^*$. 
	
	The form of $\alpha_{\theta^*, \pi}$ makes it difficult to provide detailed general guidelines on the appropriate choice of $c$. For well-specified models $c$ should be interpreted as the effective sample size of the prior, which we make explicit in Section~\ref{section:choice of c} for conjugate priors in exponential families. Importantly, as opposed to the choice of $w_0$ in Algorithm~\ref{alg:pb_general}, the optimal choice of $c$ depends both on the likelihood and the prior.  Parameter $T$ does not appear in the second order Edgeworth expansion and we do not expect it to have much effect on the results.
	\begin{remark}\label{remark:pseudosamples} Posterior Bootstrap via pseudo-samples can also be used for hierarchical models \eqref{model:hierarchical} within Algorithm~\ref{alg:pb_hierarchical_gibbs}.
		The procedure would involve drawing $\bar{\lambda}$ from $p(\lambda| \hat{\theta}_1, \ldots, \hat{\theta}_K)$, then drawing pseudo-samples given $\bar{\lambda}$ to finally  compute $\tilde{\theta}_1, \ldots, \tilde{\theta}_K$. Using the same value of $\bar{\lambda}$ across all $K$ groups induces correlation between $\tilde{\theta}_1, \ldots, \tilde{\theta}_K$. Inference on $\tilde{\lambda}$ given $\tilde{\theta}_1, \ldots \tilde{\theta}_K$ can then be performed using methods described in Section~\ref{section:inference_lambda}, and its theoretical properties would be analogous to those presented in Proposition \ref{prop:wlb_lambda_posterior} and Theorem \ref{thm:hierarchical_large_K_result}.
	\end{remark}
	
	As mentioned in Section~\ref{section:reminders_wlb}, another method of incorporating prior information is Algorithm~1 of \cite{lyddon2018nonparametric}. It is similar to Algorithm~\ref{alg:pb_pseudosamples_prior}, but instead of drawing $\bar{\theta}$ from the prior $\pi(\theta)$, $\bar{\theta}$ is drawn from the Bayesian posterior $\pi(\theta|x_{1:n})$. 	Nevertheless, it turns out that the impact of the prior in that algorithm is negligible, as the prior information does not enter into the second order Edgeworth expansion.
	
	To see this, note that the second order Edgeworth expansion for Algorithm~1 of \cite{lyddon2018nonparametric} can be obtained using analogous calculations to those required to prove Theorem~\ref{thm:edgeworth_expansion_pseudo}. The expansion would thus be given by \eqref{eq:edgeworth_expansion_pseudo_misspec_multivariate}, where instead of the constant $\alpha_{\theta^*,  \pi}$  we have
	\begin{equation*}
	\beta_{\theta^*, \pi}= \int   \nabla_{\theta} \left\{\log f(x|\theta^*)\right\}  \left\{\int f(x|\theta) \pi(\theta|x_{1:n}) \mathrm{d} \theta \right\} \mathrm{d} x.
	\end{equation*}
	Since the posterior $\pi(\theta|x_{1:n})$ concentrates around $\theta^*$ as $n \to \infty$, we have $\beta_{\theta^*, \pi} = o_P(1)$  and therefore the prior information does not enter into the second order Edgeworth expansion. Details are provided in Supplementary Material  \ref{supplementary:other_methods}.

	\subsection{Choice of $c$ in Algorithm~\ref{alg:pb_pseudosamples_prior}}\label{section:choice of c}
	We consider a special case of regular exponential families in a canonical form 
	\begin{equation*}
	f(x|\theta) = h(x) \exp \left\{T(x)^T \theta - A(\theta)\right\}.
	\end{equation*}
	The corresponding conjugate prior on $\theta$ is 
	\begin{equation}\label{eq:exponential_family_prior_density}
	\pi(\theta) = \pi(\theta|\tau, n_0) = H(\tau, n_0) \exp \left\{ \tau^T \theta - n_0 A(\theta)\right\}.
	\end{equation}
	for some hyperparameters $\tau$ and $n_0$. The resulting posterior follows $\pi(\theta|\tau + \sum_{i=1}^n T(x_i), n_0+n)$.
	We can therefore interpret $n_0$ as the number  of pseudo-samples quantifying the amount of information in the prior, called also effective sample size.
	
	In this special case we can compute the expression $c \alpha_{\theta^*, \pi}$ appearing in the Edgeworth expansion for Algorithm~\ref{alg:pb_pseudosamples_prior}. We have 
	\begin{align*}
	\alpha_{\theta^*, \pi} = &  \int   \left\{T(x) - \nabla_{\theta} A(\theta^*)  \right\} \int f(x|\theta) \pi(\theta) \mathrm{d} \theta dx = \int \nabla_{\theta} A(\theta) \pi(\theta) \mathrm{d} \theta - \nabla_{\theta} A(\theta^*).
	\end{align*} 
	where the last equality follows from the property of exponential families that $\E_{\theta} \left\{ T(X) \right\} = \nabla_{\theta} A(\theta)$. Since $\int \nabla_{\theta} \pi(\theta) \mathrm{d} \theta =0$, see \cite{diaconis1979conjugate}, we get that
	\begin{align*}
	\tau  -  n_0 \int \nabla_{\theta} A(\theta) \pi(\theta) \mathrm{d} \theta =  & \int \left\{ \tau -  n_0  \nabla_{\theta} A(\theta) \right\} \pi(\theta)  \mathrm{d} \theta  = \int \nabla_{\theta} \pi(\theta) \mathrm{d} \theta =0.
	\end{align*}
	Then $\alpha_{\theta^*, \pi} = \tau/n_0 - \nabla_{\theta} A(\theta^*).$
	Recall that when we instead use Algorithm~\ref{alg:pb_general}, the corresponding term in the Edgeworth expansion  is given by $w_0^T\nabla_{\theta} \log \pi(\theta^*)$, which in case of the prior \eqref{eq:exponential_family_prior_density} is equal to $w_0^T \left\{\tau - n_0 A(\theta^*)\right\}$. Therefore in this special case setting $c = n_0$ in Algorithm~\ref{alg:pb_pseudosamples_prior} yields the same second order Edgeworth expansion as setting $w_0 = 1$ in Algorithm~\ref{alg:pb_general}. The above reasoning confirms our intuition that for well-specified models we should set $c$ representing the weight of the prior included via pseudo-samples to be equal to the effective sample size of that prior.
	
	In fact using Posterior Bootstrap for well-specified exponential models with conjugate priors seems rather contrived, as one can in this case sample directly from the Bayesian posterior. However, under the misspecified scenario Posterior Bootstrap is more robust than Bayesian inference. In that case we propose to set the optimal $c$ as 
	\begin{equation}\label{eq:setting_c_conjugate_model}
	c^* = \bar{w}_0^* n_0.
	\end{equation}
	for $\bar{w}_0^*$ defined in \eqref{eq:setting_bar_w_0}.
	\FloatBarrier
	\section{Illustrations}\label{section:illustrations}
	\subsection{Toy examples}\label{section:toy_examples}
	We first analyse the following toy model
	\begin{align}
	\begin{split}\label{model:toy_normal_univariate}
	x_1, \ldots, x_n  \sim N(\theta, 1), \quad \theta \sim \text{Gamma}(5,3).\\
	\end{split}
	\end{align}
	We consider three different scenarios by generating $x_1, \ldots, x_{n} \sim N(10, \sigma^2)$ for $n=200$ with $\sigma^2=1$, $\sigma^2=2.8$ or $\sigma^2 =0.6$. The correct model would then be \begin{align}
	\begin{split}\label{model:toy_normal_univariate_true}
	x_1, \ldots, x_n  \sim N(\theta, \sigma^2), \quad \theta \sim \text{Gamma}(5,3).\\
	\end{split}
	\end{align}
	\begin{figure}[!ht]
		\centering
		\includegraphics[width = 0.8\textwidth]{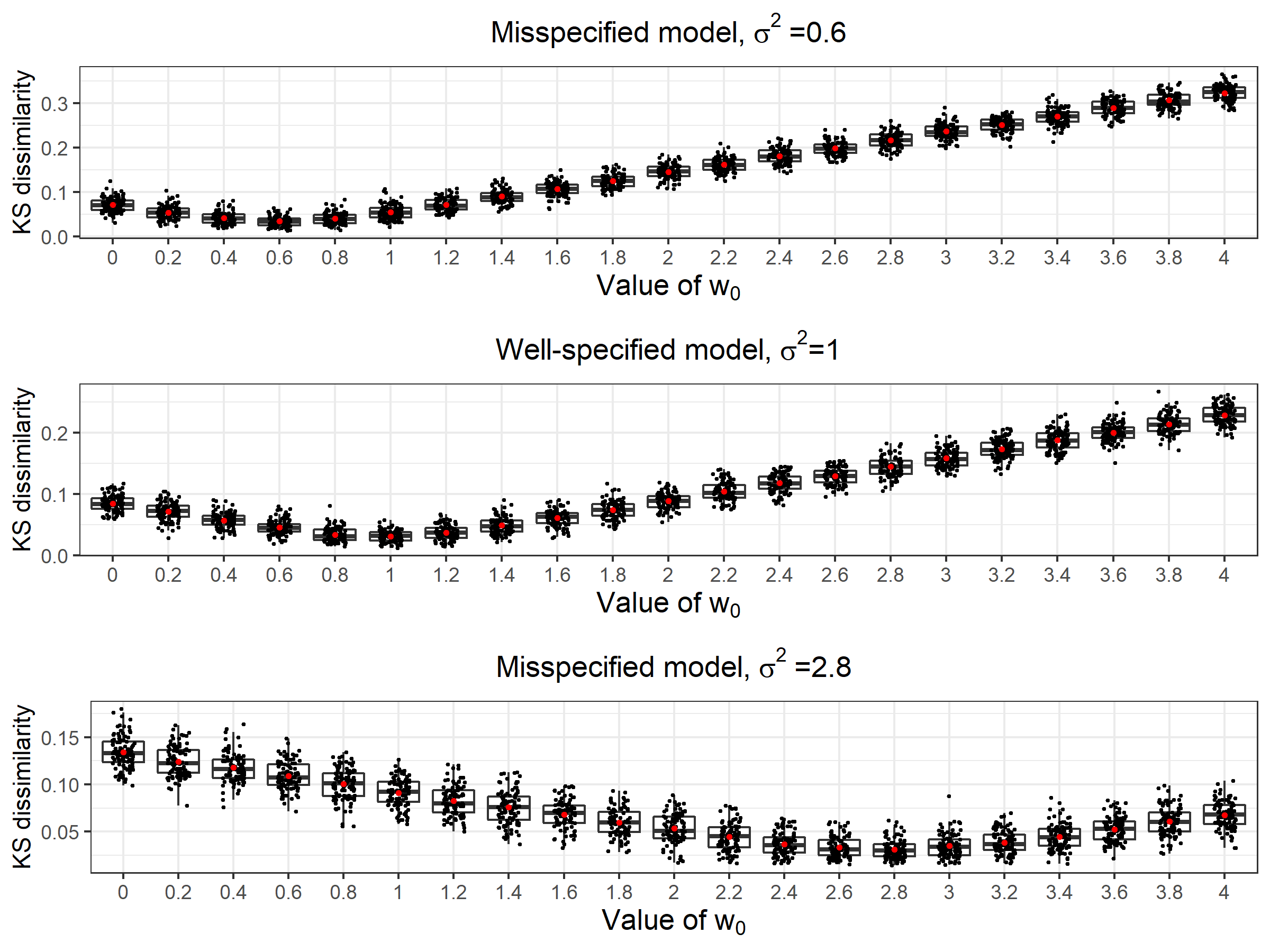}
		\caption{\small Kolmogorov--Smirnov dissimilarity between samples drawn from Posterior Bootstrap for model \eqref{model:toy_normal_univariate}, and samples from the Bayesian posterior for the correct model \eqref{model:toy_normal_univariate_true}. Three scenarios are considered, $\sigma^2 = 0.6$ (top panel), $\sigma^2 =1$ (middle panel), and $\sigma^2 = 2.8$ (bottom panel). Posterior Bootstrap simulations were performed using Algorithm~\ref{alg:pb_general} for a range of values of $w_0$. Each experiment was repeated on 100 simulated datasets with $N = 2000$.}\label{fig:univariate_normal_KS_w_0_boxplot}
	\end{figure}
\noindent Recall from Section~\ref{section:choice_w_0} that $w_0^*$ would in this case be equal to $\sigma^2$. We ran Algorithm~\ref{alg:pb_general} for this model with different values of $w_0$ and measured the Kolmogorov--Smirnov dissimilarity between the empirical distribution of the obtained samples, and samples from the Bayesian posterior applied to the correct model. In this special case making such a comparison is justified since  $\theta$ has the same interpretation of being the mean under both models, and the difference between these models lies only in the uncertainty around the mean. 	As shown in Figure \ref{fig:univariate_normal_KS_w_0_boxplot}, we indeed observe the shortest distance for $w_0 = w_0^*$. For additional plots and comparisons with other methods we refer the reader to Supplementary Material \ref{supplementary:toy_examples}.

	We now analyse another toy example, with a misspecified model
	\begin{align}
	\begin{split}\label{model:toy_normal_multivariate}
	x_1, \ldots, x_n  \sim N(\theta, \Sigma_1), \quad \theta \sim N(\mu_0, \Sigma_0),
	\end{split}
	\end{align}
	whereas correct model is
	\begin{align}
	\begin{split}\label{model:toy_normal_multivariate_correct}
	x_1, \ldots, x_n  \sim N(\theta, \Sigma_2), \quad \theta \sim N(\mu_0, \Sigma_0).
	\end{split}
	\end{align} 
	We set
	$$
	\mu_0 = (5,5), \quad\Sigma_0 = \begin{bmatrix}
	1 & 0 \\ 0 & 1
	\end{bmatrix}, \quad \Sigma_1 = \begin{bmatrix}
	1 & 0.5 \\ 0.5 & 2
	\end{bmatrix}, \quad \Sigma_2 = \begin{bmatrix}
	0.7 & 0.6\\ 0.6 & 3
	\end{bmatrix},
	$$
	so that the example is characterized by overdispersion of one coordinate, and underdispersion of the other. Moreover, following the discussion in Section~\ref{section:choice_w_0}, in this example correlation between variables makes it impossible for \eqref{eq:try_setting_w_0} to hold.
	In our experiment the data  $x_1, \ldots, x_n$ for $n = 200$ is generated from  $N \left\{(0, 0), \Sigma_2\right\}$.
	 In Figure \ref{fig:multivariate_normal_plot} we compare three versions of Posterior Bootstrap for model \eqref{model:toy_normal_multivariate} with standard Bayesian posterior and two methods of robust Bayesian inference: BayesBag  \citep{huggins2019using} and power posteriors with the optimal power $\eta^*$ as proposed by \cite{lyddon2019general}. The plots show that Algorithm~\ref{alg:pb_general} outperforms other methods on this example returning samples that are closer to the Bayesian posterior under the correct model \eqref{model:toy_normal_multivariate_correct}. We notice that the power posterior approach governed by a single parameter $\eta$ is unable to capture overdispersion of one coordinate and underdispersion of the other, whereas BayesBag with the choice of hyperparameters advocated in Section 2.2.1 of \cite{huggins2019using} tends to slightly overestimate the uncertainty. We also conclude that treating $w_0$ as a two-dimensional vector rather than a single value helps to better incorporate the prior information.
	 
	 We present an additional toy example for a hierarchical model in Supplementary Material \ref{supplementary:toy_example_hierarchical}.
	
	\begin{figure}[!ht]
		\centering
		\includegraphics[width = 0.9\textwidth]{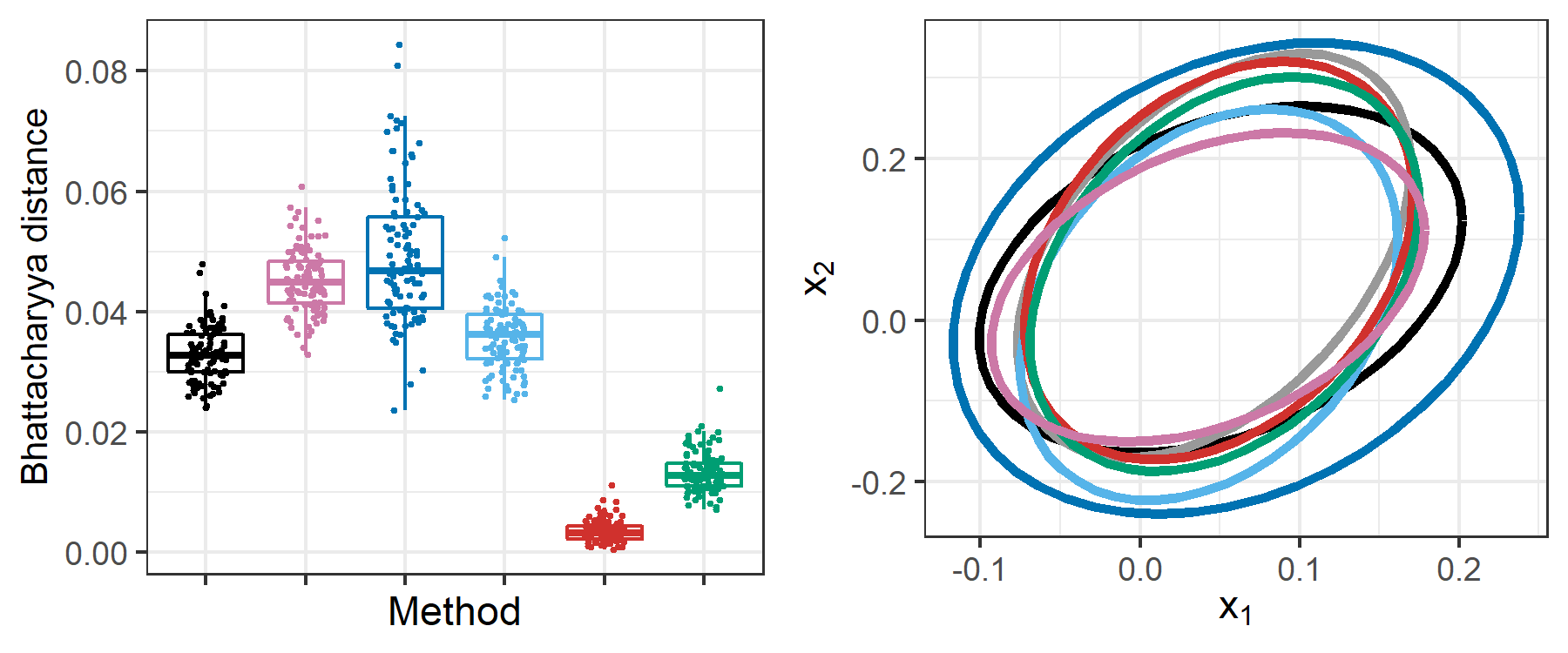}
		\caption{ \small Left: Bhattacharyya distance between samples drawn from the Bayesian posterior for correct model \eqref{model:toy_normal_multivariate_correct}  and  standard Bayesian inference for the misspecified model (black), power posterior inference with optimal parameter $\eta$ according to \cite{lyddon2019general} (pink), BayesBag \citep{huggins2019using} (dark blue), Posterior Bootstrap Algorithm~\ref{alg:pb_pseudosamples_prior} with the optimal parameter $ c^{\star}$ (light blue),  Posterior Bootstrap Algorithm~\ref{alg:pb_general} with $ w_0 $ set according to  \eqref{eq:setting_w_0} (red) and  \eqref{eq:setting_bar_w_0} (green). Each experiment was repeated on 100 simulated datasets with $N = 2000$ samples drawn from each method. Optimal $c^*$ for Algorithm~\ref{alg:pb_pseudosamples_prior} was set as the value that empirically minimized the Bhattacharyya distance over a grid of valued. The size of each bootstrapped dataset in BayesBag was set according to Section 2.2.1 of \cite{huggins2019using}. Right: density contours for samples obtained using the above methods on a single synthetic dataset, and for samples from the Bayesian posterior for the correct model \eqref{model:toy_normal_multivariate_correct} (grey). }\label{fig:multivariate_normal_plot}
	\end{figure}
	
	\subsection{Overdispersed and underdispersed Poisson regression}\label{section:poisson_regression}
	A common approach for modelling count data is the Poisson regression:
	\begin{equation}\label{model:Poisson_regression}
	Y_i \sim \text{Poisson}(\lambda_i), \quad \lambda_i = \exp \left(\beta_0 + \beta_1 x_{1,i} + \cdots + \beta_p x_{p,i}\right), \quad i = 1, \ldots, n,
	\end{equation}
	with a prior distribution on the vector of coefficients $\beta$.
	The above model assumes that $E\left(Y_i|x_i\right) = \text{var}\left(Y_i|x_i\right)$. In practice this assumption may be violated and the data exhibits either overdispersion or underdispersion relative to model \eqref{model:Poisson_regression}. In case of the former an alternative is using the negative binomial regression model, for the latter however no simple alternative is known. 
	
	The COM-Poisson regression \citep{shmueli2005useful, sellers2010flexible} was proposed as a generalization of the Poisson model, which can handle different levels of dispersion, including underdispersion.  Since the COM-Poisson likelihood \citep{conway1962queuing} has an intractable constant,  Bayesian inference with the COM-Poisson regression requires non-standard and computationally costly  Monte Carlo methods, based on a rejection sampler used within the exchange algorithm \citep{moller2006efficient, murray2012mcmc}, as proposed by \cite{chanialidis2018efficient} and \cite{benson2021bayesian}. Another issue with this approach is interpretability since in the COM-Poisson model it is no longer true that $E\left(Y_i|x_i\right) =\beta_0 + \beta_1 x_{1,i} + \cdots + \beta_p x_{p,i}$.

	\begin{figure}
		\centering
		\includegraphics[width = 1\textwidth]{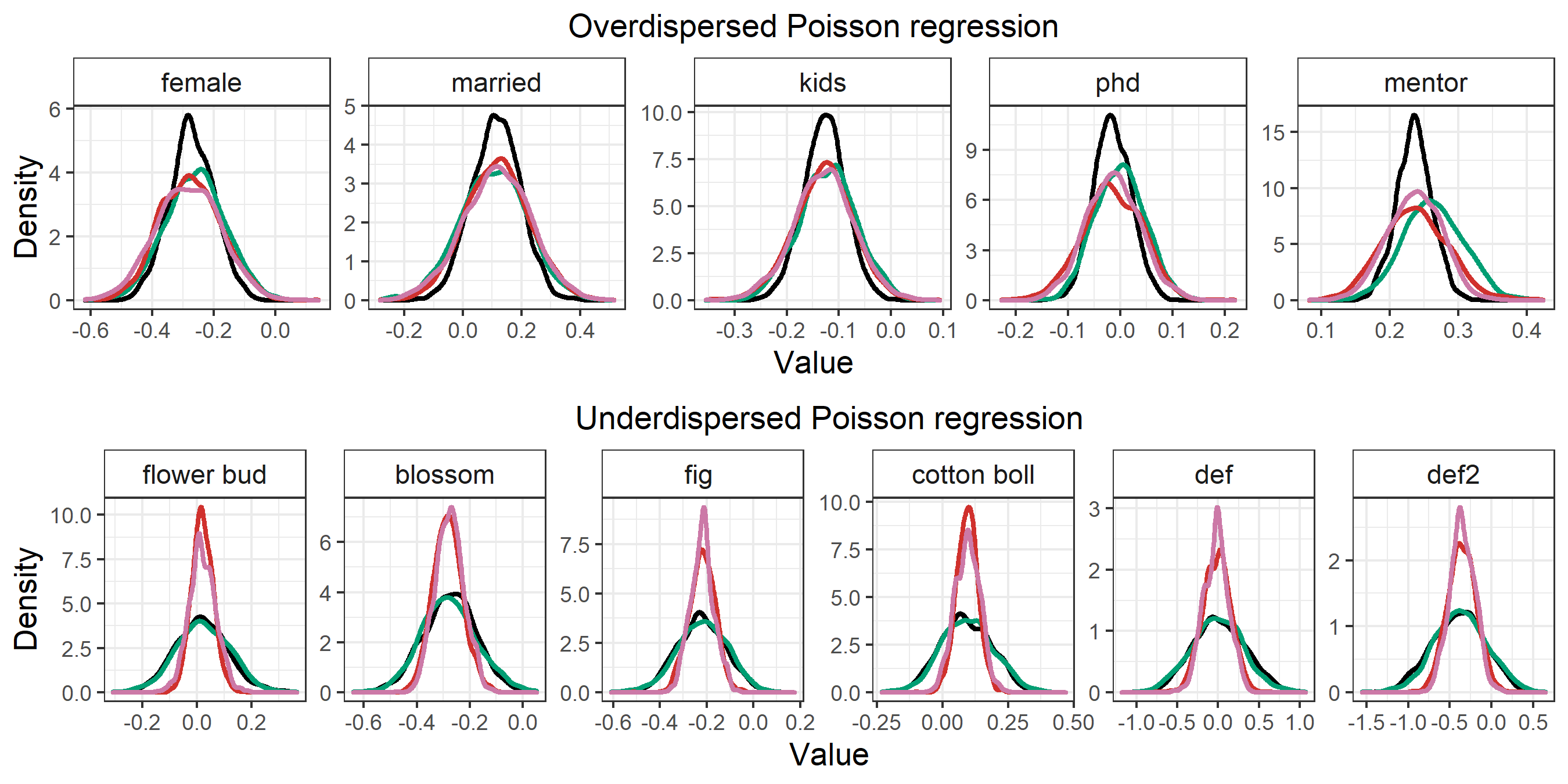}
		\caption{\small Kernel density estimates for the posterior on coefficients $\beta$ in model \eqref{model:Poisson_regression} using four different methods: Bayesian inference  for model \eqref{model:Poisson_regression} (black), Bayesian inference for the negative binomial regression model (green), power posterior inference with optimal parameter $\eta$ according to \cite{lyddon2019general} (pink), and Algorithm \ref{alg:pb_general} (red). We use the \texttt{Articles} dataset (top panel) and the \texttt{cottonbolls} dataset (bottom panel). For both examples we used the independent $N(0, 10^2)$ prior on parameters $\beta$.}\label{fig:poisson_regression_phd_cotton}
	\end{figure}

	An alternative could be to use Algorithm~\ref{alg:pb_general}, which preserves the usual interpretability of parameter $\beta$, while being able to handle both underdispersion and overdispersion. The fact that Algorithm~\ref{alg:pb_general} can be fully parallelized and does not require tuning gives it a significant computational advantage over complex Monte Carlo methods proposed for the COM-Poisson regression. We illustrate the overdispersed Poisson regression with the \texttt{Articles} example from the \texttt{Rchoice} package \citep{sarrias2016discrete} in $\texttt{R}$, which studies the number of articles published during scientists' Ph.D. studies. We preprocess the data following \cite{chanialidis2018efficient}, who analysed the same example. For underdispersed Poisson models, we use as an illustration the \texttt{cottonbolls} example from the \texttt{mpcmp} package \citep{fung2020package}, modelling number of bolls produced by the cotton plants. Figure \ref{fig:poisson_regression_phd_cotton} shows that Algorithm \ref{alg:pb_general} corrects both for overdispersion and underdispersion. Power posteriors also perform well on these examples,  however in general this approach is less flexible in accounting for misspecification than Posterior Bootstrap. 

	\subsection{Dirichlet allocation model for opioid crisis data}\label{section:dirichlet_example}
	We illustrate Posterior Bootstrap for hierarchical models, as well as the use of pseudosamples on a Dirichlet allocation model \citep{glynn2019bayesian} of deaths caused by drug overdoses in the United States \citep{overdose2021}. The data comprise 139 observations, where each observation is the numbers of deaths due to each of six drug types in a given state and year. Our primary interest lies in $\lambda = \left(\lambda_1, \ldots, \lambda_6 \right)$, which describes the expected proportion of deaths in each category. For each coordinate of $ \lambda $, we assume a $\mathrm{N}(0, \tau^2)$ prior truncated to the positive half-line, denoted $ \mathrm{TN}(0, \tau^2) $ below. For $m_k = \left(m_{k,1}, \ldots, m_{k,6}\right)$, where $ m_{k,l} $ is the deaths by drug type $ l $ in state-year $k$, we write the model:
	\begin{gather}\label{model:Dirichlet_allocation}
	\begin{aligned}
	m_k &\sim \mathrm{Multinomial}(\theta_k, n_k), & k &= 1, \dotsc, 139,\\
	\theta_k &\sim \mathrm{Dirichlet}(\lambda), & k &= 1, \dotsc, 139,\\
	\lambda_l &\sim \mathrm{TN}(0, \tau^2), & l &= 1, \dotsc, 6.
	\end{aligned}
	\end{gather}
	By comparing $I_g$ and $J_g$ as defined in \eqref{eq:hierarchical_info_matrices}, we suspect that the Dirichlet model for $\theta_k$ given $\lambda$ is misspecified. This, as well as the fact that the number of classes $K = 139$ is large motivates using \eqref{eq:lambda_posterior_bootstrap} for inference on $\lambda$. As for the inference on $\theta_k$, following the discussion in Section~\ref{subsection:pb_pseudosamples_methodology}  we cannot use the penalization-based approach, due to the $\text{Dirichlet}(\lambda)$ prior possibly with some components smaller than 1. We therefore follow the strategy based on pseudo-samples, and conjugacy of the model motivates setting $c$ according to \eqref{eq:setting_c_conjugate_model}. We present the full resulting algorithm in Supplementary Material \ref{supplement:dirichlet_allocation}.

	Figure~\ref{fig:opioid_density_plot} displays the posterior distribution on  parameter $ \lambda $ in \eqref{model:Dirichlet_allocation} using Posterior Bootstrap and Bayesian inference performed with Metropolis-within-Gibbs Markov chain Monte Carlo.  As $ \tau $ decreases, the truncated normal prior concentrates on zero and we observe a shift in distributions under both methods, which confirms that our method incorporates the prior in a similar way to Bayesian inference. The results for the two methods are different, and we expect this to be due to model misspecification. In Supplementary Material \ref{supplement:dirichlet_allocation} we present an analogous analysis on a synthetic dataset generated from model \eqref{model:Dirichlet_allocation} and show that the results delivered by these methods are then the same. In addition to correcting for model misspecification, Posterior Bootstrap runs in parallel with virtually no tuning and without the requirement for convergence diagnostics in Markov chain Monte Carlo-based inference.
	\begin{figure}
		\centering
		\includegraphics[width = 1\textwidth]{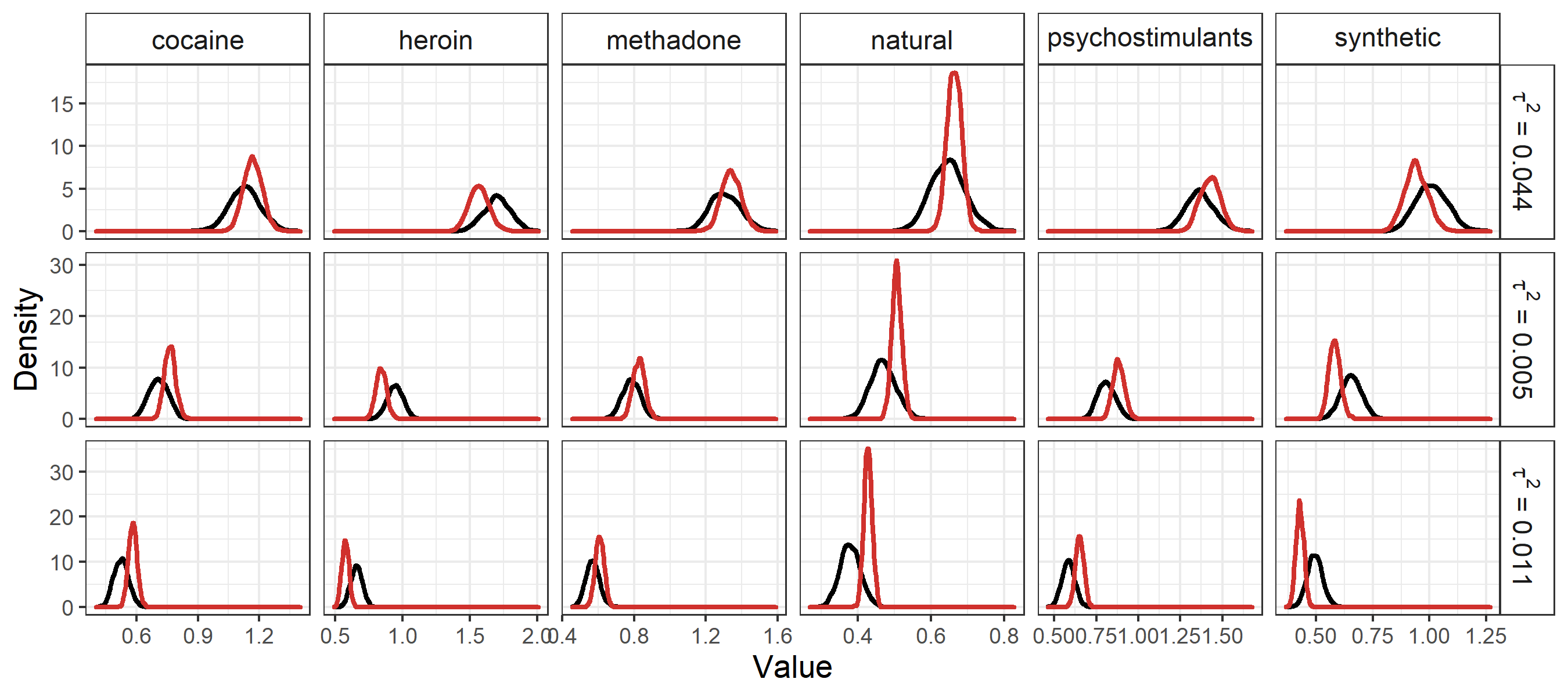}
		\caption{\small Kernel density estimates under standard Bayesian inference (black) and Posterior Bootstrap (red) for parameters $\lambda_l$ in model \eqref{model:Dirichlet_allocation}, corresponding to six different types of drugs. We test this example on three prior distributions considered in \cite{glynn2019bayesian}.}\label{fig:opioid_density_plot}
	\end{figure}

	\section{Limitations of methods based on Weighted Likelihood Bootstrap}\label{section:limitations}
	In previous sections we focused on positive aspects and properties of Weighted Likelihood Bootstrap and its extensions, such as its computational advantages and its capability to correct the uncertainty in case of model misspecification. The aim of this section is to answer some frequently asked questions about limitations of Weighted Likelihood Bootstrap and its extensions.

	There are numerous situations when Posterior Bootstrap  cannot be used for computational reasons, whereas standard Bayesian modelling combined with some Markov chain Monte Carlo methods prove to be successful.	Firstly, by construction of the algorithms it is implicitly assumed that one can numerically find a maximum of the randomized likelihood  function, which is typically done using gradient-based methods. However, if the parameter of interest live in a discrete space, or a mix of discrete and continuous state spaces, this will not be possible. There exists a vast literature on Gibbs samplers that deal effectively with such models, for example  \citep{richardson1997bayesian, titsias2017hamming, zanella2020informed}. Secondly, the Posterior Bootstrap algorithms cannot be used in case of intractable likelihoods, that is, when the likelihood has a form $f(x| \theta)/Z(\theta)$ for some unknown function $Z(\theta)$. Again, several Markov chain Monte Carlo methods have been proposed for target distributions of that type, for example \citep{moller2006efficient, andrieu2009pseudo, murray2012mcmc, deligiannidis2018correlated,  middleton2020unbiased}.

	We would also like to note a limitation related specifically to Algorithm~\ref{alg:pb_pseudosamples_prior}. Recall that the guidance on setting parameter $c$ is well-understood only for conjugate models due to its link with the effective sample size of the prior. Our intuition, confirmed by simulations, is that in non-conjugate scenarios choosing $c$ may be problematic. Pseudo-samples drawn from the prior predictive are often outliers with respect to the original dataset, and this phenomenon is particularly pronounced for vague priors. Therefore assigning too large importance to the pseudo-samples leads to unstable results and consequently distorts the inference. Another issue  is that a one-dimensional parameter $c$ controls the impact of the prior on all coordinates, as we showed in Theorem~\ref{thm:edgeworth_expansion_pseudo}.
	
	There have been attempts in the Bayesian literature to extend the notion of the effective sample size beyond the conjugate scenario, and quantify the amount of information in the prior relative to the assumed sampling distribution, for example \cite{morita2008determining,morita2012prior, neuenschwander2020predictively}.  Developing the appropriate methodology for the choice of $c$ in general models, potentially based on the approaches cited above, as well as supporting it with theoretical results, is an interesting future work direction.

	Even though Posterior Bootstrap methods can alleviate the consequences of model misspecification,  this approach is not always optimal. In particular, as we showed in Theorem~\ref{thm:risk_wlb_power_posteriors}, in certain cases power posteriors are provably better in terms of prediction.  Furthermore, it is  important to emphasize that Posterior Bootstrap, as well as power posteriors and BayesBag described in Section~\ref{section:model_misspecification}, rely on correcting automatically the uncertainty around parameters, while preserving the same asymptotic mean. There are cases when the user has some prior knowledge about the nature of misspecification and has more trust in certain parts of the model than others. The cut model approach, also called modular inference, relies on cutting feedback, or the information flow, from the part of the model which we put less trust in to the one we trust more. This way of tackling misspecification allows the user to make better use of their understanding of the model than when applying more automatic approaches listed earlier. Bayesian literature on the cut model approach includes \citep{liu2009modularization, plummer2015cuts, jacob2017better, carmona2020semi}.

	\subsection*{Acknowledgements}
	This work is supported by the EPSRC and MRC Centre for Doctoral Training in Next Generation Statistical Science: the Oxford-Warwick Statistics Programme, EP/L016710/1, and the Clarendon Fund.  The author thanks Judith Rousseau, who helped enormously with all parts of the paper, and without whom this work would never exist. The author also thanks a person who asked the question of extending Weighted Likelihood Bootstrap to hierarchical models, but wishes to remain anonymous, as well as Luke~J.~Kelly,  Pierre E.~Jacob, Roman D.~Stasi\'nski, Stephen G.~Walker and Jonathan Huggins for helpful discussions and feedback.
	
	\appendix
	
	\section{Proof of Theorem~\ref{thm:risk_wlb_power_posteriors}}\label{supplementary:risk_proof}
	
	\begin{proof}
		Let  $\hat{p}_{\hat{\theta}}$, $\hat{p}_B$ and $\hat{p}_{\pi_{\eta}}$ denote the predictive distributions under $\hat{\theta}$, Efron's bootstrap \citep{efron1979bootstrap} and power posterior $\pi_{\eta}$ respectively.
		
		Firstly, note that for the risk associated with $\hat{p}_B$, we can use directly the formula given by \cite{fushiki2005bootstrap} as follows:
		\begin{align*}
		\begin{split}
		E_{x_{1:n}} \left\{D_{KL}\left(p^*||\hat{p}_B\right) \right\}& =  E_{x_{1:n}}  \left\{ D_{KL}(p^*||\hat{p}_{\hat{\theta}}) \right\} \\
		&  - \frac{1}{2n} \left[ \tr \left\{I(\theta^*) J^{-1}(\theta^*)I(\theta^*) J^{-1}(\theta^*) \right\} - \tr \left\{I(\theta^*) J^{-1}(\theta^*)\right\} \right]  + o\left(1/n\right).
		\end{split}
		\end{align*} 
		To obtain an analogous formula for the risk associated with $\hat{p}_{\pi_{\eta}}$, we can use Lemma 5 of \cite{shimodaira2000improving} with
		\begin{align*}
		K_w^{[1 \cdot 1]} = I(\theta^*), \quad  K_w^{[2]} = J(\theta^*),  \quad
		H_w = \eta J(\theta^*),  \quad  K_w^{[1]} = 0,
		\end{align*}
		where the last equality holds because the expectation of the score function at $\theta^*$ is 0. We therefore have 
		\begin{align*}\label{eq:power_posterior_risk}
		\begin{split}
		E_{x_{1:n}} \left\{D_{KL}\left(p^*||\hat{p}_{\pi_{\eta}}\right)\right\} & =  E_{x_{1:n}}\left\{ D_{KL}(p^*||\hat{p}_{\hat{\theta}})\right\} \\
		&  - \frac{1}{2n} \left[\eta^{-1}\tr \left\{I(\theta^*) J^{-1}(\theta^*)\right\} -\eta^{-1} d \right]  + o\left(1/n\right).
		\end{split}
		\end{align*}
		Recall that $\lambda_1, \ldots, \lambda_d$ denote the eigenvalues of $J^{-1/2}(\theta^*)I(\theta^*) J^{-1/2}(\theta^*)$. By modifying calculations from Section~2 of \cite{fushiki2005bootstrap} we get
		\begin{align*}
		\begin{split}
		E_{x_{1:n}}\left\{ D_{KL}\left(p^*||\hat{p}_{\pi_{\eta}}\right) \right\}  - &  E_{x_{1:n}} \left\{D_{KL}\left(p^*||\hat{p}_B\right)\right\}   \\
		= &  \frac{1}{2n} \left[ \tr \left\{I(\theta^*) J^{-1}(\theta^*)I(\theta^*) J^{-1}(\theta^*) \right\} - \tr \left\{I(\theta^*) J^{-1}(\theta^*)\right\} \right] \\
		&   - \frac{1}{2n} \left[\eta^{-1}\tr \left\{I(\theta^*) J^{-1}(\theta^*)\right\} -\eta^{-1} d \right]  + o\left(1/n\right)\\
		= & \frac{1}{2n}  \tr \left[ \left\{J^{-1/2}(\theta^*)I(\theta^*) J^{-1/2}(\theta^*) \right\}^2\right] - \frac{(1+\eta^{-1})}{2n}\tr \left\{J^{-1/2}(\theta^*)I(\theta^*) J^{-1/2}(\theta^*) \right\} \\
		&  + \frac{\eta^{-1}d}{2n} + o\left(1/n\right) \\
		= &\frac{1}{2n} \sum_{i=1}^d \left\{\lambda_i^2 -(1+\eta^{-1}) \lambda_i + \eta^{-1}\right\}  + o\left(1/n\right) \\
		= &\frac{1}{2n} \sum_{i=1}^d (\lambda_i -1) (\lambda_i - \eta^{-1}) + o\left(1/n\right).
		\end{split}
		\end{align*} 
		Finally, we use Theorem~1 of \cite{fushiki2010bayesian} which captures the relationship between the risk associated with Weighted Likelihood Bootstrap and the risk associated with Efron's bootstrap, that is
		$$
		E_{x_{1:n}}  \left\{D_{KL}\left(p^*||\hat{p}_{\text{WLB}}\right) \right\} = E_{x_{1:n}} \left\{ D_{KL}\left(p^*||\hat{p}_B\right) \right\}  + o\left(n^{-3/2}\right),
		$$
		where $\hat{p}_{\text{WLB}}$ denotes the predictive distribution under Weighted Likelihood Bootstrap. Therefore also
		$$
		E_{x_{1:n}} \left\{D_{KL}\left(p^*||\hat{p}_{\pi_{\eta}}\right) \right\}  -  E_{x_{1:n}}\left\{ D_{KL}\left(p^*||\hat{p}_{\text{WLB}}\right)\right\} = \frac{1}{2n} \sum_{i=1}^d (\lambda_i -1) (\lambda_i - \eta^{-1}) + o\left(1/n\right).
		$$
		We finally observe that the expression $\sum_{i=1}^d (\lambda_i -1) (\lambda_i - \eta^{-1})$ is negative when all eigenvalues are between 1 and $\eta^{-1}$. This completes the proof.
	\end{proof}
	\section{Proofs of theoretical results presented in Sections \ref{section:pb_penalisation} and \ref{section:pb_pseudosamples}}\label{supplementary:edgeworth_proof}
	\subsection{Summary and Assumption \ref{assump:log_lik_smoothness}}\label{supplementary:summary_assumption}
	The proofs of Theorems \ref{thm:edgeworth_expansion_penalisation} and \ref{thm:edgeworth_expansion_pseudo} require similar reasoning and calculations. We therefore present below a detailed version of the proof of  Theorem~\ref{thm:edgeworth_expansion_penalisation} and refer to it when proving Theorem~\ref{thm:edgeworth_expansion_pseudo}. The plan of action is as follows. We start by obtaining an expansion  for $n^{1/2}(\tilde{\theta} - \hat{\theta})$. Then applying Lemma \ref{lemma:auxiliary_lemma_Edgeworth}, we show that the Edgeworth expansion of this expression is indeed valid. Finally, we compute the coefficients of the Edgeworth expansion using the formula provided in Chapter 2 of \cite{ghosh1994higher}. 
	We start by presenting the full version of Assumption \ref{assump:log_lik_smoothness} and introducing the notation used in the proofs.
	
	\setcounter{assumption}{2}
	\begin{assumption}\label{assump:log_lik_smoothness_full}\normalfont(Smoothness of the log likelihood function).(Smoothness of the log likelihood function). There is an open ball $B$ containing $\theta^*$ such that  $\ell(x, \theta)$ is three times continuously differentiable with respect to $\theta \in B$ almost surely on $x$. Furthermore, there exist measurable functions $G_1, G_{2}$ and  $G_{3}$ such that for $\theta \in B$ \begin{align*}
		\Big|\frac{\partial \ell(x,\theta)}{\partial \theta_j}\Big| &\leq G_{1}(x),
		& \E_{p*} \left\{G_1(X)^4\right\} &< \infty,\\
		\Big|\frac{\partial^2 \ell(x,\theta)}{\partial \theta_j \partial \theta_k}\Big| &\leq G_{2}(x),
		& \E_{p*}\left\{ G_{2}(X)^4 \right\} &< \infty,\\
		\Big|\frac{\partial^3 \ell(x, \theta)}{\partial \theta_j \partial \theta_k \partial \theta_l}\Big| &\leq G_{3}(x),
		& \E_{p*} \left\{G_{3}(X)\right\} &< \infty.
		\end{align*}
	\end{assumption}

	\subsection{Notation}
	We will adopt the following notation. Suppose $C \in \R^{d\times d \times d}$ and $v$ is a vector in $\R^d$. Then $v^3 \in \R^{d\times d \times d}$ and the element $ijk$ of $v^3$ is given by 
	\begin{equation*}
	v^3_{[ijk]} = v_i v_j v_k.
	\end{equation*}
	We also denote
	\begin{equation*}
	Cv^3 = \sum_{i=1}^d \sum_{j=1}^d \sum_{k=1}^d C_{[ijk]}  v_i v_j v_k.
	\end{equation*}
	Furthermore, for $v \in \R^d$ and $w \in \R^d$ we define $z = Cvw$ as a $d$-dimensional vector such that
	\begin{equation*}
	z_i = \sum_{j=1}^d \sum_{k=1}^d C_{[ijk]}  v_j w_k.
	\end{equation*}
	In particular, if $z = Cv^2$, then 
	\begin{equation*}
	z_i  = \sum_{j=1}^d \sum_{k=1}^d C_{[ijk]}  v_j v_k.
	\end{equation*}
	The above notation simplifies to standard multiplication for $d=1$.

	\subsection{Auxiliary results}
	Our main tools for obtaining the Edgeworth expansion from the Taylor expansion are the two formulas below. The first one links the characteristic function of a variable $Y \in \R^d$ with its density, that is if $|\psi_Y(t)|<\infty$,
	\begin{equation}\label{eq:density_characteristic_relation}
	g_Y(y)  =	\frac{1}{(2 \pi)^d}  \int_{\R^d} \exp(-it^Ty) \psi_Y(t) \mathrm{d} t_1, \ldots \mathrm{d}t_d.
	\end{equation}
	The second one is useful for computing the integral in equation \eqref{eq:density_characteristic_relation}:
	\begin{equation}\label{eq:characteristic_function_identity}
	\frac{1}{(2 \pi)^d} \int_{-\infty}^{\infty} \exp(-zy) \exp(-z^2/2) (iz)^k \mathrm{d}z = H_k(y) \phi(y),
	\end{equation}
	where $H_k(y)$ is the $k$-th Hermite polynomial and $\phi(y)$ is the density of the standard normal distribution. In particular, we can infer from \eqref{eq:density_characteristic_relation} the following properties, which will be useful in our calculations. For $1 \leq q \leq d$ we have
	\begin{align}\label{eq:characteristic_function_identity_1}
	\frac{1}{(2 \pi)^d} \int \left\{\exp\left(-t^Ty\right) \exp\left(-t^Tt/2\right) i^3 t_q^3 \right\}\mathrm{d}t_1, \ldots \mathrm{d}t_d = H_3(y_q) \phi(y_1)\cdot \ldots\cdot \phi(y_d).
	\end{align}
	Furthermore, for $ 1 \leq q,r \leq d, q \neq r$ we get
	\begin{align}\label{eq:characteristic_function_identity_2}
	\frac{1}{2 \pi}  \int_{\R^d} \left\{\exp\left(-t^Ty\right) \exp\left(-t^Tt/2\right) i^3 t_q^2 t_r\right\} dt_1, \ldots dt_d = H_2(y_q)H_1(y_r) \phi(y_1)\cdot \ldots\cdot \phi(y_d).
	\end{align}
	Finally, for $ 1 \leq q,r,s\leq d$, for pairwise different indices $q,r,s$ we obtain the following equality
	\begin{align}\label{eq:characteristic_function_identity_3}
	\frac{1}{(2 \pi)^d}  \int_{\R^d} \left\{\exp\left(-t^Ty\right) \exp\left(-t^Tt/2\right) i^3 t_q t_r t_s \right\} \mathrm{d}t_1, \ldots \mathrm{d}t_d = H_1(y_q)H_1(y_r)H_1(y_s) \phi(y_1)\cdot \ldots\cdot \phi(y_d).
	\end{align}
	
	Our proof of existence of valid Edgeworth expansions will rely on Lemma \ref{lemma:auxiliary_lemma_Edgeworth} below. Before stating the result, we introduce some additional notation. Let $a_j = \left(a_{j,1}, \ldots, a_{j,d} \right)\in \R^d$ and define $\bar{a}_{2,n} \in \R^{d \times d}$  as
	$$
	\bar{a}_{2,n} = \frac{1}{n} \sum_{j=1}^n a_{j} a_{j}^T.
	$$
	Hence element $(q,r)$  of $\bar{a}_{2,n}$ is equal to $\frac{1}{n} \sum_{j=1}^n a_{j,q} a_{j,r}$. Similarly, let $\bar{a}_{3,n} \in \R^{d \times d \times d}$ be such that element $(q,r,s)$ of $\bar{a}_{3,n}$  is given by $\frac{1}{n} \sum_{j=1}^n a_{j,q} a_{j,r} a_{j,s}$. Furthermore, we let $b_{j}  = \bar{a}_{2, n}^{- 1/2} a_{j}$, and similarly set
	$\bar{b}_{3,n}\in \R^{d \times d \times d}$  so that element $(q,r,s)$ of $\bar{b}_{3,n}$  is given by $\frac{1}{n} \sum_{j=1}^n b_{j,q} b_{j,r} b_{j,s}$, where $b_{j,q}$ denotes element $q$ of $b_{j}$. Note that $a_j$ and $b_j$ may depend on $n$.
	\begin{lemma}\label{lemma:auxiliary_lemma_Edgeworth}
		We consider $\R$-valued  independent random variables $Z_1, \ldots, Z_n$ distributed according to some continuous distribution $Z$ having all moments and satisfying $E(Z) = 1$ and $\text{var}(Z) =1$. Let  $\left(a_j\right)_{j=1}^{n} \in \R^d$, be a sequence of vectors in $\R^d$, possibly depending on $n$ and satisfying conditions a) -- c) below.
		
		Assume that there exists a sequence $\epsilon_n = o(1)$ such that
		\begin{enumerate}[a)]
			
			\item There exist $A_2 \in \R^{d \times d}$, $A_2 \succ 0$ and $A_3 \in \R^{d \times d\times d}$ such that $ \|\bar{a}_{2,n} - A_2\| \leq \epsilon_n $ and $ \|\bar{a}_{3,n} - A_3\| \leq \epsilon_n $ for $n \geq 1$. 
			\item We have $\sup_{n} \frac{1}{n} \sum_{j=1}^n\|a_{j}\|^4 <\infty$. 
			\item For any $n\geq 1$ we have $\sup_{j \leq n} \|a_{j}\| \leq \epsilon_n n^{1/2}$.
		\end{enumerate}

		Then $S_n$ defined as
		\begin{equation*}
		S_n = \frac{\bar{a}_{2,n}^{- 1/2}}{n^{1/2}} \sum_{j=1}^n a_{j} Z_j ,
		\end{equation*}
		has a valid Edgeworth expansion. Moreover the expansion of the characteristic function $\psi_{S_n}$ of $S_n$ up to the term $o\left(n^{-1/2}\right)$ is given by 
		\begin{align*}
		\psi_{S_n}(t) 
		= & \exp \left(  - \frac{t^Tt}{2}  \right) \left( 1 + i^3 \gamma \frac{\sum_{1 \leq q,r,s \leq d} B_{q,r,s} t_{q} t_{r} t_{s}}{6 n^{1/2}}\right) + o \left(n^{-1/2}\right),
		\end{align*}
		where $\gamma = \E \left(Z^3\right)$ and $B= \lim_{n \to \infty} \bar{b}_{3,n}$.
	\end{lemma}
	
	\begin{proof}
		Before proceeding with the proof, observe that the limit $\lim_{n \to \infty} \bar{b}_{3,n}$ exists by definition of $b_{j}$ and assumption a), thus $B$ is well-defined with all elements finite.
		
		The characteristic function of $S_n$ is given by:
		\begin{equation*}
		\psi_{S_n}(t)=  \E  \left\{\exp \left(i t^TS_n \right)\right\} = \E \left[\exp \left\{i t^T \left(\frac{\bar{a}_{2,n}^{- 1/2}}{n^{1/2}} \sum_{j=1}^n a_{j} Z_j  \right)\right\}\right] = \prod_{j=1}^n  \psi_{Z}\left(  \frac{t^Tb_{j}}{n^{1/2}}\right),
		\end{equation*}
		where  $\psi_{Z}$ is the characteristic function of $Z$. Using the fact that $\E(Z) = 0$ and $\E \left(Z^2\right)= 1$ we have
		\begin{align*}
		\psi_{Z}\left(  \frac{t^Tb_{j}}{n^{1/2}}\right)= & \E \left\{\exp \left(i  \frac{ t^Tb_{j}}{n^{1/2}} Z \right)\right\} \\
		= & \E  \left\{ 1 +  \frac{i \left(t^T b_{j}\right)Z}{n^{1/2}} + \frac{i^2\left(t^T b_{j}\right)^2 Z^2}{2n} +   \frac{i^3\left(t^T b_{j}\right)^3 Z^3}{6 n^{3/2}}\right\} + O \left(\frac{\|b_{j}\|^4 \E Z^4}{n^2} \right)\\
		= & 1 -  \frac{\left(t^T b_{j}\right)^2}{2 n} +  \frac{i^3 \left(t^T b_{j}\right)^3\gamma }{6  n^{3/2}} + O \left(\frac{\|b_{j}\|^4 \E Z^4}{n^2} \right)
		\\
		= & 1 -  \frac{\left(t^T b_{j}\right)^2}{2 n} +  \frac{i^3 \left(t^T b_{j}\right)^3\gamma }{6  n^{3/2}} + O \left(\frac{\|b_{j}\|^4}{n^2} \right)
		\end{align*}
		since $\E\left( Z^4\right) < \infty$. 
		
		Note that the definition of $b_{j}$ yields the following simple equality:
		\begin{align}\label{eq:b_j_n_property}
		\begin{split}
		\frac{1}{n} \sum_{j =1}^n (t^T b_{j})^2 = &  \frac{1}{n} \sum_{j =1}^n \left( t^T\bar{a}_{2, n}^{- 1/2} a_{j} a_{j}^T \bar{a}_{2, n}^{- 1/2} t\right)\\ 
		= & t^T\bar{a}_{2, n}^{- 1/2}\left(\frac{1}{n} \sum_{j =1}^n  a_{j} a_{j}^T\right) \bar{a}_{2, n}^{- 1/2} t = t^Tt.
		\end{split}
		\end{align} Therefore 
		\begin{align*}
		\prod_{j=1}^n  \psi_{Z}\left(  \frac{t^Tb_{j}}{n^{1/2}}\right)= & \exp \left[ \sum_{j=1}^n \log \left\{ 1 -  \frac{\left(t^T b_{j}\right)^2}{2 n} +  \frac{i^3 \left(t^T b_{j}\right)^3\gamma }{6 n^{3/2}} + O \left(\frac{\|b_{j}\|^4}{n^2} \right)  \right\} \right] \\
		= & \exp \left( \sum_{j=1}^n \left[ -  \frac{\left(t^T b_{j}\right)^2}{2 n} +  \frac{i^3 \left(t^T b_{j}\right)^3\gamma }{6 n^{3/2}} + O \left(\frac{\|b_{j}\|^4}{n^2} \right)  + O \left\{ \frac{\left(t^T b_{j}\right)^4}{4 n^2} \right\}\right] \right)  \\
		= & \exp \left\{  - \frac{t^Tt}{2} +  \frac{i^3 \sum_{j=1}^n \left(t^T b_{j}\right)^3\gamma }{6 n^{3/2}} + O \left(n^{-1} \right) \right\} \\
		= & \exp \left(  - \frac{t^Tt}{2}  \right) \left( 1 + i^3 \gamma \frac{\sum_{1 \leq q,r,s \leq d} B_{q,r,s} t_{q} t_{r} t_{s}}{6 n^{1/2}}\right) + o \left(n^{-1/2}\right).
		\end{align*}
		The second equality above comes from the Taylor expansion of $\log (1-x)$ around $x = 0$ and the fact that  for some constant $c$ and any $n$ we have
		$
		\sup_{j \leq n} \|b_{j}\| \leq c \epsilon_n n^{1/2}$  so that 
		$\sup_{j \leq n} {\|b_{j}\|^2}/n = o(1),
		$ 
		which follows from  convergence of $\bar{a}_{2,n}$ to $A_2$ and assumption c).
		The third equality above results from \eqref{eq:b_j_n_property} and the fact that $\sup_{n}  \sum_{j=1}^n\|b_j\|^4/n <\infty$ by convergence of $\bar{a}_{2,n}$ to $A_2$ and assumption b). The last equality comes from  $\|\bar{b}_{3,n} - B\| = o(1)$.
		Applying formulas \eqref{eq:density_characteristic_relation} as well as \eqref{eq:characteristic_function_identity_1}, \eqref{eq:characteristic_function_identity_2} and \eqref{eq:characteristic_function_identity_3} we get the Edgeworth expansion up to $o \left( n^{-1/2}\right)$.
	\end{proof}

	\subsection{Proof of Theorem~\ref{thm:edgeworth_expansion_penalisation}}
	\subsubsection*{Taylor expansion}
	We begin the proof by writing the Taylor expansion as follows. Let
	\begin{equation*}
	L_n(\theta)  = \sum_{j=1}^n w_j \ell(x_i, \theta) + w_0^T \log \pi(\theta).
	\end{equation*}
	Then by Assumptions \ref{assump:log_lik_smoothness_full} and \ref{assump:prior_smoothness}, using notation from Sections \ref{section:edgeworth_expansions} and \ref{section:theoretical_pb_penalization} we have  
	\begin{equation*}
	0 = \nabla_{\theta} L_n(\tilde{\theta})  = \nabla_{\theta} L_n(\hat{\theta}) + D^2_{\theta}L_n(\hat{\theta})(\tilde{\theta} - \hat{\theta}) + D^3_{\theta}L_n(\hat{\theta})\frac{(\tilde{\theta} - \hat{\theta})^2}{2}  +   o_P \left( \|\tilde{\theta} - \hat{\theta} \|^2\right)
	\end{equation*}
	and consequently
	\begin{equation}\label{eq:grad_L_n_Taylor_exp_1}
	\left\{ -\frac{1}{n}D^2_{\theta}L_n(\hat{\theta}) \right\} n^{1/2}(\tilde{\theta} - \hat{\theta}) =  \frac{\nabla_{\theta} L_n(\hat{\theta})}{n^{1/2}} + \frac{D^3_{\theta}L_n(\hat{\theta})(\tilde{\theta} - \hat{\theta})^2}{2n^{1/2}}   + o_P \left( \|\tilde{\theta} - \hat{\theta} \|^2\right)
	\end{equation}
	We now use Assumptions \ref{assump:log_lik} and \ref{assump:identifiability} to state that $(\tilde{\theta} - \hat{\theta}) = o_P(1)$. 
	We use the following notation 
	\begin{align*}
	S_{1,w,n}  = & \frac{\sum_{j=1}^n (w_j -1) \nabla_{\theta} \ell(x_j, \hat{\theta})}{n^{1/2}} ,  & S_{2,w,n} = & \frac{\sum_{j=1}^n (w_j -1) \left\{-D^2_{\theta} \ell(x_j, \hat{\theta})\right\}}{n^{1/2}},  \\
	J_{w,n}  =   &-\frac{1}{n}\sum_{j=1}^n w_j D^2_{\theta} \ell(x_j,\hat{\theta}),  &  \mu_{3,w,n} = & \frac{1}{n}\sum_{j=1}^n w_j D^3_{\theta}\ell(x_j,\hat{\theta}),\\
	\mu_{3,n} = & \frac{1}{n}\sum_{j=1}^n  D^3_{\theta}\ell(x_j,\hat{\theta}),& \mu_{3} = & \E_{p^*} \left\{ D^3_{\theta}\ell(X,{\theta}^*)\right\},\\
	A_{3,n}= & \frac{1}{n}\sum_{j=1}^n \left\{\nabla_{\theta}\ell(x_j,\hat{\theta})\right\}^3, & A_3 = &  \E_{p^*} \left\{\nabla_{\theta}\ell(X,{\theta}^{*})\right\}^3, \\
	\tilde{L}_{1,2,n}= &\frac{1}{n}\sum_{j=1}^n  D^2_{\theta} \ell(x_j,\hat{\theta}) J_n^{-1} \nabla_{\theta}\ell(x_j,\hat{\theta}), &   \tilde{L}_{1,2}= & \E_{p^*} \left\{D^2_{\theta} \ell(X,{\theta}^{*}) J(\theta^*)^{-1} \nabla_{\theta}\ell(X,{\theta}^{*}) \right\}.
	\end{align*}
	We observe that since $\|J_n - J(\hat{\theta})\| =o_P(1)$ and $\|I_n - I(\hat{\theta})\| =o_P(1)$, then by Assumption~\ref{assump:positive_definite} matrices $I_n, J_n$ as well as $J_{w,n}$ are positive definite and hence invertible for large $n$. Furthermore,  Assumption~\ref{assump:log_lik_smoothness_full} ensures that $A_3$, $\tilde{L}_{1,2}$ and $\mu_3$ are finite. In passing,  $A_3$ and $\mu_3$ are defined above for an arbitrary dimension $d$, they are however consistent with the one-dimensional definitions presented in \eqref{eq:notation_univariate}.
	
	Note also that since $\sum_{j=1}^n \nabla_{\theta} \ell(x_i, \hat{\theta}) = 0$, then $S_{1,w,n}  = n^{-1/2} \left\{\sum_{j=1}^n w_j \nabla_{\theta}\ell(x_i, \hat{\theta})\right\}$.  We can therefore rewrite  \eqref{eq:grad_L_n_Taylor_exp_1} as
	\begin{align}\label{eq:grad_L_n_Taylor_exp_2}
	n^{1/2}(\tilde{\theta} - \hat{\theta}) =  & J_{w,n}^{-1} \left\{ S_{1,w,n} + \frac{w_{0} \cdot  \nabla_{\theta} \log \pi(\hat\theta)}{n^{1/2}} \right\} +J_{w,n}^{-1} \frac{\mu_{3,w,n}(\tilde{\theta} - \hat{\theta})^2}{2n^{1/2}}   +   o_P \left( \|\tilde{\theta} - \hat{\theta} \|^2\right).
	\end{align}
	
	Let $\textbf{I}$ denote the $d \times d$ identity matrix. We have $J_{w,n} = J_n + S_{2,w,n}/n^{1/2}$, and so
	\begin{align}\label{eq:J_w_n_inverse}
	\begin{split}
	J_{w,n}^{-1} = & \left\{J_n \left(\textbf{I} + J_n^{-1} S_{2,w,n}/n^{1/2}\right) \right\}^{-1} \\
	= & \left(\textbf{I} + J_n^{-1} S_{2,w,n}/n^{1/2}\right)^{-1} J_n^{-1}\\
	= & \left\{\textbf{I} - J_n^{-1} S_{2,w,n}/n^{1/2} +  o\left(n^{-1/2}\right) \right\}J_n^{-1} \\
	= & J_n^{-1} - J_n^{-1} S_{2,w,n}J_n^{-1}/n^{1/2} +  o\left(n^{-1/2}\right) \\
	= & J_n^{-1} \left(\textbf{I} - S_{2,w,n}J_n^{-1}/n^{1/2} \right)+  o\left(n^{-1/2}\right),
	\end{split}
	\end{align}
	where the third equality comes from the Taylor expansion of $1/(1+x)$ at $x = 0$.

	The first order approximation of   $n^{1/2}(\tilde{\theta} - \hat{\theta})$ is thus given by
	\begin{align}\label{eq:Taylor_first_order-approximation}
	n^{1/2}(\tilde{\theta} - \hat{\theta}) 
	= & J_{n}^{-1} S_{1,w,n}  + o \left(1\right).
	\end{align}
	Plugging \eqref{eq:Taylor_first_order-approximation} and \eqref{eq:J_w_n_inverse} into \eqref{eq:grad_L_n_Taylor_exp_2} we get
	\begin{align*}
	n^{1/2}(\tilde{\theta} - \hat{\theta}) = &   J_n^{-1} \left\{S_{1,w,n} + \frac{w_{0} \cdot  \nabla_{\theta} \log \pi(\hat{\theta})}{n^{1/2}} - \frac{S_{2,w,n}J_n^{-1}S_{1,w,n}}{n^{1/2}} \right\}\\
	& \qquad + J_{w,n}^{-1}\frac{\mu_{3,w,n} \left( J_{n}^{-1} S_{1,w,n}\right)^2}{2 n^{1/2}} + o \left(n^{-1/2}\right)\\
	= &  J_n^{-1} \left\{S_{1,w,n} + \frac{w_{0} \cdot  \nabla_{\theta} \log \pi(\hat{\theta})}{n^{1/2}} - \frac{S_{2,w,n}J_n^{-1}S_{1,w,n}}{n^{1/2}} \right\}\\
	& \qquad + J_{n}^{-1}\frac{\mu_{3} \left( J_{n}^{-1} S_{1,w,n}\right)^2}{2 n^{1/2}} + o \left(n^{-1/2}\right),
	\end{align*}
	where the last equality results from
	$
	J_{w,n} = J_n + S_{2,w,n}/n^{1/2}$ and $\mu_{3,w,n} = \mu_{3} + o_P(1).
	$
	We finally get
	\begin{align}\label{eq:final_Taylor-exp_alg_pb_general}
	\begin{split}
	I_n^{-1/2} J_{n} n^{1/2}(\tilde{\theta} - \hat{\theta}) 
	= &  I_n^{-1/2} \left\{S_{1,w,n} + \frac{w_{0} \cdot \nabla_{\theta} \log \pi(\hat{\theta})}{n^{1/2}} - \frac{S_{2,w,n}J_n^{-1}S_{1,w,n}}{n^{1/2}} \right\}\\
	& \qquad +I_n^{-1/2}\frac{\mu_{3} \left( J_{n}^{-1} S_{1,w,n}\right)^2}{2 n^{1/2}} + o \left(n^{-1/2}\right).
	\end{split}
	\end{align}
	We will use the above result to obtain the Edgeworth expansion.

	\subsubsection*{Edgeworth expansion -- multidimensional case}
	\begin{proof}
		Our proof will rely mainly on results from Chapter 2 of \cite{ghosh1994higher}.
		Firstly, we apply Lemma~\ref{lemma:auxiliary_lemma_Edgeworth} to 
		$$
		Z_j = w_j -1 \quad \text{and}\quad a_{j} =  \left\{\frac{\partial \ell(x_j, \hat{\theta})}{\partial \theta_p}, \frac{\partial^2 \ell(x_j, \hat{\theta})}{\partial \theta_r \partial \theta_s} \right\}, \quad 1 \leq p \leq d, 1 \leq r \leq s \leq d,
		$$
		that is, $a_{j}$ is a vector of all first and second partial derivatives of $\ell(x_j, \hat{\theta})$, to get that $(S_{1,w,n}, S_{2,w,n})$  has a joint valid Edgeworth expansion. We justify below briefly why  the conditions of Lemma~\ref{lemma:auxiliary_lemma_Edgeworth} are satisfied.  Observe that the assumptions about $Z_j$ hold by Assumption \ref{assump:weights}. Conditions b) and c), as well as convergence to $A_2$ and $A_3$ in condition a) of Lemma \ref{lemma:auxiliary_lemma_Edgeworth}, hold with probability 1 with respect to the observations $x_1, \ldots, x_n$ by Assumption \ref{assump:log_lik_smoothness_full}, and the strong law of large numbers. We also have $A_2 \succ 0$ as required in condition a) by Assumption \ref{assump:positive_definite}.
		
		Now let $X_n$ denote the right hand side of equation \eqref{eq:final_Taylor-exp_alg_pb_general}. Since conditioning on the observations $x_1, \ldots, x_n$ the variable $X_n$ is a polynomial of $(S_{1,w,n}, S_{2,w,n})$, it also has a valid Edgeworth expansion, and we can apply the theory outlined in Chapter 2 of \cite{ghosh1994higher} to obtain the coefficients of this expansion.
		
		By equation 2.7 of \cite{ghosh1994higher} the characteristic function of $X_n$ satisfies
		\begin{equation}\label{eq:ghosh_chapter_2_formula}
		\psi_{X_n} =  \exp \left( -\frac{t^T t}{2}\right) \left\{1 + \kappa_{1,n} \left(it\right) + \frac{\kappa_{3,n} \left(it\right)^3}{6}  \right\}+ o \left(n^{-1/2}\right),
		\end{equation}
		where $\kappa_{1,n}$ and $\kappa_{3,n}$ are the first and the third cumulant, respectively, calculated up to $o \left(n^{-1/2}\right)$.
		
		We also define $Z_n: = (Z_{n,1}, \ldots, Z_{n,d})$ where
		$$
		Z_{n,q} : = \frac{1}{n}\sum_{j=1}^n  \nabla\ell(x_j,\hat{\theta})^T J_n^{-1} \mu_{3,q} J_n^{-1} \nabla\ell(x_j,\hat{\theta}), \quad  q=1, \ldots, d,
		$$
		and $\mu_{3,q} \in \R^{p \times p}$ is a matrix such that element $[r,s]$ of  $\mu_{3,q}$ is equal to element $[q,r,s]$ of $ \mu_{3}$. In an analogous way we define $Z = (Z_1, \ldots, Z_d)$, where
		\begin{equation*}
		Z_{q}  = \E_{p^*} \left\{ \nabla\ell(X,\theta^*)^T J(\theta^*)^{-1} \mu_{3,q} J(\theta^*)^{-1} \nabla\ell(X,\theta^*)\right\}, \quad  q=1, \ldots, d.
		\end{equation*}
		We observe that by Assumption \ref{assump:log_lik_smoothness_full} $Z_q$ is finite for each $q=1, \ldots, d$. We obtain
		\begin{equation*}
		\E\left( X_n\right) =    I_n^{-1/2} \left\{\frac{ w_{0} \cdot  \nabla_{\theta} \log \pi(\hat{\theta})}{n^{1/2} } + \frac{Z_n}{2 n^{1/2}} + \frac{ \tilde{L}_{1,2,n}}{n^{1/2}}+ o \left(n^{-1/2}\right) \right\}.
		\end{equation*}
		Therefore we can write
		\begin{equation*}
		\kappa_{1,n} = I(\theta^*)^{-1/2} \left(\frac{ w_{0} \cdot  \nabla_{\theta} \log \pi(\theta^*)}{n^{1/2} } + \frac{Z}{2 n^{1/2}} + \frac{ \tilde{L}_{1,2}}{n^{1/2}} \right).
		\end{equation*}
		Recall $ \E \left(X_n - \E X_n \right)^3 = \kappa_{3,n} + o \left( n^{-1/2} \right)$. Since the term $ n^{-1/2}\left\{ w_0^T \nabla_{\theta} \log \pi(\hat{\theta})\right\}$ cancels in the expression $X_n - \E X_n$, we get that $\kappa_{3,n}$ does not depend on the prior. Finally, using formula \eqref{eq:ghosh_chapter_2_formula}, and equations \eqref{eq:characteristic_function_identity_1}, \eqref{eq:characteristic_function_identity_2} and \eqref{eq:characteristic_function_identity_3} we conclude with the general form of the Edgeworth expansion
		\begin{equation*}
		\phi(y) \left[1 + I(\theta^{*})^{-1/2}\frac{ \left\{ w_{0} \cdot \nabla_{\theta}\log \pi(\theta^*)\right\}^T y}{n^{1/2}} + \frac{P_{3,n}(y)}{n^{1/2}} \right] + o \left(n^{-1/2}\right),
		\end{equation*}
		where $P_{3,n}(y)$ is an order three polynomial independent from the prior. This completes the proof of equation \eqref{eq:edgeworth_expansion_general_pb_misspec_multivariate}.
	\end{proof}
	
	\subsubsection*{Edgeworth expansion -- one-dimensional case}
	\begin{proof}
		When $d=1$, we can obtain a more explicit form of the Edgeworth expansion. In particular, $\kappa_{1,n}$ becomes
		\begin{equation*}
		\kappa_{1,n} = \frac{ w_0 \nabla_{\theta} \log \pi(\theta^*)}{n^{1/2} I(\theta^{*})^{1/2}} + \frac{\mu_{3}}{2 n^{1/2}}\frac{I(\theta^{*})^{1/2}}{J(\theta^{*})^2} + \frac{ L_{1,2}}{n^{1/2}I(\theta^{*})^{1/2}J(\theta^{*})},
		\end{equation*}
		for $L_{1,2}$ defined in \eqref{eq:notation_univariate}. We additionally define
		$$
		L_{1,2,n}= \frac{1}{n} \sum_{j=1}^n  D_{\theta}^2 \ell(x_j,{\theta}^{*}) \nabla_{\theta}\ell(x_j,{\theta}^{*}).
		$$
		We then get
		\begin{align*}
		\E \left(X_n - \E X_n \right)^3 =   \E \left(\frac{S_{1,w,n}^3}{I_n^{3/2}}\right)  + & 3 \E \left\{ \frac{S_{1,w,n}^2}{I_n} \left( \frac{\mu_{3}}{2 n^{1/2}}\frac{S_{1,w,n}^2}{I_n^{1/2}J_{n}^2} - \frac{\mu_{3}}{2 n^{1/2}}\frac{I_n}{I_n^{1/2}J_{n}^2}\right) \right\}\\
		- & 3 \E  \left\{\frac{S_{1,w,n}^2}{I_n} \left( \frac{ L_{1,2,n}}{n^{1/2}I_n^{1/2}J_{n}} + \frac{S_{1,w,n}S_{2,w,n}}{n^{1/2}I_n^{1/2}J_{n}}\right)\right\}+ o \left(n^{-1/2}\right) .
		\end{align*}
		We have
		\begin{align*}
		\E \left(S_{1,w,n}^3\right) = & \frac{\gamma}{n^{1/2}}   A_{3,n},\\
		\E \left(S_{1,w,n}^4\right) = & 3 I_n^2 +  o \left(n^{-1/2}\right),\\
		\E\left( S_{1,w,n}^3 S_{2,w,n}\right) = & -3 I_n L_{1,2,n} +  o \left(n^{-1/2}\right),
		\end{align*}
		therefore
		\begin{align*}
		\E \left(X_n - \E X_n \right)^3 =    \frac{\gamma A_{3,n}}{n^{1/2}I_n^{3/2}}  +  &\left( \frac{9\mu_{3}}{2 n^{1/2}}\frac{I_{n}^2}{I_n^{3/2}J_{n}^2} - \frac{3 \mu_{3}}{2 n^{1/2}}\frac{I_n^2}{I_n^{3/2}J_{n}^2}\right)\\
		- & \left( \frac{ 3I_n L_{1,2,n}}{n^{1/2}I_n^{3/2}J_{n}} - \frac{9 I_n L_{1,2,n}}{n^{1/2}I_n^{3/2}J_{n}}\right)+ o \left(n^{-1/2}\right).
		\end{align*}
		We get the following value of $\kappa_{3,n}$
		\begin{equation*}
		\kappa_{3,n} = \frac{\gamma A_{3}}{n^{1/2}I(\theta^{*})^{3/2}} + \frac{6\mu_{3}}{2 n^{1/2}}\frac{I(\theta^{*})^{1/2}}{J(\theta^{*})^2} + \frac{6 L_{1,2}}{n^{1/2}I(\theta^{*})^{1/2}J(\theta^{*})}
		\end{equation*}
		Using equation \eqref{eq:ghosh_chapter_2_formula} we get 
		\begin{align*}
		\psi_{X_n} =  & \exp \left( -\frac{t^2}{2} \right) \Bigg[1 +\left\{  \frac{ w_0 \nabla_{\theta} \log \pi(\theta^*)}{n^{1/2} I(\theta^{*})^{1/2}} + \frac{\mu_{3}}{2 n^{1/2}}\frac{I(\theta^{*})^{1/2}}{J(\theta^{*})^2} +	 \frac{ L_{1,2}}{n^{1/2}I(\theta^{*})^{1/2}J(\theta^{*})}\right\} (it)   \\
		& \quad +  \left\{\frac{\gamma A_{3}}{6n^{1/2}I(\theta^{*})^{3/2}} + \frac{\mu_{3}}{2 n^{1/2}}\frac{I(\theta^{*})^{1/2}}{J(\theta^{*})^2} + \frac{ L_{1,2}}{n^{1/2}I(\theta^{*})^{1/2}J(\theta^{*})} \right\} (it)^3  \Bigg]+ o \left(n^{-1/2}\right),
		\end{align*}
		and finally equations \eqref{eq:characteristic_function_identity} and \eqref{eq:density_characteristic_relation} yield the Edgeworth expansion of the density
		\begin{align*}
		g_{X_n}(y) =  & \phi(y) \Bigg[1 +\left\{  \frac{ w_0 \nabla_{\theta} \log \pi(\theta^*)}{n^{1/2} I(\theta^{*})^{1/2}} + \frac{\mu_{3}}{2 n^{1/2}}\frac{I(\theta^{*})^{1/2}}{J(\theta^{*})^2} + \frac{ L_{1,2}}{n^{1/2}I(\theta^{*})^{1/2}J(\theta^{*})}\right\} H_1(y)   \\
		& \quad +  \left\{ \frac{\gamma A_{3}}{6n^{1/2}I(\theta^{*})^{3/2}} + \frac{\mu_{3}}{2 n^{1/2}}\frac{I(\theta^{*})^{1/2}}{J(\theta^{*})^2} + \frac{ L_{1,2}}{n^{1/2}I(\theta^{*})^{1/2}J(\theta^{*})} \right\} H_3(y)  \Bigg]+ o \left(n^{-1/2}\right),
		\end{align*}
		which completes the proof of \eqref{eq:edgeworth_expansion_general_pb_misspec}.
	\end{proof}
	
	\subsection{Proof of Theorem~\ref{thm:edgeworth_expansion_pseudo}}\label{supplementary:psudosamples_proof}
	
	\begin{proof}
		Analogous steps that we took to obtain the Edgeworth expansion for Algorithm~\ref{alg:pb_general} will lead us to the result for Algorithm~\ref{alg:pb_pseudosamples_prior}. When considering the Taylor expansion, instead of  equation \eqref{eq:final_Taylor-exp_alg_pb_general} we get
		\begin{align*}
		\tilde{X}_n = & I_n^{-1/2} \left\{S_{1,w,n} + \frac{ \sum_{j=n+1}^{n+T} w_j  \nabla_{\theta}\ell(x_j, \hat{\theta})}{n^{1/2}} - \frac{S_{2,w,n}J_n^{-1}S_{1,w,n}}{n^{1/2}} \right\}\\
		& \qquad +I_n^{-1/2}\frac{\mu_{3} \left( J_{n}^{-1} S_{1,w,n}\right)^2}{2 n^{1/2}} + o \left(n^{-1/2}\right).
		\end{align*}
		Notice that ${n^{-1/2}}\sum_{j=n+1}^{n+T} w_j  \nabla_{\theta}\ell(x_j, \hat{\theta})$ is independent from $(S_{1,w,n}, S_{2,w,n})$ conditionally on the data, and has a valid Edgeworth expansion, therefore there exists a valid joint expansion of 
		\begin{equation}\label{eq:joint_expression_for_expansion}
		\left\{S_{1,w,n}, S_{2,w,n}, \frac{1}{n^{1/2}}\sum_{j=n+1}^{n+T} w_j  \nabla_{\theta}\ell(x_j, \hat{\theta})\right\}.
		\end{equation}
		Since $\tilde{X}_n$ is a polynomial of \eqref{eq:joint_expression_for_expansion}, similar reasoning to that used for Theorem~\ref{thm:edgeworth_expansion_penalisation} yields an  expansion
		\begin{equation*}
		\psi_{\tilde{X}_n} =  \exp \left( -\frac{t^Tt}{2} \right) \left\{1 + \tilde{\kappa}_{1,n} (it) +\tilde{\kappa}_{3,n} \frac{(it)^3}{6}  \right\}+ o \left(n^{-1/2}\right),
		\end{equation*}
		where $\tilde{\kappa}_{1,n}$ and $\tilde{\kappa}_{3,n}$ are the first and the third cumulant, respectively, calculated up to $o \left(n^{-1/2}\right)$. Since  $\E\left( w_j\right) = c/T$ for $j = n+1, \ldots, n+T$ we have
		\begin{align}\label{eq:prior_predictive_calculation}
		\begin{split}
		\E \left\{\sum_{j=n+1}^{n+T} w_j  \nabla_{\theta}\ell(x_j, \hat{\theta}) \right\} = &\sum_{j=n+1}^{n+T} \E_{\pi(\theta)}\left[\E_{f(x_j|\theta)} \left\{ w_j  \nabla_{\theta}\ell(x_j, \hat{\theta}) \Big| \theta \right\}\right]\\
		= &  c  \int   \nabla_{\theta} \left\{\log f(x|\hat{\theta})\right\}  f_{\pi}(x) \mathrm{d} x,
		\end{split}
		\end{align}
		for the prior predictive $f_{\pi}(x)$ defined in Section~\ref{section:pb_pseudosamples}. 
		By Assumption \ref{assump:prior_integral} we have $$\int   \nabla_{\theta} \left\{\log f(x|\hat{\theta})\right\}  f_{\pi}(x) \mathrm{d} x \to \alpha_{\theta^*, \pi}$$ as $n\to\infty$, which leads to the following formula for $\tilde{\kappa}_{1,n}$
		\begin{equation*}
		\tilde{\kappa}_{1,n}= I(\theta^*)^{-1/2} \left(\frac{ c \alpha_{\theta^*,\pi}}{n^{1/2} } + \frac{Z}{2 n^{1/2}} + \frac{ \tilde{L}_{1,2}}{n^{1/2}} \right).
		\end{equation*}
		Since the quantities 
		$$
		I_n^{-1/2}\frac{ \sum_{j=n+1}^{n+T} w_j  \nabla\ell(x_j, \hat{\theta})}{n^{1/2}}, \quad\tilde{X}_n - I_n^{-1/2}\frac{ \sum_{j=n+1}^{n+T} w_j  \nabla\ell(x_j, \hat{\theta})}{n^{1/2}}
		$$
		are independent, we get that the third cumulant does not depend on the prior. We therefore get $\tilde{\kappa}_{3,n} = \kappa_{3,n}$ for $\kappa_{3,n}$ defined in the proof of Theorem~\ref{thm:edgeworth_expansion_penalisation}. This completes the proof.
	\end{proof}
	
	\subsection{Edgeworth expansions for  other  methods based on Weighted Likelihood Bootstrap}\label{supplementary:other_methods}
	Our results obtained in Theorem~\ref{thm:edgeworth_expansion_penalisation} and \ref{thm:edgeworth_expansion_pseudo} can easily be extended to obtain Edgeworth expansions for other  methods based on Weighted Likelihood Bootstrap, in particular \cite{newton1991weighted} and \cite{lyddon2018nonparametric}.

	We start by discussing the method proposed by \cite{newton2018weighted} presented in Algorithm~\ref{alg:pb_wbb_newton}. The regularization parameter $\lambda$ plays a similar role to $w_0$ in Algorithm~\ref{alg:pb_general} if $w_0 \in \R$, it is however multiplied by a random variable $\tilde{w}$ drawn from the exponential distribution. When the prior factorizes, \cite{newton2018weighted} alternatively propose to assign a different random weight $\tilde{w}_i \sim \text{Exp}(1)$ to each coordinate $i = 1, \ldots, d$. The penalization term would then be given by
	\begin{equation*}
	\sum_{i=1}^d \lambda\tilde{w}_i \log \pi_i(\theta_i).
	\end{equation*}
	Note that since the regularization parameter stays fixed for all coordinates, this approach is of  different nature to the way we proposed to set $w_0$ in Section~\ref{section:choice_w_0}. In our case the idea is to assign potentially different regularization parameters to different coordinates, accounting for situations when overdispersion or underdispersion is more severe for some coordinates than others.  
	
		\begin{algorithm}[!ht]
		\caption{Weighted Bayesian Bootstrap \citep{newton2018weighted}}\label{alg:pb_wbb_newton}
		\begin{algorithmic}[1]	
			\For {j = 1, \ldots, N} \Comment{\textbf(in parallel)}
			\State{Draw $\tilde{w}^{(j)}, w_1^{(j)}, \ldots, w_{n}^{(j)} \sim\text{Exp}(1)$.}
			\State{ Set $\tilde{\theta}^{(j)} = \argmax_{\theta} \left\{\sum_{i=1}^{n} w_i \log f(x_{i}| \theta) + \lambda \tilde{w}^{(j)} \log \pi(\theta)\right\}$.}
			\EndFor
			\State \Return {Samples $\big\{\tilde{\theta}^{(j)}\big\}_{j=1}^N$. }
		\end{algorithmic}
	\end{algorithm}
	
	Assuming that $\tilde{w}$ is a one-dimensional variable drawn from $\text{Exp}(1)$, we will show that Algorithm~\ref{alg:pb_general} and Algorithm~\ref{alg:pb_wbb_newton} have the same Edgeworth expansion as long as $w_0 = \lambda$. To this end, we use very similar arguments to those used in the proof of Theorem~\ref{thm:edgeworth_expansion_penalisation}. Firstly, note that the Taylor expansion of $I_n^{-1/2} J_{n} n^{1/2}(\tilde{\theta} - \hat{\theta})$ for $\tilde{\theta}$ drawn from Algorithm~\ref{alg:pb_wbb_newton} is given by
	$$
	I_n^{-1/2} \left(S_{1,w,n} + \frac{\lambda \tilde{w} \nabla_{\theta} \log \pi(\hat{\theta})}{n^{1/2}} - \frac{S_{2,w,n}J_n^{-1}S_{1,w,n}}{n^{1/2}} \right) +I_n^{-1/2}\frac{\mu_{3} \left( J_{n}^{-1} S_{1,w,n}\right)^2}{2 n^{1/2}}
	$$
	
	Repeating the arguments used in the proofs and Theorems \ref{thm:edgeworth_expansion_penalisation} and \ref{thm:edgeworth_expansion_pseudo} we conclude that $$	\left(S_{1,w,n}, S_{2,w,n}, \frac{1}{n^{1/2}}\lambda \tilde{w} \nabla_{\theta} \log \pi(\hat{\theta})\right)$$ has a valid Edgeworth expansion. Since $\E \lambda \tilde{w} = \lambda$, then the corresponding cumulants $\kappa_{1,n}$ and $\kappa_{3,n}$ are the same as in the proof of Theorem~\ref{thm:edgeworth_expansion_penalisation} for $\lambda = w_0$.

	\begin{algorithm}[!ht]
		\caption{Posterior Bootstrap \citep{lyddon2018nonparametric} (with pseudo-samples from the posterior predictive) }\label{alg:pb_pseudosamples_lyddon}
		\begin{algorithmic}[1]	
			\For {j = 1, \ldots, N} \Comment{\textbf(in parallel)  }
			\State{Draw $\theta_{\text{post}}^{(j)}$ from the Bayesian posterior  $\pi(\theta|x_1, \ldots, x_n)$.}
			\State{Generate posterior pseudo-samples $x_{n+1}^{(j)}, \ldots, x_{n+T}^{(j)} \sim f(\cdot |\theta_{\text{post}}^{(j)})$.}
			\State{Draw $w_1^{(j)}, \ldots, w_{n}^{(j)} \sim\text{Exp}(1)$ and $w_{n+1}^{(j)}, \ldots, w_{n+T}^{(j)} \sim\text{Exp}(T/c)$}
			\State{ Set $\tilde{\theta}^{(j)} = \argmax_{\theta} \left\{\sum_{i=1}^{n} w_i^{(j)} \log f(x_{i}| \theta) + \sum_{i=1}^{T} w_{n+i}^{(j)} \log f(x_{n+i}^{(j)}| \theta)  \right\}$.}
			\EndFor
			\State \Return {Samples $\big\{\tilde{\theta}^{(j)}\big\}_{j=1}^N$. }
		\end{algorithmic}
	\end{algorithm}
	
	As for the method proposed by \cite{lyddon2018nonparametric}, which we present in Algorithm~\ref{alg:pb_pseudosamples_lyddon}, its second order Edgeworth expansion was discussed in Section~\ref{section:pb_pseudosamples}. To make the statement more precise,
	let $\tilde{\theta}$ be drawn from Algorithm~\ref{alg:pb_pseudosamples_lyddon}, and  $\theta$ be drawn from  standard Weighted Likelihood Bootstrap \citep{newton1994approximate}. Then up to the term of order $o\left(n^{-1/2} \right)$ the Edgeworth expansion of $I_n^{-1/2} J_{n} n^{1/2}(\tilde{\theta} - \hat{\theta})$ is the same as the Edgeworth expansion $I_n^{-1/2} J_{n} n^{1/2}(\theta - \hat{\theta})$, for any choice of fixed parameters $c$ and $T$.
	
	To prove the above statement, we follow the steps from the proof of Theorem~\ref{thm:edgeworth_expansion_pseudo} in Supplementary Material \ref{supplementary:psudosamples_proof}. However, instead of the calculation in \eqref{eq:prior_predictive_calculation}, in case of Algorithm~\ref{alg:pb_pseudosamples_lyddon} we have
	\begin{align*}
	\begin{split}
	\E \left\{\sum_{j=n+1}^{n+T} w_j  \nabla_{\theta}\ell(x_j, \hat{\theta}) \right\} = & \E_{\pi(\theta|x_{1:n})} \sum_{j=n+1}^{n+T}\left[\E_{f(x_j|\theta)} \left\{ w_j  \nabla_{\theta}\ell(x_j, \hat{\theta}) \Big| \theta \right\}\right]\\
	= &  c  \int   \nabla_{\theta} \left\{\log f(x|\bar{\theta})\right\}  \left\{\int f(x|\theta) \pi(\theta|x_{1:n}) \mathrm{d} \theta \right\} \mathrm{d} x = o_P(1).
	\end{split}
	\end{align*}
	The last equality follows from the fact that the posterior concentrates around the maximum likelihood estimator as $n \to \infty$.
	
	\section{Proofs of theoretical results presented in Section~\ref{section:hierarchical_models}}\label{supplementary:hierarchical}
	\subsection{Proof of Proposition \ref{prop:wlb_lambda_posterior}}\label{supplement:proof_lambda_finite_K}
	\begin{proof}
		Fix $\epsilon >0$. Recall that we denote by $P^{HB,p}(\cdot|x_{1:K, 1:n_k})$ the distribution of samples drawn from Algorithm~\ref{alg:pb_hierarchical_gibbs}. Let $p^{HB,p}(\cdot|x_{1:K, 1:n_k})$ denote its corresponding density. Let $\tilde{\theta} = (\tilde{\theta}_1, \ldots, \tilde{\theta}_K)$ and $\theta^* =(\theta^*_1, \ldots, \theta^*_K)$, and let $p(\lambda|\tilde\theta)$ be the Bayesian posterior of $\lambda$ given $\tilde{\theta}$ in model \ref{model:hierarchical}.  We have
		\begin{align}
		\nonumber p^{HB,p}(\tilde{\lambda}\big|x_{1:K, 1:n_k}) = & \int p(\tilde{\lambda}|\tilde{\theta}) p^{HB,p}(\tilde{\theta}\big|x_{1:K, 1:n_k}) \mathrm{d} \tilde{\theta} \\ 
		\label{eq:integral_A_epsilon}	= &  \int_{A_{\epsilon}} p(\tilde{\lambda}\big|\tilde{\theta}) p^{HB,p}(\tilde{\theta}\big|x_{1:K, 1:n_k}) \mathrm{d} \tilde{\theta} + \int_{A_{\epsilon}^c} p(\tilde{\lambda}\big|\tilde{\theta}) p^{HB,p}(\tilde{\theta}\big|x_{1:K, 1:n_k}) \mathrm{d} \tilde{\theta} 
		\end{align}
		where $A_{\epsilon} = \{\tilde{\theta}:\|\tilde\theta - \theta^*\|\leq \epsilon\}$.
		Since $\tilde{\theta}$ under $P^{HB,p}$ converges weakly to $\delta_{\theta^*}$, then by the continuity assumption for any $\lambda$ the first integral in \eqref{eq:integral_A_epsilon} converges to  $p(\tilde{\lambda}\big|\theta^*)$. Since $\sup_{\lambda\in \Lambda, \theta\in \Theta} p(\lambda|\theta) < \infty$, the second integral converges to 0. Finally we get that
		$p^{HB,p}(\tilde{\lambda}\big|\tilde{\theta}) \to p(\tilde{\lambda}\big|\theta^*)$  as $n  \to \infty$.
		By similar arguments the Bayesian posterior $p\left(\lambda\big|x_{1:K, 1:n_k}\right)$ has the same limiting distribution.
	\end{proof}

	\subsection{Additional information  for Section \ref{section:inference_lambda} and proof of Theorem \ref{thm:hierarchical_large_K_result}}\label{supplement:proof_lambda_K_to_infty}
	
	Recall that in Section~\ref{section:inference_lambda} we considered the following hierarchical model
	\begin{gather*}
	\begin{aligned}
	x_{k,1},\ldots x_{k,n_k} & \sim f(x_{k,i}| \theta_k),  &  k &= 1, \ldots, K_n,\\
	\theta_k  & \sim g(\theta_k|\lambda),  &k &= 1, \ldots, K_n,\\
	\lambda &\sim p(\lambda).
	\end{aligned}
	\end{gather*}
	We proposed using for this model a version of Posterior Bootstrap, in which $\tilde{\lambda}$ is, as opposed to Algorithm \ref{alg:pb_hierarchical_gibbs}, obtained via bootstrapping. We present this method in Algorithm \ref{alg:pb_hierarchical_large_K}.	
	
	\begin{algorithm}[!ht]
		\caption{Posterior Bootstrap for hierarchical models via prior penalization 
			under the large $K$ scenario, returns samples from $P^{HB,pp}(\cdot|x_{1:K, 1:n_k})$}\label{alg:pb_hierarchical_large_K}
		\begin{algorithmic}[1]	
			\For {$k = 1, \ldots, K_n$} 
			\State{ Set $\hat{\theta}_k = \arg \max_{\theta_k}\sum_{i=1}^{n_k} \log f(x_{k,i}| \theta_k)$.}
			\EndFor
			\For {j = 1, \ldots, N} \Comment{\textbf(in parallel)  }
			\State{Draw $z_1^{(j)}, \ldots, z_{K_n}^{(j)} \sim\text{Exp}(1)$.}
			\State{Set $\bar{\lambda}^{(j)}  =  \arg \max_{\lambda} \left\{ \sum_{k=1}^{K_n} z_k^{(j)} \log g(\hat{\theta}_k| \lambda) + w_{g}^T \log p\left(\lambda\right)\right\}$.}
			\For {$k = 1, \ldots, K_n$} 
			\State{Draw $w_1^{(j)}, \ldots, w_{n_k}^{(j)} \sim\text{Exp}(1)$.}
			\State{ Set $\tilde{\theta}_k^{(j)} = \arg \max_{\theta_k} \left\{ \sum_{i=1}^{n_k} w_i^{(j)} \log f(x_{k,i}| \theta_k) + w_{0,[k]}^T \log g\left(\theta_k|\bar{\lambda}^{(j)}\right) \right\}$.}
			\EndFor
			\State{Draw $v_1^{(j)}, \ldots, v_{K_n}^{(j)} \sim\text{Exp}(1)$.}
			\State{Set $\tilde{\lambda}^{(j)} =  \arg \max_{\lambda} \left\{ \sum_{k=1}^{K_n} v_k^{(j)} \log g(\tilde{\theta}_k^{(j)}| \lambda) + w_{g}^T \log p\left(\lambda\right) \right\}$.}
			\EndFor
			\State \Return {Samples $\Big\{ \left(\tilde{\lambda}^{(j)}, \tilde{\theta}_1^{(j)}, \ldots, \tilde{\theta}_K^{(j)}\right) \Big\}_{j=1}^N$. }
		\end{algorithmic}
	\end{algorithm}

	Before proceeding with the proof, we present the regularity conditions for Theorem \ref{thm:hierarchical_large_K_result}.
	\setcounter{assumption}{7}
	\begin{assumption}\label{assump_hierarchical:data} \normalfont {(Data-generating mechanism)}
		The mechanism of generating data $x_{1:K_n, 1:n_k}$  follows~\eqref{eq:hierarchical_data_generating}. Furthermore, $\theta_1, \ldots, \theta_{K_n}$ have a density $q^*$ with respect to the Lebesgue measure, defined on a compact subset of $\Theta$.
	\end{assumption}
	
	\begin{assumption}\label{assump_hierarchical:log_lik} \normalfont {(Log likelihood functions)}. The log likelihood function $\ell(x, \theta) = \log f(x|\theta) $ is  measurable and bounded from above, with $E_{p_k^*} \big|\ell(X, \theta) \big| < \infty$ for all $\theta \in \Theta$ and each $k = 1, \ldots, K_n$ almost surely with respect to $Q^*$. Moreover, the likelihood function $\ell_g(\theta, \lambda) = \log g(\theta|\lambda) $ is  measurable and bounded from above,  with $E_{Q^*} \big|\ell_g(\Theta^*, \lambda) \big| < \infty$ for all $\lambda \in \Lambda$.
	\end{assumption}
	
	\begin{assumption}\label{assump_hierarchical:identifiability_lambda}\normalfont(Identifiability of $\lambda$) There  exists a unique maximizing parameter value $\lambda^*$
		$$	\lambda^*=	\arg\max_{\lambda} E_{Q^*}\left\{\ell_g(\Theta^*|\lambda) \right\},
		$$
		and further for all $\delta > 0$ there exists an $\epsilon >0$ such that 
		\begin{equation*}
		\lim \inf_{n \to \infty} Q^* \left[ \sup_{\|\lambda - \lambda^*\| \geq \delta} \frac{1}{K_n} \sum_{k=1}^{K_n} \left\{ \ell_g (\theta_k^*, \lambda) - \ell_g (\theta_k^*, \lambda)\right\} \leq - \epsilon \right] =1.
		\end{equation*}
		
	\end{assumption}
	
	\begin{assumption}\label{assump_hierarchical:identifiability_theta}\normalfont (Identifiability of $\theta$). For each $k=1, \ldots,K_n$ there  exists a unique maximizing parameter value $\theta_k^* = \arg \max_{\theta \in \Theta} \E_{p_k^*} \left\{\ell(\theta, x)\right\},
		$ and further for all $\delta > 0$ there exists an $\epsilon >0$ such that 
		\begin{equation*}
		\lim_{n \to \infty} P_k^* \left[\sup_{\|\theta_k - \theta_k^*\| \geq \delta} \frac{1}{n_k} \sum_{i=1}^{n_k} \left\{ \ell (x_{k,i},\theta_k) - \ell (x_{k,i}, \theta_k^*)\right\} \leq - \epsilon \right] =1.
		\end{equation*}
	\end{assumption}
	
	\begin{assumption}\label{assump_hierarchical:smoothness_lambda}\normalfont(Smoothness of $l_g$). There is an open ball $B$ containing $\lambda^*_{\theta^*}$ such that  $\ell_g(\theta, \lambda)$ is three times continuously differentiable with respect to $\lambda \in B$ for almost all $\theta \in \Theta$. Furthermore, there exist measurable functions $G_1, G_{2}$ and  $G_{3}$ such that for $\lambda \in B$
		\begin{align*}
		\Big|\frac{\partial \ell_g(\theta, \lambda)}{\partial \lambda_j}  \Big|\leq G_{1}(\theta) \qquad &\E_{q*} \left\{G_1(\Theta^*)^4\right\} < \infty,\\
		\Big|\frac{\partial^2 \ell_g(\theta,\lambda)}{\partial \lambda_j \partial \lambda_l}  \Big|\leq G_{2}(\theta)\qquad &  \E_{q*} \left\{G_{2}(\Theta^*)^4 \right\}< \infty,\\
		\Big|\frac{\partial^3 \ell_g(\theta, \lambda)}{\partial \lambda_j \partial \lambda_l \partial \lambda_m}  \Big|\leq G_{3}(\theta) \qquad &  \E_{q*} \left\{G_{3}(\Theta^*)\right\} < \infty.
		\end{align*}
	\end{assumption}
	
	\begin{assumption}\label{assump_hierarchical:smoothness_lambda_wrt_theta}\normalfont(Smoothness of $\ell_g$ with respect to $\theta$). For any $\lambda \in \Lambda$, the function $\ell_g(\theta, \lambda)$ is  measurable as a function of $\theta$, and three times continuously differentiable with respect to $\theta$. Besides, $\ell_g(\theta, \lambda)$ is jointly bounded as a function of $\lambda\in \Lambda, \theta \in \Theta$. Furthermore, there exist $ L_1, L_2, L_3 >0$ such that for any $\lambda \in \Lambda$ and any $\theta_1, \theta_2 \in \Theta$ 
		\begin{align*}
		\begin{split}
		|\ell_g( \theta_1, \lambda) - \ell_g(\theta_2, \lambda)  | & \leq L_1 \|\theta_1 - \theta_2\|,\\
		\|\nabla_{\lambda}\ell_g(\theta_1, \lambda) - \nabla_{\lambda}\ell_g(\theta_2, \lambda)  \| & \leq L_2 \|\theta_1 - \theta_2\|,\\
		\|D^2_{\lambda}\ell_g(\theta_1, \lambda) - D^2_{\lambda}\ell_g(\theta_2, \lambda) \| & \leq L_3 \|\theta_1 - \theta_2\|.
		\end{split}
		\end{align*}
	\end{assumption}
	
	\begin{assumption}\label{assump_hierarchical:smoothness_theta}\normalfont(Smoothness of $\ell$). The function  $\ell(x, \theta)$ is three times continuously differentiable with respect to $\theta \in \Theta$ almost surely on $x$.  
		We assume that there exist measurable functions $H_1, H_{2}$ and  $H_{3}$ such that for $\theta \in \Theta$ 
		\begin{align*}
		\Big|\frac{\partial \ell(x, \theta)}{\partial \theta_j}  \Big|\leq H_{1}(x) \quad &\E_{Q^*}\left[ \max_{1\leq k \leq K_n} \E_{p_k*}\left\{ H_1(X)^{4}\right\} \right] < \infty,\\
		\Big|\frac{\partial^2 \ell(x, \theta)}{\partial \theta_j \partial \theta_l}  \Big|\leq H_{2}(x)\quad &  \E_{Q^*}\left[ \max_{1\leq k \leq K_n} \E_{p_k*}\left\{ H_{2}(X)^{4}\right\} \right]< \infty,\\
		\Big|\frac{\partial^3 \ell(x, \theta)}{\partial \theta_j \partial \theta_l \partial \theta_m}  \Big|\leq H_{3}(x) \quad &  \E_{Q^*}\left[ \max_{1\leq k \leq K_n} \E_{p_k*}\left\{ H_{3}(X)\right\} \right] < \infty.
		\end{align*}
	\end{assumption}
	
	\begin{assumption}\label{assump_hierarchical:positive_definite_lambda}
		\normalfont (Positive definiteness of information matrices for $\lambda$). The matrices $I_g(\lambda)$ and $J_g(\lambda)$ defined in \eqref{eq:hierarchical_info_matrices} are positive definite for $\lambda \in B$ with all elements finite.
	\end{assumption}
	
	\begin{assumption}\label{assump_hierarchical:positive_definite_theta}
		\normalfont (Positive definiteness of information matrices for $\theta$ and linear independence of partial derivatives). The matrices $I(\theta)$ and $J(\theta)$ defined in \eqref{eq:I_J_matrices} are positive definite for $\theta \in \Theta$ with all elements finite. Furthermore, $\det J(\theta)$ and $\det I(\theta)$  are bounded away from $0$ on any compact subset of $\Theta$. Besides, the first and second partial derivatives
		$$
		\left\{\frac{\partial \ell(x, \theta)}{\partial \theta_p}, \frac{\partial^2 \ell(x, \theta)}{\partial \theta_r \partial \theta_s} \right\}, \quad 1 \leq p \leq d, 1 \leq r \leq s \leq d
		$$
		are linearly independent as functions of $x$ at any $\theta\in \Theta$.
	\end{assumption}
	
	\begin{assumption}\label{assump_hierarchical:prior_smoothness}\normalfont {(Smoothness of log prior density).} The function $\log p(\lambda)$ is  measurable, upper-bounded on $\lambda$, and three times continuously differentiable at $\lambda^*$.
	\end{assumption}

The strategy of the proof is as follows. We first prove consistency of $\tilde{\lambda}$, and then its asymptotic normality. The key element in establishing consistency of $\tilde{\lambda}$ is obtaining the upper bound on the rate of convergence of   $\tilde{\theta}_k$ to $\theta_k^*$. We obtain this upper bound in Lemma \ref{lemma:auxiliary_hierarchical}, which relies largely on the multivariate  Berry--Esseen theorem with explicit constants, proved by \cite{raivc2019multivariate}.
\begin{lemma}\label{lemma:auxiliary_hierarchical}
	Let $\epsilon_n$ be a deterministic  sequence depending on $n$ such that 
	$\epsilon_n = o(1)$ and 
	\begin{align*}
	\epsilon_n n^{1/2}/(\log n)^{1/2} \to \infty
	\end{align*}  as $n \to \infty.$
	Then under Assumptions \ref{assump_hierarchical:data}--\ref{assump_hierarchical:prior_smoothness} we have
	\begin{equation*}
	\sup_{1 \leq k \leq K_n} P^{HB, pp} \left(\|\tilde{\theta}_k - \theta_k^*\| > \epsilon_n \big|{x}_{1:K_n, 1:n_k}\right)  = O \left(n^{-1/2}\right).
	\end{equation*}
	where $n = \min_{1 \leq k \leq K_n} n_k$.
\end{lemma}
\begin{proof}
	We fix $k \leq K_n$. In the remaining part of the proof for simplicity of notation we write $x = {x}_{1:K_n, 1:n_k}$ and we drop the subscript $k$. We have
	\begin{equation*}
	P^{HB, pp}\left(\|\tilde{\theta} -  \theta^*\| >\epsilon_n \big|x\right)  \leq P^{HB, pp}\left(\|\tilde{\theta} -  \hat{\theta}\| > \epsilon_n/2 \big|x\right) + \text{pr} \left( \|\hat{\theta} -  \theta^*\| > \epsilon_n/2\right).
	\end{equation*}
	
	We first deal with the term	$P^{HB, pp} \left(\|\tilde{\theta} - \hat{\theta}\| > \epsilon_n \big|x \right)$.  To this end we expand $n^{1/2}\left(\tilde{\theta} -  \hat{\theta} \right)$ as follows. We define
	\begin{equation*}
	\tilde{L}_n(\theta)  = \sum_{j=1}^n w_j \ell(x_i, \theta) + w_0^T \log g(\theta|\bar{\lambda}),
	\end{equation*}
	so that
	\begin{equation*}
	0 = \nabla_{\theta} \tilde{L}_n(\tilde{\theta})  = \nabla_{\theta} \tilde{L}_n(\hat{\theta}) + D^2_{\theta}\tilde{L}_n(\check{\theta})(\tilde{\theta} - \hat{\theta})  
	\end{equation*}
	for some $\check{\theta}$  between $\hat{\theta}$ and $\tilde{\theta}$. Consequently
	\begin{equation}\label{eq:grad_L_n_hierarchical}
	I_n^{-1/2}\tilde{J}_{w,n}(\check\theta) n^{1/2}(\tilde{\theta} - \hat{\theta}) =  	I_n^{-1/2}\frac{\sum_{j=1}^n (w_j-1) \nabla_{\theta}\ell(x_i, \hat\theta) + w_0 \cdot \nabla_{\theta}\log g(\hat\theta|\bar{\lambda})}{n^{1/2}},
	\end{equation}
	where $\tilde{J}_{w,n}(\check\theta) = -n^{-1}D^2_{\theta}\tilde{L}_n(\check\theta)$, and for $I_n$ we follow the standard notation used elsewhere in the paper. Let $Z_{w,n}$ denote a random variable following the normal distribution 
	$$
	N \left\{\tilde{J}_{w,n}(\check{\theta})^{-1} \frac{ w_0 \cdot  \nabla_{\theta}\log g(\hat\theta|\bar{\lambda})}{n^{1/2}}, \tilde{J}_{w,n}(\check{\theta})^{-1} I_{n} \tilde{J}_{w,n}(\check{\theta})^{-1} \right\}.
	$$
	Observe that 
	\begin{align*}
	P^{HB,pp}\left(\|\tilde{\theta} -  \hat{\theta}\| > \epsilon_n/2 \big| x \right) 
	= & P^{HB,pp}\left(n^{1/2}\|\tilde{\theta} -  \hat{\theta}\| > \epsilon_n n^{1/2}/2 \big|x\right) \\ 
	\leq & \sup_{A \subseteq \Theta} \Big|  P^{HB,pp} \left\{n^{1/2}\left(\tilde{\theta} -  \hat{\theta}\right)\in A \big|x\right\}  - \text{pr} \left(Z_{w,n} \in A\right)\Big| \\
	& \qquad + 2\text{pr} \left(\|Z_{w,n}\| > \epsilon_n n^{1/2}/2 \right),
	\end{align*}
	where the supremum in the above equation is taken over all measurable convex sets $A$.
	Let $Z$ follow the standard $d$-variate normal distribution and note that by \eqref{eq:grad_L_n_hierarchical} we have 
	\begin{align*}
	& \sup_{A \subseteq \Theta} \Big|  P^{HB,pp} \left\{n^{1/2}\left(\tilde{\theta} -  \hat{\theta} \right)\in A \Big|x\right\}  - \text{pr} \left(Z_{w,n} \in A \right)\Big| \\
	= & \sup_{A \subseteq \Theta} \Bigg|  P^{HB,pp} \left\{I_n^{-1/2}\tilde{J}_{w,n}(\check\theta) n^{1/2}\left(\tilde{\theta} -  \hat{\theta} \right)\in A \Big|x \right\} \\
	& \qquad \qquad - \text{pr} \left\{Z + I_n^{-1/2} \frac{ w_0 \cdot  \nabla_{\theta}\log g(\hat\theta|\bar{\lambda})}{n^{1/2}} \in A   \right\}\Bigg| \\
	= & \sup_{A \subseteq \Theta} \Bigg|  P^{HB,pp} \left\{I_n^{-1/2}\frac{\sum_{j=1}^n (w_j-1) \nabla_{\theta}\ell(x_i, \hat\theta) + w_0 \cdot  \nabla_{\theta}\log g(\hat\theta|\bar{\lambda})}{n^{1/2}}\in A \Big|x \right\}  \\
	& \qquad \qquad - \text{pr} \left\{Z + I_n^{-1/2} \frac{w_0 \cdot  \nabla_{\theta}\log g(\hat\theta|\bar{\lambda})}{n^{1/2}} \in A   \right\} \Bigg|\\
	= & \sup_{A \subseteq \Theta} \left|  P^{HB,pp} \left\{I_n^{-1/2}\frac{\sum_{j=1}^n (w_j-1) \nabla_{\theta}\ell(x_i, \hat\theta)}{n^{1/2}}\in A  \Big|x\right\}  - \text{pr} \left(Z \in A \right)\right|
	\end{align*}
	To obtain the upper bound on the last expression above, we now apply Theorem 1.1 of \cite{raivc2019multivariate} to the expression
	\begin{equation*}
	I_n^{-1/2}\frac{\sum_{j=1}^n (w_j-1) \nabla_{\theta}\ell(x_i, \hat\theta)}{n^{1/2}}.
	\end{equation*}
	We get that for a constant $c_0 = 42d^{1/4} + 16$ and for all measurable convex sets $A \subseteq \Theta$ 
	\begin{align*}
	& \sup_{A \subseteq \Theta} \Big|  P^{HB,pp} \left\{n^{1/2}\left(\tilde{\theta} -  \hat{\theta} \right)\in A\right\}  - \text{pr} \left(Z_{w,n} \in A \right)\Big|  \\
	\leq& \frac{c_0}{n^{1/2}} \frac{1}{n}\sum_{j=1}^n \E_{w_{1:n}} \left\{\|(w_j-1) I_n^{-1/2}\nabla_{\theta}\ell(x_i, \hat\theta)\|^3\right\}\\
	= & \frac{2c_0}{n^{1/2}} \frac{1}{n}\sum_{j=1}^n \| I_n^{-1/2}\nabla_{\theta}\ell(x_i, \hat\theta)\|^3.
	\end{align*}
	Note that by Assumption \ref{assump_hierarchical:smoothness_theta} we have
	$$n^{-1}\sum_{j=1}^n \| I_n^{-1/2}\nabla_{\theta}\ell(x_i, \hat\theta)\|^3\leq \frac{ \|I_n^{-1/2}\|}{n} \sum_{i=1}^n \|H_3(x_i)\|,
	$$
	which is bounded by a constant that does not depend on $\theta^*$.
	
	To bound $\text{pr} \left(Z_{w,n} > \epsilon_n n^{1/2}/2 \right)$, note that on a set such that 
	\begin{equation}\label{eq:J_w_n_event_of_interest}
	\|\tilde{J}_{w,n}^{-1}(\check{\theta}) I_{n} \tilde{J}_{w,n}^{-1}(\check{\theta})\| \leq c_1, \quad 
	\left\|\tilde{J}_{w,n}(\check{\theta})^{-1} \frac{ w_0 \cdot  \nabla_{\theta}\log g(\hat\theta|\bar{\lambda})}{n^{1/2}}\right\| \leq c_2
	\end{equation}
	we have 
	\begin{align*}
	\text{pr} \left(\|Z_{w,n} \|> \epsilon_n n_k^{1/2}/2 \right) \leq  d \left\{1-\Phi_1 \left(c_5 \epsilon_n n^{1/2}\right)\right\},
	\end{align*}
	for some constant $c_5 >0$,  where $\Phi_1$ denotes a cumulative distribution function of a univariate standard normal distribution. We use the fact that 
	\begin{equation}\label{eq:exponential_tail_property}
	1-\Phi_1(x) \leq \frac{1}{x (2\pi)^{1/2}} \exp \left(-x^2/2 \right), \quad x \in \R
	\end{equation}
	so that
	\begin{align}\label{eq:hierarchical_tile_theta_part_2}
	1-\Phi_1\left(c_5 \epsilon_n n^{1/2} \right) & \leq   n^{-1/2}\frac{\exp \left(-c_5^2 \epsilon_n^2 n/2 - \log \epsilon_n \right)}{ c_5  (2\pi)^{1/2}}.
	\end{align}
	By the assumption that $\epsilon_n n^{1/2} \to \infty$ as $n \to \infty$, we have $\epsilon_n > n^{-1/2}$ for large enough $n$. Thus as $n$ goes to infinity
	$$
	c_5^2 \epsilon_n^2 n/2 + \log \epsilon_n > c_5^2 \epsilon_n^2 n/2 - (\log n)/2 =  \log n/2 \left\{c_5^2 \epsilon_n^2 n/\log n - 1  \right\} \to \infty,
	$$
	which implies that the right hand side of \eqref{eq:hierarchical_tile_theta_part_2} is $o\left(n^{-1/2}\right)$. 
	
	It remains to show that probability of event \eqref{eq:J_w_n_event_of_interest} goes to 1 as $n \to \infty$. Let $C$ be a compact subset of~$\Theta$. Let $ J_{n}(\theta)  =  -n^{-1}\sum_{j=1}^n  D_{\theta}^2 \ell(x_j,\theta)$. We start by observing that for some $c_3, c_4 >0$ we have
	\begin{equation}\label{eq:I_n_J_n_bound}
	Q^* \left(\|I_n\| < c_3, \sup_{\check{\theta} \in C}\|J_n(\check\theta)^{-1}\| < c_3 \right) \leq \frac{c_4}{n},
	\end{equation}
	which follows from Assumptions \ref{assump_hierarchical:smoothness_theta} and \ref{assump_hierarchical:positive_definite_theta}, and the Chebyshev's inequality. We apply once again the 
	the Chebyshev's inequality to show that for some $c_5, c_6 >0$
	$$
	E_{Q^*} \left\{P^{HB,pp} \left(\sup_{\check{\theta} \in C}\|\tilde{J}_{w,n}(\check\theta)^{-1}\| < c_5 \Big| x \right) \right\} \leq \frac{c_6}{n}.
	$$
	This completes the proof of the claim for $ P^{HB,pp}\left(n^{1/2}\|\tilde{\theta} -  \hat{\theta}\| > \epsilon_n n^{1/2}/2 \big|x\right)$.
	
	Using similar reasoning we now find an upper bound for $\text{pr} \left( \|\hat{\theta}_k -  \theta_k^*\| > \epsilon_n/2\right)$. To this end, similarly to the previous case, we define 
	$L_n(\theta)  = \sum_{j=1}^n \ell(x_i, \theta)$ and write
	\begin{equation*}
	0 = \nabla_{\theta} L_n(\hat{\theta})  = \nabla_{\theta} L_n(\theta^*) + D^2_{\theta}L_n(\check{\theta})(\hat{\theta} - \theta^*)  
	\end{equation*}
	for some $\check{\theta}$ between $\hat{\theta}$ and $\theta^*$. Therefore 
	$$
	I(\theta^*)^{-1/2}\tilde{J}_{n}(\check\theta) n^{1/2}(\hat{\theta} - \theta^*) =  I(\theta^*)^{-1/2}n^{-1/2}	\sum_{j=1}^n \nabla_{\theta}\ell(x_i, \theta^*),
	$$
	where
	$\tilde{J}_{n}(\check\theta) = n^{-1}D^2_{\theta}L_n(\check{\theta})$. For $Z_n$ following the normal distribution $N \left\{0, \tilde{J}_{n}(\check{\theta})^{-1} I(\theta^*) \tilde{J}_{n}(\check{\theta})^{-1} \right\}$ we have
	\begin{align*}
	\text{pr}\left(\|\theta^* -  \hat{\theta}\| > \epsilon_n/2  \right) 
	= & \text{pr}\left(n^{1/2}\|\hat{\theta} - \theta^*\| > \epsilon_n n^{1/2}/2 \right) \\ 
	\leq & \sup_{A \subseteq \Theta} \Big|  \text{pr} \left\{n^{1/2}\left(\hat{\theta} - \theta^*\right)\in A \right\}  - \text{pr} \left(Z_{n} \in A\right)\Big| \\
	& \qquad + 2\text{pr} \left(\|Z_{n}\| > \epsilon_n n^{1/2}/2 \right).
	\end{align*}
	Let $Z'$ follow the $d$-variate standard normal distribution. Then for measurable convex sets $A \subseteq \Theta$ we have
	\begin{align*}
	& \sup_{A \subseteq \Theta} \Big|  \text{pr} \left\{n^{1/2}\left(\hat{\theta} - \theta^*\right)\in A \right\}  - \text{pr} \left(Z_{n} \in A\right)\Big| \\
	= & \sup_{A \subseteq \Theta} \Big|  \text{pr} \left\{I(\theta^*)^{-1/2} \tilde{J}_{n}(\check{\theta})n^{1/2}\left(\hat{\theta} - \theta^*\right)\in A \right\}  - \text{pr} \left(I(\theta^*)^{-1/2} \tilde{J}_{n}(\check{\theta}) Z_{n} \in A\right)\Big| \\
	= & \sup_{A \subseteq \Theta} \left|  \text{pr} \left\{I(\theta^*)^{-1/2}n^{-1/2}	\sum_{j=1}^n \nabla_{\theta}\ell(x_i, \theta^*)\in A \right\}  - \text{pr} \left(Z' \in A\right)\right|\\
	\leq &  \frac{2c_0}{n^{1/2}} \frac{1}{n}\sum_{j=1}^n \E_{p_k^*} \left\{\| I(\theta^*)^{-1/2}\nabla_{\theta}\ell(x_i, \theta^*)\|^3 \right\} \\
	\leq& \frac{2c_0}{n^{1/2}}  \E_{p_k^*}\left\{\| I(\theta^*)^{-1/2} \nabla_{\theta} \ell(X, \theta^*)\|^3 \right\}, 
	\end{align*}
	where the second last equality is again obtained by applying Theorem 1.1 of \cite{raivc2019multivariate}.  We observe that $\E_{p_k^*}\left\{\| I(\theta^*)^{-1/2} \nabla_{\theta} \ell(X, \theta^*)\|^3 \right\}$ is bounded by a constant independent from $\theta^*$ by Assumption~\ref{assump_hierarchical:smoothness_theta}.
	To bound  $\text{pr} \left(\|Z_{n}\| > \epsilon_n n^{1/2}/2 \right)$ by $o\left(n^{-1/2}\right)$, we again use property \eqref{eq:exponential_tail_property} of the tail of the normal distribution, and equation \eqref{eq:I_n_J_n_bound}.
	
	Finally, since the constants appearing in the proof do not depend on $k$, we get that 
	\begin{equation*}
	\sup_{1 \leq k \leq K_n} P^{HB, pp} \left(\|\tilde{\theta}_k - \theta_k^*\| > \epsilon_n \big|{x}_{1:K_n, 1:n_k}\right)  = O \left(n^{-1/2}\right),
	\end{equation*}
	as required.
\end{proof}

\subsubsection*{Proof of Theorem \ref{thm:hierarchical_large_K_result}}
		\begin{proof}
		Let $\epsilon_n$ be a deterministic sequence depending on $n$ such that $\epsilon_n n^{1/2}/(\log n)^{1/2}\to \infty$ as $n \to \infty$ and  $\epsilon_n K_n^{1/2} = o(1)$, which is possible by assumption  that $K_n = O\left(n^{1/2}/\log n\right)$.  Consider  a set $\Omega_n$ such that on $\Omega_n$ we have $\sup_{1 \leq k \leq K_n}\|\tilde{\theta}_k -  \theta_k^*\| \leq \epsilon_n$. By Lemma \ref{lemma:auxiliary_hierarchical} we get that 
		\begin{equation*}
		P^{HB,pp} \left( \Omega_n\right) \leq K_n \sup_{1 \leq k \leq K_n} P^{HB,pp} \left( \|\tilde{\theta}_k -  \theta_k^*\| \leq \epsilon_n\right) \leq  K_n O \left(n_k^{-1/2}\right)= o(1).
		\end{equation*}
		
		We fix $\delta>0$  and choose $n$ so large that $L_1\epsilon_n < \epsilon/3$, where $\delta$ and $\epsilon$ are like in Assumption \ref{assump_hierarchical:identifiability_lambda}, and write 
		\begin{align}
		&\frac{1}{K_n} \sum_{k=1}^{K_n} \left\{v_k \ell_g (\theta_k^*,\tilde{\lambda}) - v_k\ell_g (\theta_k^*, \lambda^*) \right\} \nonumber \nonumber \\
		= & \frac{1}{K_n} \sum_{k=1}^{K_n} v_k\left\{ \ell_g ( \theta_k^*, \tilde{\lambda}) - \ell_g ( \tilde{\theta}_k,\tilde{\lambda}) \right\} + \frac{1}{K_n} \sum_{k=1}^{K_n}v_k \left\{ \ell_g ( \tilde{\theta}_k,\tilde{\lambda}) - \ell_g ( \tilde{\theta}_k, \lambda^*)\right\} \label{eq:hierarchical_proof_three_expressions_1} \\
		& \qquad +  \frac{1}{K_n} \sum_{k=1}^{K_n} v_k\left\{ \ell_g ( \tilde{\theta}_k, \lambda^*) - \ell_g ( \theta_k^*, \lambda^*) \right\} \label{eq:hierarchical_proof_three_expressions_2}.
		\end{align}
		By Assumption \ref{assump_hierarchical:smoothness_lambda_wrt_theta} on $\Omega_n$ we have
		\begin{equation*}
		\frac{1}{K_n} \sum_{k=1}^{K_n} \left\{ \ell_g ( \theta_k^*, \tilde{\lambda}) - \ell_g ( \tilde{\theta}_k,\tilde{\lambda}) \right\}> -\epsilon/3, \quad \frac{1}{K_n} \sum_{k=1}^{K_n} \left\{ \ell_g ( \tilde{\theta}_k, \lambda^*) - \ell_g ( \theta_k^*, \lambda^*) \right\} > -\epsilon/3.
		\end{equation*}
		Thus by Theorem 13 of Appendix A   of \cite{newton1991weighted} the first and the last expression in equations \eqref{eq:hierarchical_proof_three_expressions_1} -- \eqref{eq:hierarchical_proof_three_expressions_2} are greater than $-\epsilon/3$ with probability going to 1. As for the second expression, by definition of $\tilde{\lambda}$
		$$
		\frac{1}{K_n} \sum_{k=1}^{K_n}v_k \left\{ \ell_g ( \tilde{\theta}_k,\tilde{\lambda}) - \ell_g ( \tilde{\theta}_k, \lambda^*)\right\} + \frac{1}{K_n} \left\{w_g^T \log p(\tilde{\lambda}) -  w_g^T \log p(\lambda^*)\right\} \geq 0.
		$$
		Thus, by Assumption \ref{assump_hierarchical:prior_smoothness} for large enough $K_n$ we have
		$$
		\frac{1}{K_n} \sum_{k=1}^{K_n}v_k \left\{ \ell_g ( \tilde{\theta}_k,\tilde{\lambda}) - \ell_g ( \tilde{\theta}_k, \lambda^*)\right\} \geq -\epsilon/3.
		$$
		To summarize, the above reasoning and Assumption \ref{assump_hierarchical:identifiability_lambda} yield $\|\tilde{\lambda} - \lambda^*\| < \delta$ with probability going to 1, which completes the proof of consistency. Additionally, by Assumption \ref{assump_hierarchical:identifiability_lambda} we also have $\| \lambda^*- \hat{\lambda}_{\theta^*}\|= o_P(1)$, so in particular $\|\tilde{\lambda} - \hat{\lambda}_{\theta^*}\| = o_P(1)$. This property is useful in the proof of asymptotic normality, which we present below.
		
		We define $L_{K_n}(\lambda) = \sum_{k=1}^{K_n} v_k \ell_g(\tilde\theta_k, \lambda) + w_g^T \log p (\lambda)$. We can write $L_{K_n}(\lambda)$ as 
		\begin{align*}
		L_{K_n}(\lambda) = \sum_{k=1}^{K_n} v_k \ell_g(\theta_k^*, \lambda)+ w_g^T \log p (\lambda) + h_{K_n}(\lambda),
		\end{align*} 
		where 
		\begin{equation*}
		h_{K_n}(\lambda) = \sum_{k=1}^{K_n} v_k \left\{\ell_g(\tilde\theta_k, \lambda)- \ell_g(\theta_k^*, \lambda) \right\}.
		\end{equation*}
		We get the following expansion
		\begin{align*}
		0 = \nabla_{\lambda} L_{K_n}(\tilde\lambda)  = \nabla_{\lambda} L_{K_n}(\hat\lambda_{\theta^*}) + D^2_{\lambda} L_{K_n}(\check\lambda)\left( \tilde\lambda - \hat\lambda_{\theta^*}\right),
		\end{align*}
		for some $\check\lambda$ between $\tilde{\lambda}$ and $\hat\lambda_{\theta^*}$,  which gives
		\begin{equation}\label{eq:hierarchical_result_of_interest}
		K_n^{1/2}\left( \tilde\lambda - \hat\lambda_{\theta^*}\right)  = \left\{-\frac{1}{K_n}D^2_{\lambda} L_{K_n}(\check\lambda) \right\}^{-1} \frac{1}{K_n^{1/2}}\nabla_{\lambda} L_{K_n}(\hat\lambda_{\theta^*}) .
		\end{equation}
		Since $\sum_{k=1}^{K_n}  \nabla_{\lambda} \ell_g(\theta_k^*, \hat\lambda_{\theta^*}) = 0$, we have
		\begin{align*}
		\frac{\nabla_{\lambda} L_{K_n}(\hat\lambda_{\theta^*})}{K_n^{1/2}} = \frac{1}{K_n^{1/2}}\sum_{k=1}^{K_n} (v_k -1) \nabla_{\lambda} \ell_g(\theta_k^*, \hat\lambda_{\theta^*})+ \frac{1}{K_n^{1/2}}\left\{w_g \cdot  \nabla_{\lambda} \log p (\hat\lambda_{\theta^*}) + \nabla_{\lambda} h_{K_n}(\hat\lambda_{\theta^*}) \right\}.
		\end{align*}
		What is more, 
		\begin{align}\label{eq:hierarchical_second_derivative_equation}
		-\frac{D^2_{\lambda} L_{K_n}(\check\lambda)}{K_n}  =  -  \frac{1}{K_n}\sum_{k=1}^{K_n} v_k D^2_{\lambda} \ell_g(\theta_k^*, \check\lambda) - \frac{1}{K_n} \left[ D^2_{\lambda}\left\{ w_g^T  \log p (\check\lambda)\right\} +  D^2_{\lambda}  h_{K_n}(\check\lambda) \right].
		\end{align}
		By Assumptions \ref{assump_hierarchical:smoothness_lambda_wrt_theta} and \ref{assump_hierarchical:prior_smoothness} we have on $\Omega_n$
		\begin{equation}\label{eq:hierarchical_first_derivative_property}
		\frac{1}{K_n^{1/2}}\left\{w_g^T  \nabla_{\lambda} \log p (\hat\lambda_{\theta^*}) + \nabla_{\lambda} h_{K_n}(\hat\lambda_{\theta^*})\right\} = O_P\left(K_n^{-1/2}\right) + O_P\left(K_n^{1/2} \epsilon_n\right) =o_P(1),
		\end{equation}
		where the last equality holds since we assumed that $K_n \epsilon_n = o(1)$, and the probability is considered with respect to $v_{1:K_n}$. We similarly have in \eqref{eq:hierarchical_second_derivative_equation}
		\begin{equation}\label{eq:hierarchical_second_derivative_property}
		\frac{1}{K_n} \left[ D^2_{\lambda}\left\{ w_g^T  \log p (\check\lambda)\right\} +  D^2_{\lambda}  h_{K_n}(\check\lambda) \right] = o_P(1).
		\end{equation}
		By consistency of $\tilde{\lambda}$ and Theorem 13 of Appendix A  of \cite{newton1991weighted} we get that 
		$$- \frac{1}{K_n}\sum_{k=1}^{K_n} v_k D^2_{\lambda} \ell_g(\theta_k^*, \check\lambda)$$
		converges to $J_g(\lambda^*)$ in probability. Therefore by \eqref{eq:hierarchical_second_derivative_property} 
		$- K_n^{-1}D^2_{\lambda} L_{K_n}(\check\lambda)$ also converges to $J_g(\lambda^*)$ in probability. Finally we use the central limit theorem, equation \eqref{eq:hierarchical_first_derivative_property} and Slutsky's theorem to show that \eqref{eq:hierarchical_result_of_interest} converges in distribution to 
		$
		N\left\{0, J_g(\lambda^*)^{-1} I_g(\lambda^*) J_g(\lambda^*)^{-1} \right\}
		$ as required.
	\end{proof}

	\section{Illustrations -- further details}\label{supplementary:illustrations}
	
	\subsection{Additional information for Section~\ref{section:toy_examples}}\label{supplementary:toy_examples}
	Throughout the paper the Kolmogorov--Smirnov dissimilarity is computed using the \texttt{KS.diss} function from the \texttt{provenance} package \citep{vermeesch2016r} in \texttt{R}, while the Bhattacharyya distance is computed using the \texttt{bhattacharyya.dist} function from the \texttt{fpc} package \citep{hennig2015package}. All our computations for the BayesBag algorithm are based on 50 bootstrapped datasets.
		\begin{figure}[!ht]
		\centering
		\includegraphics[width = 0.8\textwidth]{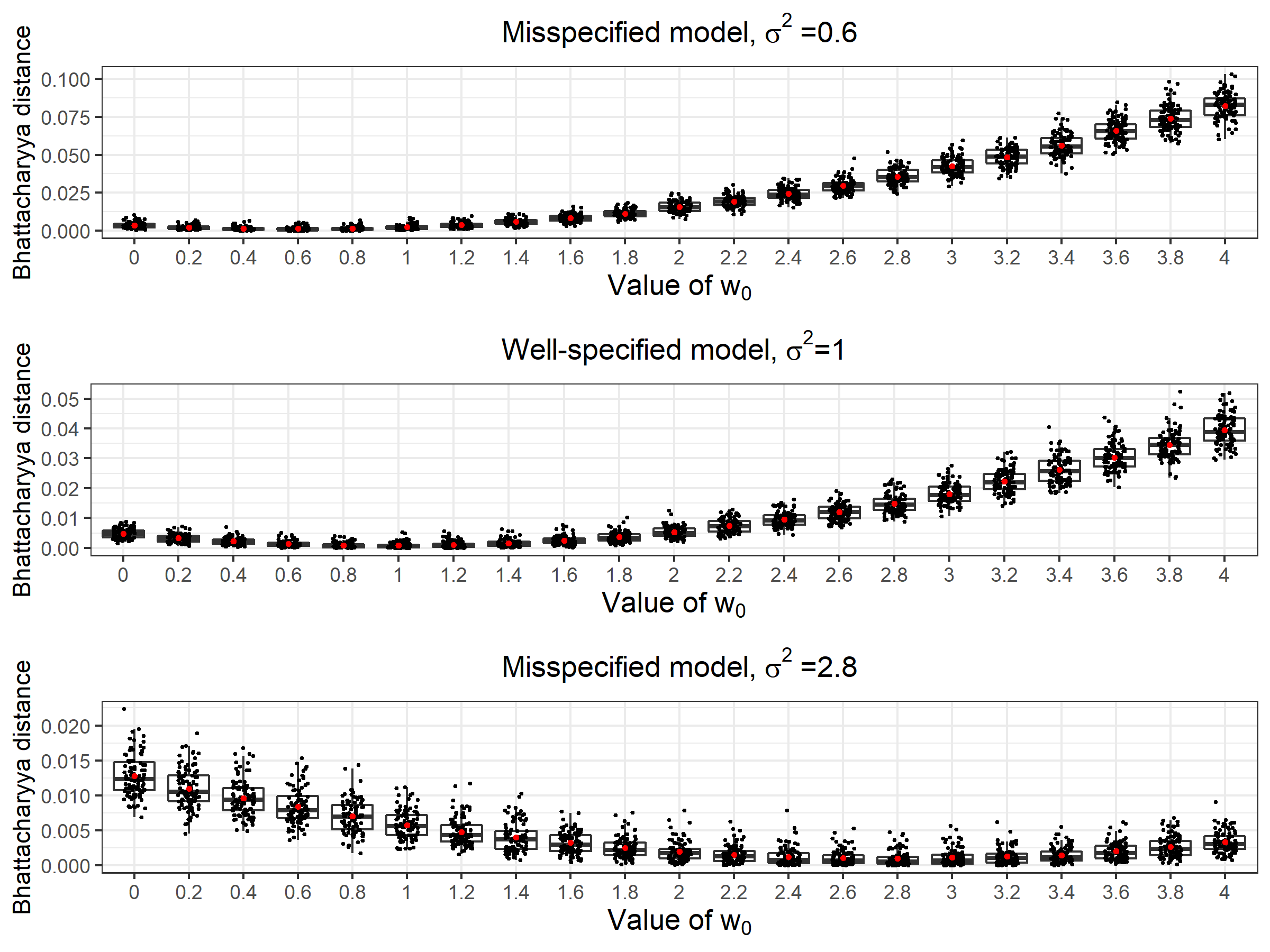}
		\caption{\small Bhattacharryya distance between samples drawn from Posterior Bootstrap for model \eqref{model:toy_normal_univariate}, and samples from the Bayesian posterior for the correct model \eqref{model:toy_normal_univariate_true}. Three scenarios are considered, $\sigma^2 = 0.6$ (top panel), $\sigma^2 =1$ (middle panel), and $\sigma^2 = 2.8$ (bottom panel). Posterior Bootstrap simulations were performed using Algorithm~\ref{alg:pb_general} for a range of values of $w_0$. Every experiment was repeated on 100 simulated datasets with $N = 2000$ samples generated from Algorithm~\ref{alg:pb_general} for each $w_0$. }\label{fig:normal_univariate_Bhat_w0_boxplot}
	\end{figure}

In Figures~\ref{fig:normal_univariate_Bhat_w0_boxplot} -- \ref{fig:normal_univariate_w0_boxplot_comparison} we present additional results for model \eqref{model:toy_normal_univariate}.
 Figure~\ref{fig:normal_univariate_Bhat_w0_boxplot} confirms, using another measure of distance, our conclusion drawn from Figure~\ref{fig:univariate_normal_KS_w_0_boxplot} that it is optimal to set $w_0 = \sigma^2$, which coincides with formula \eqref{eq:setting_w_0}. In Figures \ref{fig:normal_univariate_density_plots} and \ref{fig:normal_univariate_w0_boxplot_comparison} we present also comparisons with BayesBag. For each experiment the size of the bootstrapped datasets was set according to the guidelines for the Gaussian location model provided in Section 2.2.1 of \cite{huggins2019using}. According to results presented in Figure \ref{fig:normal_univariate_density_plots}, both Posterior Bootstrap and BayesBag capture well the uncertainty around the parameter, and under misspecification perform better than standard Bayesian inference. Figure \ref{fig:normal_univariate_w0_boxplot_comparison} shows that our method outperforms BayesBag, which in turn performs better than Weighted Likelihood Bootstrap.  We do not make comparisons with power posteriors, as we did for model \eqref{model:toy_normal_multivariate}. The reason for this is that in this example using power posteriors with parameter $\eta$ set to $\eta^*$ proposed by \cite{lyddon2019general} is equivalent with drawing samples from the Bayesian posterior for the correct model.

	\begin{figure}[!ht]
		\centering
		\includegraphics[width = 0.8\textwidth]{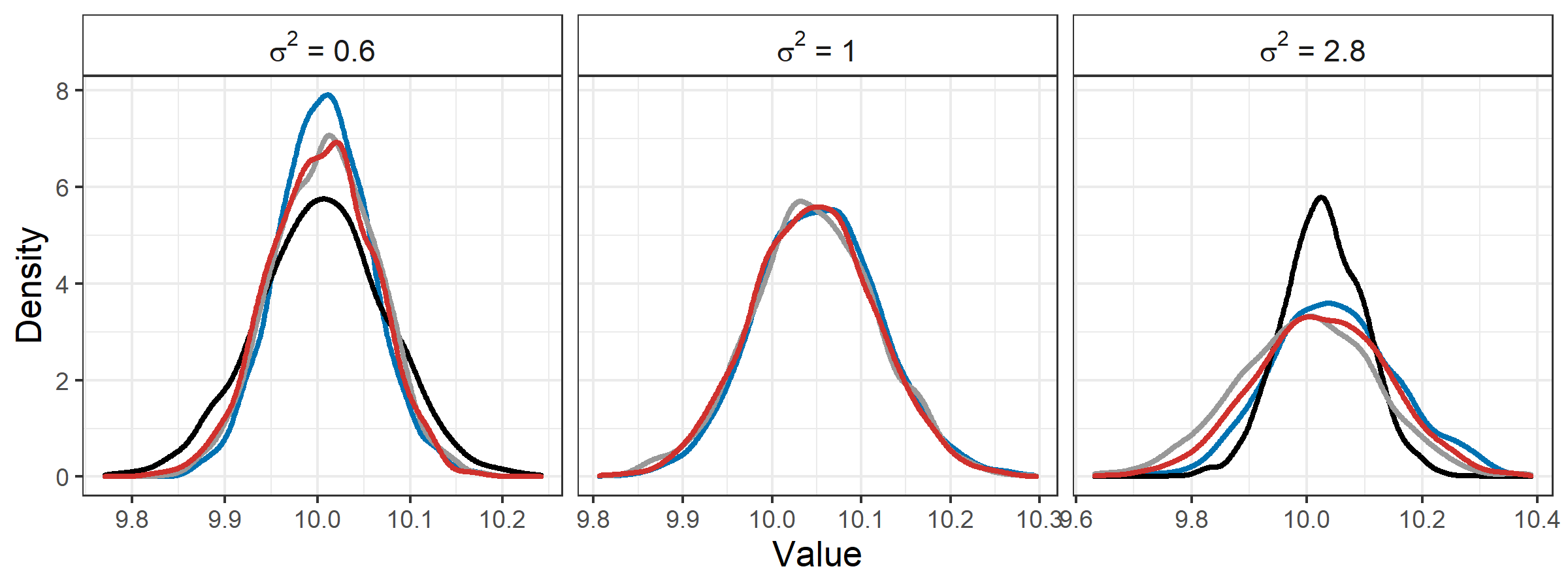}
		\caption{\small  Kernel density estimates of posterior distribution for $\theta$ based on a single synthetic dataset, under standard Bayesian inference for the misspecified model \eqref{model:toy_normal_normal_univariate} (black), Algorithm \ref{alg:pb_general} with $w_0= w_0^*$ (red), BayesBag (dark blue), and standard Bayesian inference for the correct model \eqref{model:toy_normal_normal_univariate_true} (grey). We consider three scenarios $\sigma^2 = 0.6$ (left panel), $\sigma^2 = 1$ (middle panel), and $\sigma^2 = 2.8$ (right panel). }\label{fig:normal_univariate_density_plots}
	\end{figure}

	\begin{figure}[!ht]
		\centering
		\includegraphics[width = 1\textwidth]{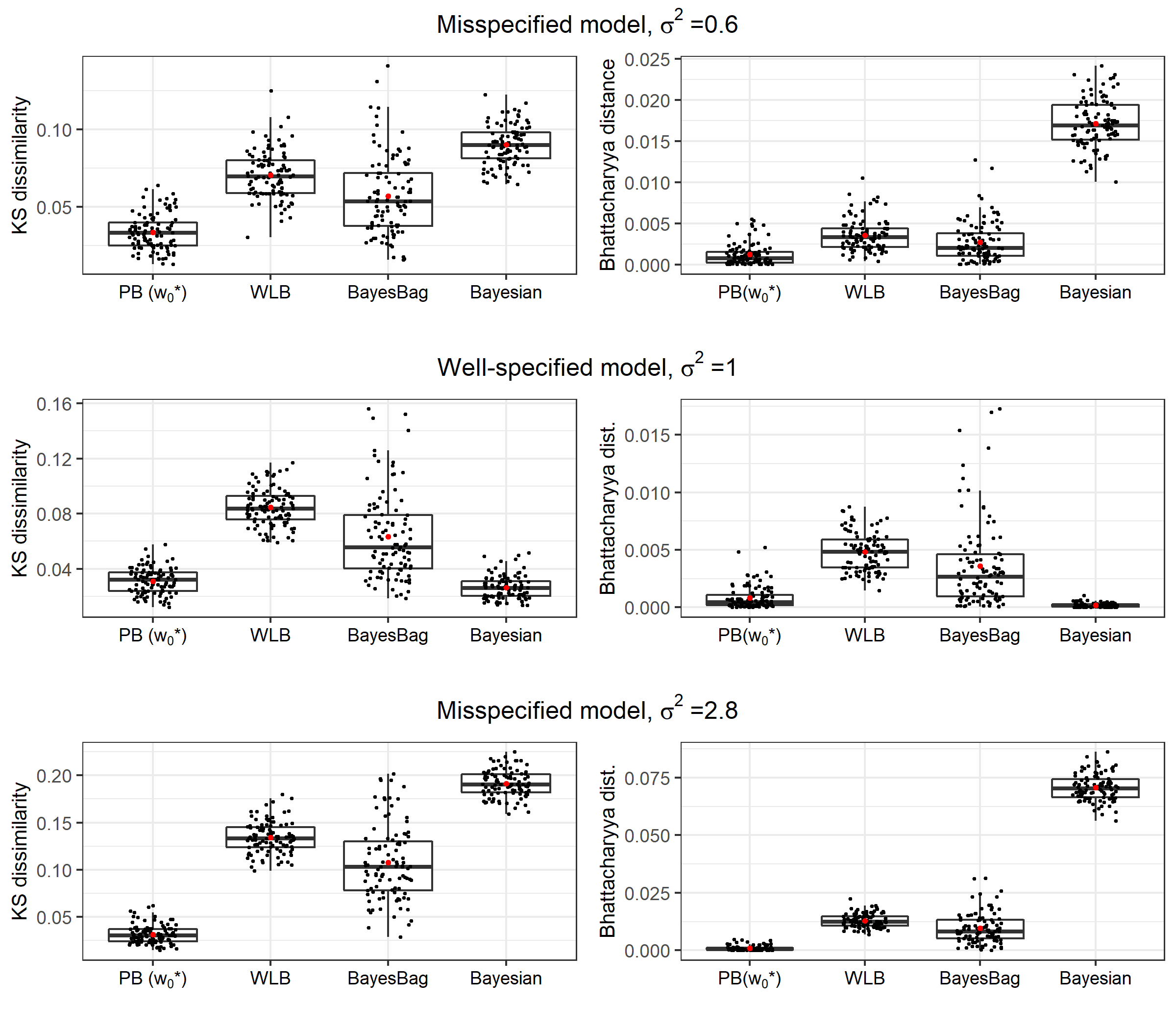}
		\caption{Kolmogorov--Smirnov dissimilarity (left panels) and Bhattacharryya distance (right panels) between samples from the Bayesian posterior for the correct model \eqref{model:toy_normal_univariate_true} and samples obtained using the following four methods:  (from left to right) Algorithm \ref{alg:pb_general} with $w_0 = w_0^*$, denoted by PB($w_0^*$), Weighted Likelihood Bootstrap (WLB), BayesBag and standard Bayesian posterior for model \eqref{model:toy_normal_normal_univariate}. Three scenarios are considered, $\sigma^2 = 0.6$ (top panels), $\sigma^2 =1$ (middle panels), and $\sigma^2 = 2.8$ (bottom panels). Each experiment was repeated on 100 simulated datasets for $N = 2000$.}\label{fig:normal_univariate_w0_boxplot_comparison}
	\end{figure}
	
	\FloatBarrier
	\subsection{Toy example of a hierarchical model}\label{supplementary:toy_example_hierarchical}
	We consider the following hierarchical model with parameters $\lambda$, $\theta_1, \theta_2, \theta_3$:
	\begin{gather}\label{model:toy_gamma_poisson}
	\begin{aligned}
	x_{k,1}, \ldots, x_{k, n_k} & \sim \text{Poisson}(\theta_k), &k=&1, 2, 3, \\ 
	\theta_1, \ldots, \theta_3 & \sim \text{Gamma}(\alpha, \lambda),  & k=&1, 2, 3 \\
	\lambda & \sim \text{Gamma}(\alpha_0, \beta_0). & \\
	\end{aligned}
	\end{gather}
 We consider two scenarios. In the first one data $x_{k,1}, \ldots, x_{k, n_k}$  is generated from the $\text{Poisson}(\theta_k)$ distribution with $\theta_1 =1$, $\theta_2 = 2$ and $\theta_3 =4$, so we have a well-specified model. In the second scenario the correct model is 
	\begin{gather}\label{model:toy_gamma_poisson_correct}
	\begin{aligned}
x_{k,1}, \ldots, x_{k, n_k} & \sim \text{NB}(\theta_k, 0.7), &k&=1, 2, 3, \\ 
\theta_1, \ldots, \theta_3 & \sim \text{Gamma}(\alpha, \lambda),  & k&=1, 2, 3 \\
\lambda & \sim \text{Gamma}(\alpha_0, \beta_0). & 
\end{aligned}
\end{gather}
where $\text{NB}(\theta_k, \omega)$ denotes the negative binomial distribution with mean $\theta_k$ and size $\omega$. Recall that variance of the $\text{NB}(\theta, \omega)$ distribution is $\theta + \theta^2/\omega$, so the negative binomial is often used for modelling overdispersed Poisson data. In the second scenario we generate $x_{k,1}, \ldots, x_{k, n_k}$ from $\text{NB}(\theta_k, 0.7)$, for  $\theta_1 =1$, $\theta_2 = 2$ and $\theta_3 =4$. In both scenarios we set $\alpha_0 = 9, \beta_0 = 3$ and $\alpha =2$, and the number of data points in each class is $n_1 = n_2 = n_3 = 100$. In this misspecified scenario we compare results obtained with different versions of Posterior Bootstrap performed on model \eqref{model:toy_gamma_poisson} with the Bayesian posterior for model \eqref{model:toy_gamma_poisson_correct}.
	
	\begin{figure}[!ht]
	\centering
	\includegraphics[width = 0.8\textwidth]{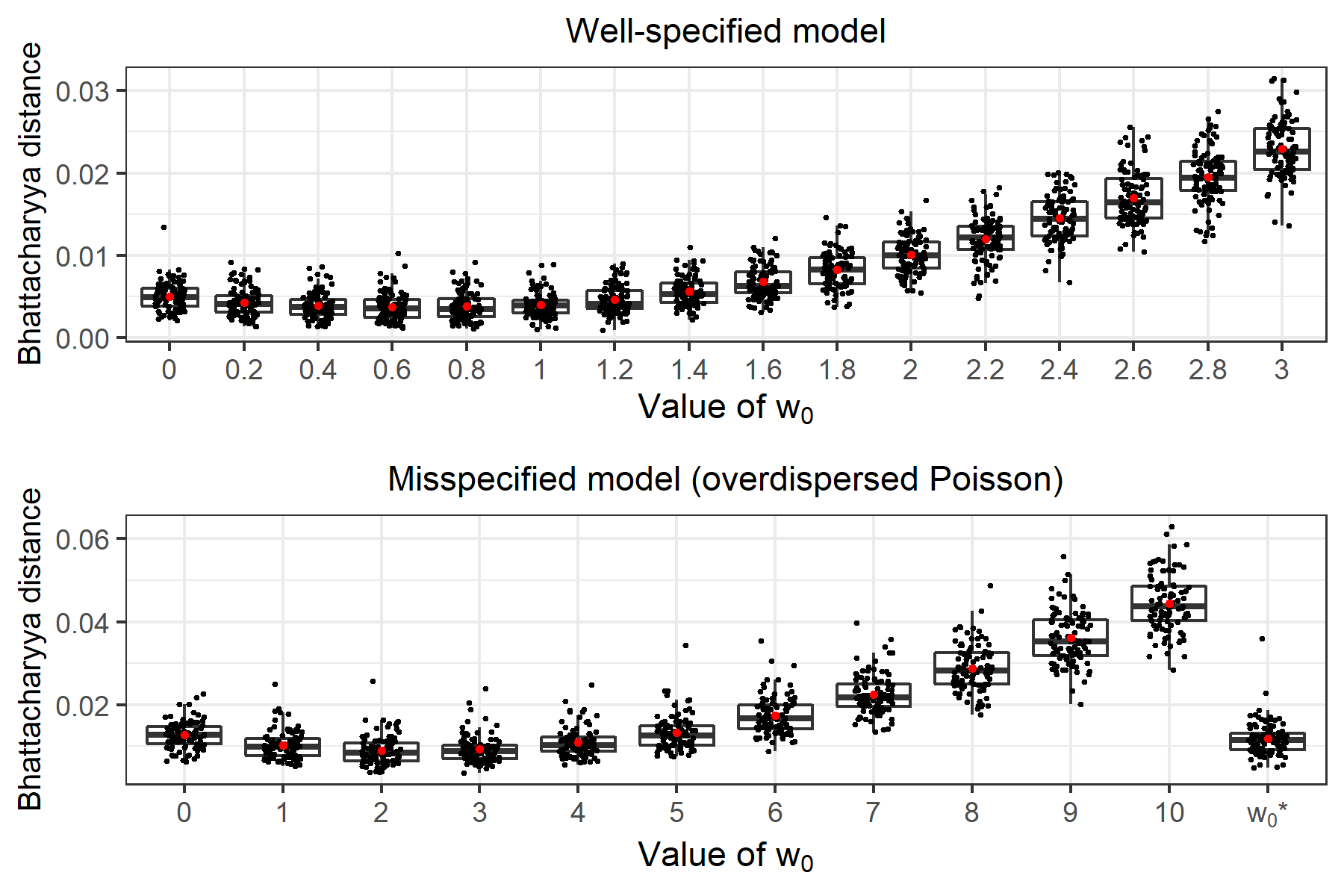}
	\caption{Bhattacharryya distance between samples from the Bayesian posterior for the correct model  and samples obtained using Algorithm \ref{alg:pb_hierarchical_gibbs} for a range of values of $w_0$. We set the same value of $w_0$ for each class, that is, we use $w_{0,[1]} = w_{0,[2]} =w_{0,[3]} = w_0$, for $w_0$ displayed on the $x$-axis of the plot.
		We consider two scenarios: well-specified (top panel), where \eqref{model:toy_gamma_poisson} is the correct model, and misspecified (bottom panel), for which \eqref{model:toy_gamma_poisson_correct} is the correct model. In case of the misspecified model we also include the results for  Algorithm~\ref{alg:pb_hierarchical_gibbs}  with $w_0^* = \left(w_{0,[1]}^*, w_{0,[2]}^*, w_{0,[3]}^*\right)$.
		Each experiment was repeated on 100 simulated datasets for $N = 2000$.}\label{fig:toy_Poisson_Bhat_w0_boxplot}
\end{figure}

We tested Algorithm \ref{alg:pb_hierarchical_gibbs} on this example for a range of values of $w_0$ and the results of these experiments are presented in Figure \ref{fig:toy_Poisson_Bhat_w0_boxplot}. As expected, in the well-specified case it is optimal to set $w_0$ to 1 for each class, whereas in the overdispersed Poisson case setting $w_0$ according to \eqref{eq:setting_w_0} for each class yields results that are closest to the posterior under the correct model.

For a fixed value of $\lambda$ model \eqref{model:toy_gamma_poisson} is conjugate with the effective sample size of the prior equal to $\lambda$. This motivates testing Posterior Bootstrap with pseudo-samples within Algorithm \ref{alg:pb_hierarchical_gibbs}, as mentioned in Remark \ref{remark:pseudosamples}. We expect to obtain best results setting parameter $c$ according to \eqref{eq:setting_c_conjugate_model} as
\begin{align}
c = (c_1, c_2, c_3) =\bar{\lambda} \left(w_{0,[1]}^*, w_{0,[2]}^*, w_{0,[3]}^* \right),
\end{align}
which is confirmed by the results of our experiments shown in Figure \ref{fig:toy_Poisson_Bhat_c_boxplot}.  The pseudo-samples approach appears to have a slight advantage over the penalization-based approach on this example, as shown in Figure \ref{fig:toy_Poisson_comparison_boxplot}.  In particular, both methods perform well when it comes to inference on $\lambda$, as shown in Figure \ref{fig:hierarchical_poisson_density_plot}. Both Posterior Bootstrap approaches give significantly better results than standard Bayesian inference in the misspecified case, as shown in the right panel of Figure \ref{fig:toy_Poisson_comparison_boxplot}.

	\begin{figure}
		\centering
		\includegraphics[width = 0.8\textwidth]{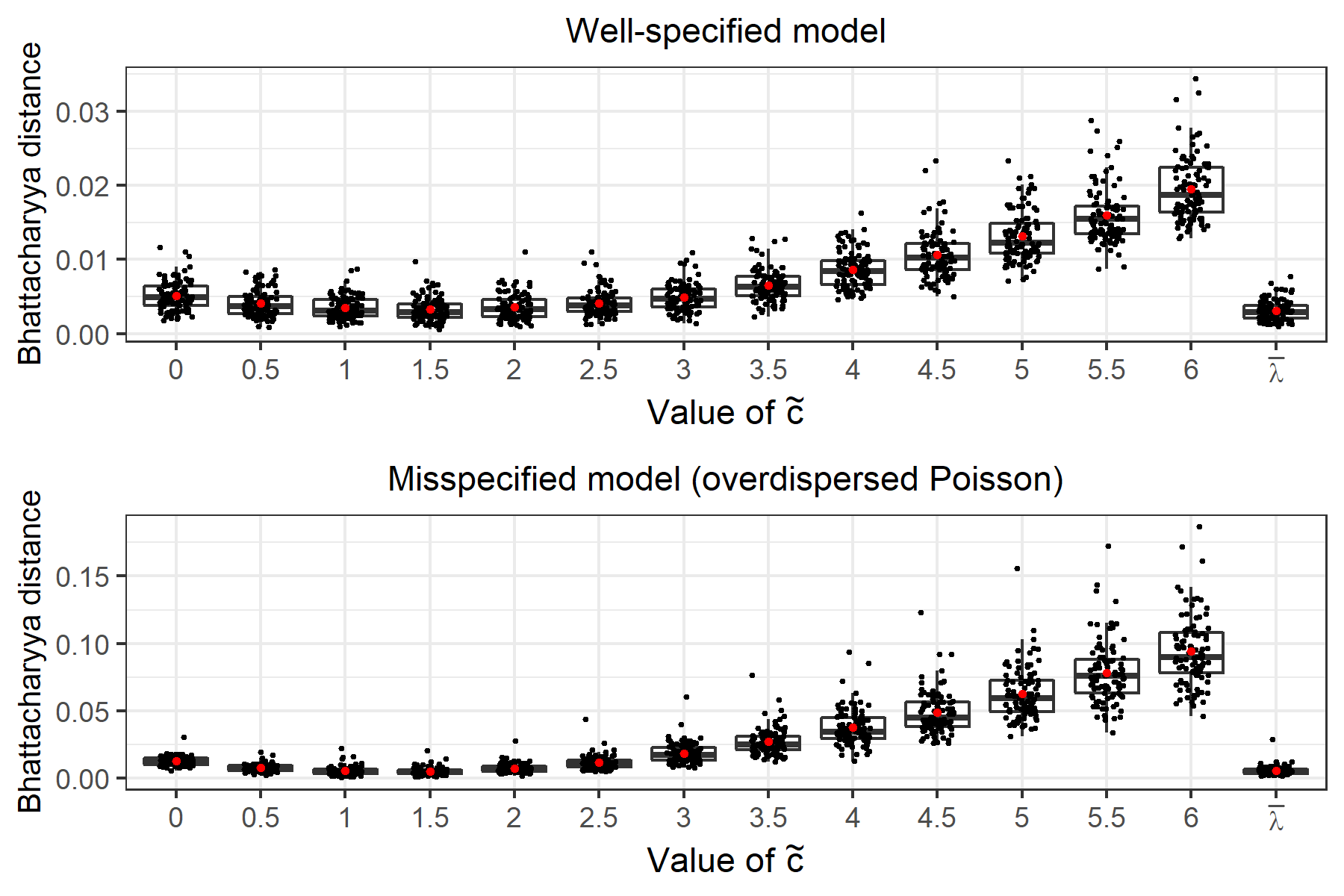}
		\caption{ Bhattacharryya distance between samples from the Bayesian posterior for the correct model  and samples obtained using Posterior Bootstrap with pseudo-samples for a range of values of $\tilde{c}$. We use the following notation: parameter $c$ in Algorithm \ref{alg:pb_pseudosamples_prior} is set to 
			$
			c = (c_1,c_2,c_3) = \tilde{c} \left(w_{0,[1]}^*, w_{0,[2]}^*, w_{0,[3]}^* \right).
			$
			We also present results for $\tilde{c} = \bar{\lambda}$, which is equivalent to setting $c$ according to \eqref{eq:setting_c_conjugate_model}.
			We consider two scenarios: well-specified (top panel), where \eqref{model:toy_gamma_poisson} is the correct model, and misspecified (bottom panel), for which \eqref{model:toy_gamma_poisson_correct} is the correct model. Each experiment was repeated on 100 simulated datasets for $N = 2000$.}\label{fig:toy_Poisson_Bhat_c_boxplot}
	\end{figure}
	\begin{figure}[!ht]
	\centering
	\includegraphics[width = 1\textwidth]{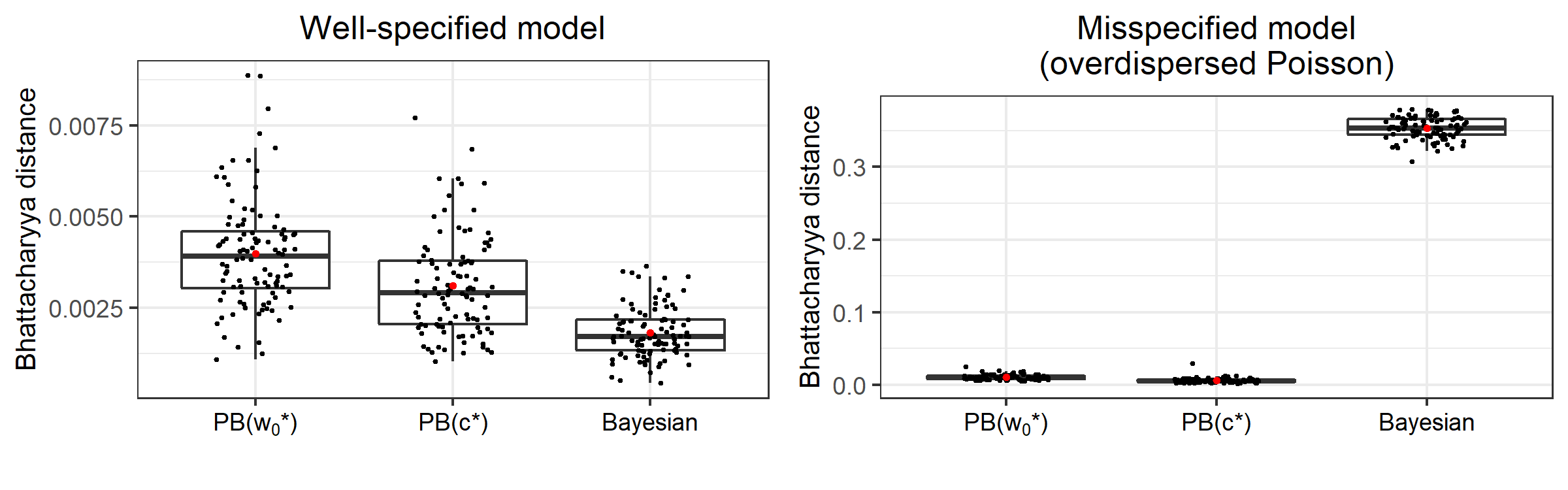}
	\caption{Bhattacharryya distance between samples from the Bayesian posterior for the correct model  and samples obtained using four different methods: (from left to right): Algorithm \ref{alg:pb_hierarchical_gibbs} with $w_0$ set according to \eqref{eq:setting_w_0} for each class, Posterior Bootstrap with pseudo-samples with $c$ set according to \eqref{eq:setting_c_conjugate_model}, and standard Bayesian posterior for model \eqref{model:toy_gamma_poisson}.
		We consider two scenarios: well-specified (left panel), where \eqref{model:toy_gamma_poisson} is the correct model, and misspecified (right panel), for which \eqref{model:toy_gamma_poisson_correct} is the correct model. Each experiment was repeated on 100 simulated datasets for $N = 2000$. Note the difference in the scales of the $y$-axes.}\label{fig:toy_Poisson_comparison_boxplot}
\end{figure}
	
		\begin{figure}
		\centering
		\includegraphics[width = 0.8\textwidth]{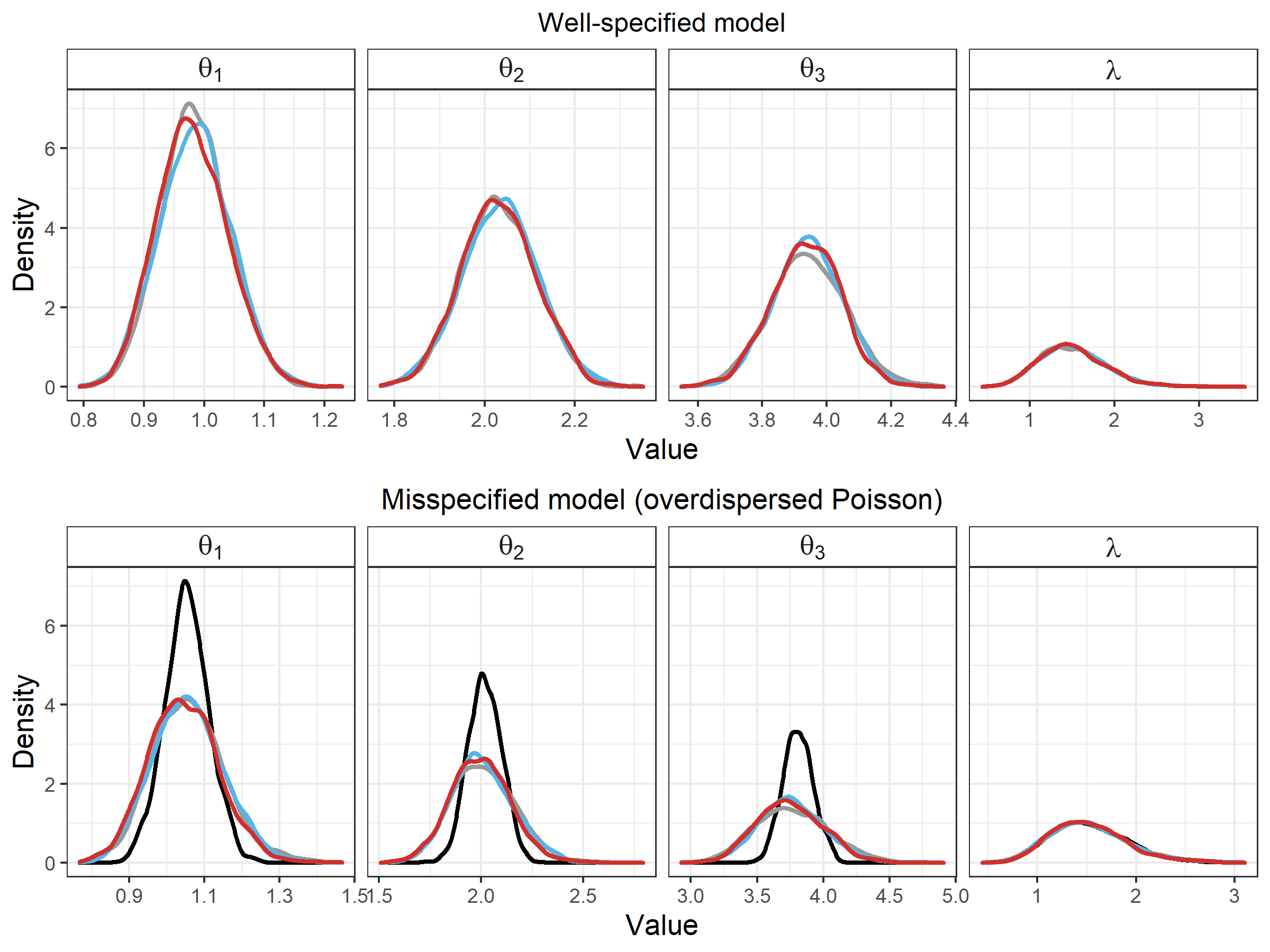}
		\caption{\small  Kernel density estimates of posterior distribution for $\theta_1, \theta_2, \theta_3$ and $\lambda$ based on a single synthetic dataset, under standard Bayesian inference for  model \eqref{model:toy_gamma_poisson} (black), Algorithm \ref{alg:pb_hierarchical_gibbs} with $w_0= w_0^*$ (red), Posterior Bootstrap with pseudo-samples with $c$ set according to \eqref{eq:setting_c_conjugate_model} (light blue). We consider two scenarios: well-specified (top panel) and misspecified (bottom panel). For the misspecified case we include also the kernel density estimate for standard Bayesian inference under the correct model \eqref{model:toy_gamma_poisson_correct} (grey).	
	}\label{fig:hierarchical_poisson_density_plot}\label{fig:toy_poisson_density_plots}
	\end{figure}

\FloatBarrier
	\subsection{Additional information on the Dirichlet allocation model}\label{supplement:dirichlet_allocation}
	We used Algorithm~\ref{alg:dirichlet} to obtain results described in Section~\ref{section:dirichlet_example}. By $g(\theta|\lambda)$ we denote here the density of the Dirichlet distribution with parameter $\lambda$. We let $f(x_{k,i}|\theta)$ denote the probability mass function of a $\text{Multinomial}(\theta, 1)$ distribution, where $x_{k,i} = p$ for $p\in \{1, \ldots, 6\}$ denotes that the reason of death for person $i$ in group $k$ was overdosing drug of type $p$, for $k = 1, \ldots, K$ and $i = 1, \ldots, n_k$.
	\begin{algorithm}[!ht]
		\caption{Posterior Bootstrap for the Dirichlet allocation model \eqref{model:Dirichlet_allocation}}\label{alg:dirichlet}
		\begin{algorithmic}[1]	
			\For {$k = 1, \ldots, K_n$} 
			\State{ Set $\hat{\theta}_k = \left(m_{k1}/n_k, \ldots, m_{k6}/n_k\right)$.}
			\EndFor
			\For {j = 1, \ldots, N} \Comment{\textbf(in parallel) }
			\State{Draw $z_1^{(j)}, \ldots, z_{K_n}^{(j)} \sim\text{Exp}(1)$.}
			\State{Set $\bar{\lambda}^{(j)}  =  \argmax_{\lambda} \left\{ \sum_{k=1}^{K_n} z_k^{(j)} \log g(\hat{\theta}_k| \lambda) + w_{g}^T \log p\left(\lambda\right)\right\}$.}
			\State{Set $c = \bar{\lambda}^{(j)}_1 + \cdots + \bar{\lambda}^{(j)}_6$.}
			\For {$k = 1, \ldots, K_n$} 
			\State{Draw $\bar{\theta}_{k,1}^{(j)}, \ldots, \bar{\theta}_{k,T}^{(j)}\sim \text{Dirichlet} (\bar{\lambda})$.}
			\State{Generate prior pseudo-samples $x_{k,n_k+i}^{(j)} \sim \text{Multinomial}\left(\bar{\theta}_{k,i}^{(j)}, 1\right)$ for $i =1, \ldots, T$.}
			\State{Draw $w_1^{(j)}, \ldots, w_{n_k}^{(j)} \sim\text{Exp}(1)$ and $w_{n_k+1}^{(j)}, \ldots, w_{n_k+T}^{(j)} \sim\text{Exp}(T/c)$.}
			\State{ Set $\tilde{\theta}_k^{(j)} = \argmax_{\theta_k} \left\{ \sum_{i=1}^{n_k} w_i^{(j)} \log f(x_{k,i}| \theta_k) + \sum_{i=n_k +1}^{n_k+T} w_i^{(j)} \log f\left(x_{k,i}^{(j)}| \theta_k\right)\right\}$.}
			\EndFor
			\State{Draw $v_1^{(j)}, \ldots, v_{K_n}^{(j)} \sim\text{Exp}(1)$.}
			\State{Set $\tilde{\lambda}^{(j)} =  \argmax_{\lambda} \left\{ \sum_{k=1}^{K_n} v_k^{(j)} \log g(\tilde{\theta}_k^{(j)}| \lambda) + w_{g}^T \log p\left(\lambda\right) \right\}$.}
			\EndFor
			\State \Return {Samples $\Big\{ \left(\tilde{\lambda}^{(j)}, \tilde{\theta}_1^{(j)}, \ldots, \tilde{\theta}_K^{(j)}\right) \Big\}_{j=1}^N$. }
		\end{algorithmic}
	\end{algorithm}
	
	We additionally tested Algorithm~\ref{alg:dirichlet} on a synthetic dataset generated from model \eqref{model:Dirichlet_allocation}with $\lambda = (12,12,12,10,10,10)$, $K=300$ and $n_k = 1000$ for $k = 1, \ldots,K$.  Figure~\ref{fig:synthetic_dirichlet_density_plot} shows that Bayesian inference and Algorithm~\ref{alg:dirichlet} return very similar results, which suggests the difference between the two methods in Figure \ref{fig:opioid_density_plot} is caused by model misspecification. In the same figure we plotted kernel density estimates of $\bar{\lambda}$. We notice that despite the fact that asymptotically $\bar{\lambda}$ and $\tilde{\lambda}$ have the same distributions, for this example $\tilde{\lambda}$ performs much better. We also observe that, as expected, using more concentrated priors causes a shift towards 0 in posterior distributions.
	
		\begin{figure}[!ht]
		\centering
		\includegraphics[width = 1\textwidth]{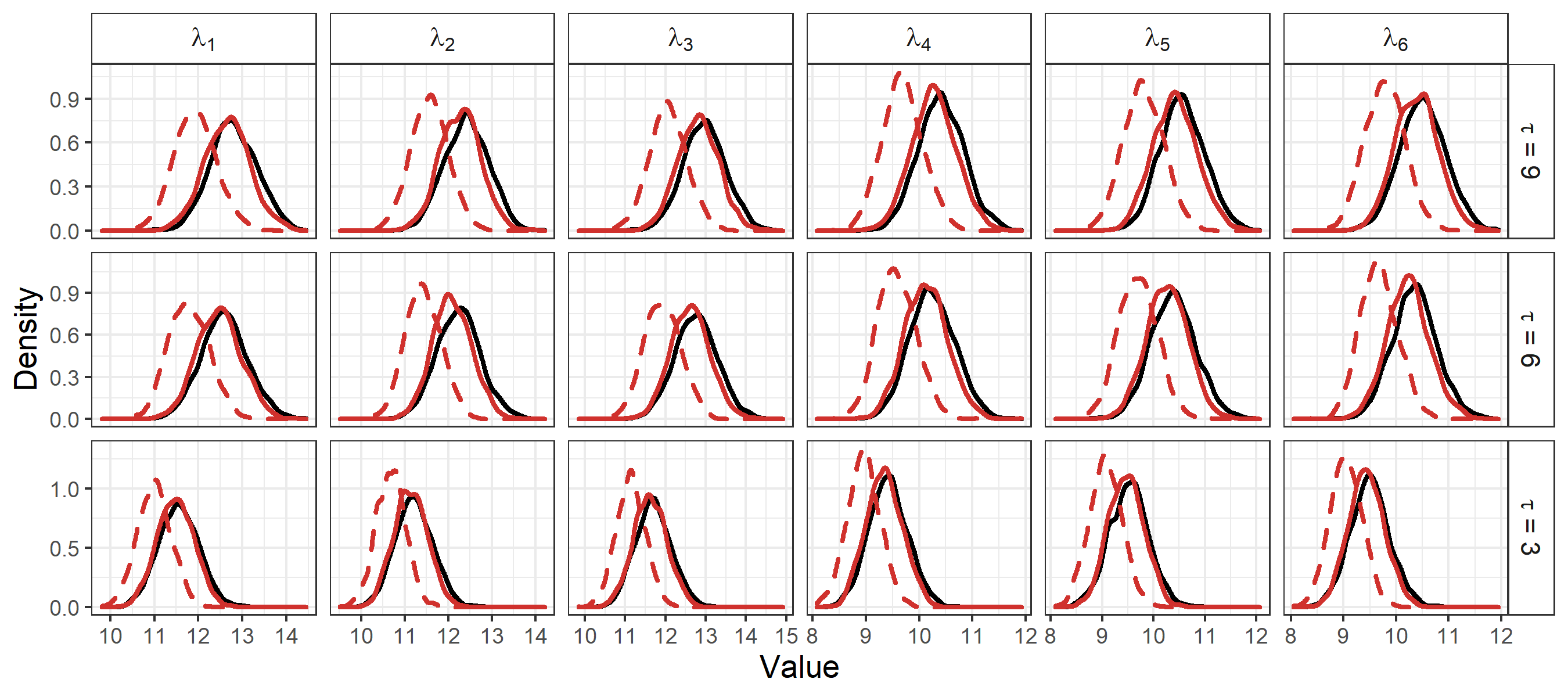}
		\caption{\small Kernel density estimates of posterior distribution on concentration parameter $\lambda$ for a synthetic dataset generated from model \eqref{model:Dirichlet_allocation} with $\lambda = (12,12,12,10,10,10)$,  $K=300$ and $n_k = 1000$ for $k = 1, \ldots,K$. We present results for Bayesian inference (black), $\tilde{\lambda}$ drawn from Algorithm \ref{alg:dirichlet} (red solid line) and $\bar{\lambda}$ drawn from Algorithm \ref{alg:dirichlet} (red dashed line).}\label{fig:synthetic_dirichlet_density_plot}
	\end{figure}

	\FloatBarrier
	\bibliographystyle{plainnat}
	\bibliography{../references_pb}
\end{document}